\documentclass[11pt, twoside, a4paper]{article}
\usepackage{bm}
\usepackage{tabularx}
\usepackage{comment}

\usepackage[hyphens]{url}
\usepackage[hidelinks]{hyperref}
\hypersetup{breaklinks=true}
\usepackage{soul}
\usepackage{geometry}

\AtBeginDocument{
\addtolength{\abovedisplayskip}{-0.9ex}
\addtolength{\abovedisplayshortskip}{-0.9ex}
\addtolength{\belowdisplayskip}{-0.9ex}
\addtolength{\belowdisplayshortskip}{-0.9ex}
}
\newcommand{\wpinf}{w_{\scalebox{0.6}{$\infty$}}}
\newcommand{\einf}{e^{\wpinf}}
\newcommand{\bVisl}{\l(V_{\scalebox{0.5}{\text{SL}}}^{\mysup}\r)}
\newcommand{\errorf}{\scalebox{0.9}{$\mathcal{E}_{f}$}}
\newcommand{\errorp}{\scalebox{0.9}{$\mathcal{E}_{o}$}}
\newcommand{\errorg}{\scalebox{0.9}{$\mathcal{E}_{\hat{g}}$}}
\newcommand{\errorb}{\scalebox{0.9}{$\mathcal{E}_{b}$}}

\newcommand{\x}{{\bf{x}}}

\usepackage{pgfplots}
\makeatletter
\def\thm@space@setup{\thm@preskip=3pt
\thm@postskip=3pt}
\makeatother
\usepackage{titlesec}
\titlespacing*{\section}
  {0pt}{0.7\baselineskip}{0.2\baselineskip}
\titlespacing*{\subsection}
  {0pt}{0.7\baselineskip}{0.2\baselineskip}

\makeatletter
\newcommand{\subalign}[1]{%
  \vcenter{%
    \Let@ \restore@math@cr \default@tag
    \baselineskip\fontdimen10 \scriptfont\tw@
    \advance\baselineskip\fontdimen12 \scriptfont\tw@
    \lineskip\thr@@\fontdimen8 \scriptfont\thr@@
    \lineskiplimit\lineskip
    \ialign{\hfil$\m@th\scriptstyle##$&$\m@th\scriptstyle{}##$\crcr
      #1\crcr
    }%
  }
}
\makeatother
\usepackage{graphicx}
\usepackage{amsmath, amsthm, amssymb, amsfonts, latexsym}
\usepackage{array}
\usepackage{caption}
\usepackage{color}

\usepackage{subcaption}

\usepackage{epsfig}
\usepackage{fancybox}
\usepackage{array}
\usepackage{shadow}
\usepackage{empheq}

\usepackage{mathrsfs}
\usepackage{longtable}
\usepackage{calc}
\usepackage{multirow}
\usepackage{hhline}
\usepackage{ifthen}
\usepackage{algorithmic}
\usepackage{algorithm}
\usepackage{color}

\newcommand{\blue}[1]{{\color{black}{#1}}}
\newcommand{\red}[1]{{\color{black}{#1}}}

\newcommand{\numPDEblue}[1]{{\color{black}{#1}}}

\usepackage[running, mathlines]{lineno}

\usepackage{wrapfig}

\newcommand{\EQ}{\begin{equation}}
\newcommand{\EN}{\end{equation}}
\newcommand{\EQS}{\begin{equation*}}
\newcommand{\ENS}{\end{equation*}}
\newcommand{\EQA}{\begin{eqnarray}}
\newcommand{\ENA}{\end{eqnarray}}
\newcommand{\EQAS}{\begin{eqnarray*}}
\newcommand{\ENAS}{\end{eqnarray*}}
\newcommand{\ds}{\displaystyle}

\newcommand{\vl}{\left(v^{\scalebox{0.5}{(1)}}\right)}
\newcommand{\vn}{\left(v^{\scalebox{0.5}{(2)}}\right)}

 \newcommand{\vh}{\hat{v}}

\newcommand{\bw}{\breve{w}}
\newcommand{\vsl}{v_{\scalebox{0.5}{\text{SL}}}}

\newcommand{\vt}{\tilde{v}}

\newcommand{\bvsl}{\l(v_{\scalebox{0.5}{\text{SL}}}\r)}
\newcommand{\bvisl}{\l(v_{\scalebox{0.5}{\text{SL}}}^{\mysup}\r)}

\newcommand{\bvlsl}{\l(v_{\scalebox{0.5}{\text{SL}}}^{\myloc}\r)}
\newcommand{\bvnsl}{\l(v_{\scalebox{0.5}{\text{SL}}}^{\mynlc}\r)}
\newcommand{\mysup}{\scalebox{0.6}{$(i)$}}
\newcommand{\myloc}{\scalebox{0.6}{$(1)$}}

\newcommand{\mynlc}{\scalebox{0.6}{$(2)$}}

\newcommand{\hvisl}{\hat{v}_{\scalebox{0.5}{\text{SL}}}^{\mysup}}

\newcommand{\vhsl}{\hat{v}_{\scalebox{0.5}{\text{SL}}}}

\newcommand{\br}{\breve{r}}
\newcommand{\z}{p}
\newcommand{\pb}{\psl}
\newcommand{\pt}{\tilde{\phi}}

\newcommand{\psl}{\phi_{\scalebox{0.5}{\text{SL}}}}

\newcommand{\plsl}{\phi_{\scalebox{0.5}{\text{SL}}}^{\myloc}}
\newcommand{\pnsl}{\phi_{\scalebox{0.5}{\text{SL}}}^{\mynlc}}
\newcommand{\pisl}{\phi_{\scalebox{0.5}{\text{SL}}}^{\mysup}}

\newcommand{\vpsl}{\varphi_{\scalebox{0.5}{\text{SL}}}}

\newcommand{\psisl}{\psi_{\scalebox{0.5}{\text{SL}}}}

\usepackage{scalerel}[2014/03/10]
\usepackage{stackengine}
\usepackage[titletoc,title]{appendix}

\def\mytilde[#1]{$%
  \reallywidetilde{#1}\,
  \scriptstyle\reallywidetilde{#1}\,
  \scriptscriptstyle\reallywidetilde{#1}
$\par}

\newcommand{\pl}{(\phi_{\scalebox{0.6}{$loc$}})}
\newcommand{\pn}{(\phi_{\scalebox{0.6}{$nlc$}})}

\newcommand{\sigz}{\sigma_{\scalebox{0.5}{$Z$}}}
\newcommand{\sigr}{\sigma_{\scalebox{0.5}{$R$}}}
\newcommand{\Wz}{W_{\scalebox{0.5}{$Z$}}}
\newcommand{\Wr}{W_{\scalebox{0.5}{$R$}}}
\newcommand{\winf}{w_{\scalebox{0.6}{-$\infty$}}}

\newcommand{\Oinf}{\Omega^{\scalebox{0.5}{$\infty$}}}
\newcommand{\myin}{\scalebox{0.7}{\text{in}}}
\newcommand{\myot}{\scalebox{0.6}{\text{c}}}

\newcommand{\mymax}{{{\scalebox{0.6}{$\max$}}}}
\newcommand{\myup}{\scalebox{0.5}{\text{$U$}}}
\newcommand{\mydn}{\scalebox{0.5}{\text{$L$}}}

\newcommand{\dtau}{\Delta \tau}

\newcommand{\taus}{\tau_{m}}

\newcommand{\err}{\error({\bf{x}}_{n,k,j}^{m}, h)}
\newcommand{\errm}{\error({\bf{x}}_{n,k,j}^{m}, h)}

\newcommand{\myerrmm}[1]{\error_{#1}({\bf{x}}_{n,k,j}^{m+1}, h)}
\newcommand{\myerrm}[1]{\error_{#1}({\bf{x}}_{n,k,j}^{m}, h)}

\newcommand{\G}{\mathcal{G}(\Oinf)}
\newcommand{\C}[1]{{\mathcal{C}^{\scalebox{0.5}{$\infty$}}{(#1)}}}

\newcommand{\error}{\scalebox{0.9}{$\mathcal{E}$}}

\newcommand{\D}{\mathbf{D}}
\newcommand{\Ld}{\mathcal{L}_{\scalebox{0.7}{$d$}}}
\newcommand{\Ls}{\mathcal{L}_{\scalebox{0.7}{$s$}}}
\newcommand{\Lg}{\mathcal{L}_{\scalebox{0.7}{$g$}}}

\newcommand{\mysum}{\ds\sideset{}{^{\boldsymbol{*}}}\sum}

\newcommand{\R}{C}
\newcommand{\N}{\mathbb{N}}
\newcommand{\Nl}{\mathbb{N}^{\myot}_{\scalebox{0.6}{\text{$\min$}}}}
\newcommand{\Nr}{\mathbb{N}^{\myot}_{\scalebox{0.6}{\text{$\max$}}}}
\newcommand{\ND}{\mathbb{N}^{\dagger}}
\newcommand{\Nc}{\mathbb{N}^{\myot}}
\newcommand{\K}{\mathbb{K}}
\newcommand{\Kc}{\mathbb{K}^{\myot}}
\newcommand{\KD}{\mathbb{K}^{\dagger}}
\newcommand{\J}{\mathbb{J}}
\newcommand{\M}{\mathbb{M}}
\newcommand{\Z}{\mathbb{Z}}
\newcommand{\Na}{\mathbb{N}^{\alpha}}
\newcommand{\Ka}{\mathbb{K}^{\alpha}}
\newcommand{\mychi}{\scalebox{1.2}{$\chi$}}

\numberwithin{equation}{section}
\numberwithin{table}{section}
\numberwithin{figure}{section}

\newtheorem{definition}{Definition}
\newtheorem{theorem}{Theorem}
\newtheorem{lemma}{Lemma}
\newtheorem{remark}{Remark}
\newtheorem{assumption}{Assumption}

\newtheorem{proposition}{Proposition}

\numberwithin{definition}{section}
\numberwithin{theorem}{section}
\numberwithin{lemma}{section}
\numberwithin{remark}{section}
\numberwithin{assumption}{section}
\numberwithin{condition}{section}
\numberwithin{property}{section}
\numberwithin{proposition}{section}
\numberwithin{corollary}{section}
\numberwithin{algorithm}{section}

\def\r{\right}
\def\l{\left}

\newsavebox{\savepar}

\makeatletter
\renewcommand*\env@matrix[1][c]{\hskip -\arraycolsep
  \let\@ifnextchar\new@ifnextchar
  \array{*\c@MaxMatrixCols #1}}
\makeatother

\pgfplotsset{compat=1.18}

\begin{document}


\title{
A semi-Lagrangian $\epsilon$-monotone Fourier method
for continuous withdrawal GMWBs under jump-diffusion with stochastic interest rate
}

\author{
Yaowen Lu \thanks{School of Mathematics and Physics, The University of Queensland, St Lucia, Brisbane 4072, Australia,
\texttt{yaowen.lu@uq.edu.au}
}
\and
Duy-Minh Dang\thanks{School of Mathematics and Physics, The University of Queensland, St Lucia, Brisbane 4072, Australia,
\texttt{duyminh.dang@uq.edu.au}
}
}
\date{July 28, 2023}
\maketitle

\begin{abstract}
\noindent
We develop an efficient pricing approach for guaranteed minimum withdrawal benefits (GMWBs) with continuous withdrawals under a realistic modeling setting with jump-diffusions and stochastic interest rate.
Utilizing an impulse stochastic control framework, we formulate the no-arbitrage GMWB pricing problem
as a time-dependent Hamilton-Jacobi-Bellman (HJB) Quasi-Variational Inequality (QVI) having three spatial dimensions with cross derivative terms.
Through a novel numerical approach built upon a combination of a semi-Lagrangian method and the Green's function of an associated linear partial integro-differential equation, we develop an $\epsilon$-monotone Fourier pricing method, where $\epsilon > 0$ is a monotonicity tolerance. Together with a provable strong comparison result for the HJB-QVI, we mathematically demonstrate convergence of the proposed scheme to the viscosity solution of the HJB-QVI as $\epsilon \to 0$.  We present a comprehensive study of the impact of simultaneously considering jumps in the sub-account process and stochastic interest rate on the no-arbitrage prices and fair insurance fees of GMWBs, as well as on the holder's optimal withdrawal behaviors.

\vspace{9pt}

\vspace{9pt}

\end{abstract}

\section{Introduction}
Variable annuities are a class of insurance products that provide the holder
with particular guaranteed stream of income without requiring him/her
to sacrifice full control over the funds invested, and hence,
allowing the holder to enjoy potentially favorable market conditions.
Therefore, these products are particularly popular among
investors who need to manage their own spending plans,
especially among retirees, considering the on-going rapid word-wide trend of replacing
defined benefit pension plans by defined contribution ones.
The current era of increased market volatility and growing inflation
has significantly boosted annuity sales.
In some countries, such as the US, annuity sales are at highest levels since the 2007-2008 Global Financial Crisis.
Specifically, the US annuity market in 2021 was valued at US\$231.63 billion,
and the market is expected to grow at a  compound annual growth rate of $4.7\%$
during the forecast period of 2022-2026, reaching US\$298.70 billion by 2026 \cite{MR2022}.

To attract investors, variable annuities are often incorporated with additional features,
among which Guaranteed Minimum Withdrawal Benefits (GMWBs) are popular.
Since first introduced in the early 2000's, GBMWs have captured great attention from both industry and academia alike, as evidenced by a substantial and growing body of literature; see \cite{Milevsky2006, chen08a, chen_2008_a, Dai,
Bacinello2016, huang:2010, huang:2010b, Huang2014, Kwok2019, Molent2016a, Gudkov2018, Molent2020, Azimzadeh2015, wang2015understanding, alonso-garcca_wood_ziveyi_2018,
DangOu2021, huang:2010a, online}, among many other publications.

In its simplest form, a GMWB is a long-dated contract, with maturity of 10 years or more,
between the policy holder (e.g.\ a retiree) and the insurer (e.g.\ an insurance company),
according to which the holder makes an up-front payment, i.e.\ the premium,
into a (personal) sub-account for investment in risky assets.
In return, by means of a guarantee account,  the insurer is stipulated to provide the holder with a stream of guaranteed cash withdrawals whose amounts (and possibly timing) are  to be determined by the holder, all of which cumulatively sum up to \emph{at least} the premium, regardless of the performance of the risky investment.
The holder may also withdraw more than the specified amount, subject to certain penalties and conditions.
Upon contract expiry, the holder can convert the remaining investment in the risky assets to cash,
and withdraw this amount. For protection from the downside in a GMWB, the insurer typically charges the holder
an insurance fee by deducting an ongoing fraction of the risky investment as opposed to an up-front one-off fee.
Underpricing typically results in undercalculated insurance fee, which adversely affects the insurer's risk management, potentially impacting the long-term sustainability of the market. The reader is referred to, for example,
\cite{choi2018valuation, Milevsky2006, chen_2008_a, Dai, Bauer}, for discussions in relation to GMWB underpricing in practice and its potential consequences.

Guaranteed Minimum Withdrawal Benefits are studied under two withdrawal scenarios, namely discrete and continuous.
It is reported in the literature that no-arbitrage prices and fair insurance fees of GMWBs,
as well as the holder's optimal withdrawal behaviors are highly sensitive to modeling assumptions and parameters,
in particular, jumps in the sub-account's balance process \cite{chen_2008_a, Benhamou2009, Kling2011, online}.
Under a discrete withdrawal scenario, fair prices and insurance fees are found to be
remarkably sensitive to interest rates, in particular, in the case of (instantaneous) short rate dynamics, such as the Vasicek model \cite{Peng2009, Pavel2017}, the Hull-White \cite{Molent2016a, Donnelly, Molent2019, Mackay2018}, and the the Cox-Ingersoll-Ross model \cite{Bacinello2011, Gudkov2018}.
Substantial impact of short rate dynamics on the holder's optimal withdrawal behavior
is recently reported in \cite{Molent2020}.
We highlight that the combined effects of jumps and stochastic interest rate
in the context of GMWBs have not been previously studied in the literature.

Numerical methods for GMWBs in a continuous withdrawal scenario is
studied through a stochastic optimal control framework.  In this withdrawal scenario, the pricing problem
can be formulated  using either impulse control or singular control. This typically results in a Hamilton-Jacobi-Bellman Quasi-Variational Inequality (HJB-QVI)
of at least two spatial dimensions, namely the balances of the sub- and guarantee accounts, which must be solved
numerically. Convergence to viscosity solutions forms the main challenge in the development of numerical methods for HJB equations. This is typically built upon the convergence
framework established by  Barles and Souganidis in
\cite{barles-souganidis:1991}; also, see \cite{crandall-ishii-lions:1992, MNIF06, Jakobsen2006, Barles2008, Seydel2009, Berestycki2014, barles-burdeau:1995}
for relevant discussions.
Specifically, provided that a strong comparison result holds, convergence to viscosity solution is ensured if
numerical methods are (i)~monotone (in the viscosity sense), (ii)~stable, and (iii)~consistent.
When a finite difference method is used, monotonicity is ensured by
a positive coefficient discretization method \cite{pooley2003, wang08, MaForsyth2015, Forsyth08}.
The reader is referred to \cite{Dai, huang:2010b, huang:2010a, huang:2010, Milevsky2006, Bauer}
and \cite{chen08a, chen_2008_a, Azimzadeh2015, online} for an analysis of singular control and  impulse control
formulations, respectively.  For GMWB contracts, impulse control is more convenient than singular control in handling complex contract features,
such as is the reset provision\cite{Dai, Milevsky2006, DangOu2021, alonso-garcca_wood_ziveyi_2018, Gudkov2018, wang2015understanding}.\footnote{Generally speaking, the impulse control approach is suitable for many complex situations
in stochastic optimal control {\numPDEblue{\cite{OksendalSulemBook3,  PvSDangForsyth2018_MQV, PvSDangForsyth2019_Distributions, PvSDangForsyth2019_Robust, PvSDangForsyth2018_TCMV, Korn1999, GuanLian2014, Baccarin2009, Masahiko2008, AsriMazid2020, BelakChristensenSeifried2017, DangForsyth2014}}}.}

In contrast to continuous withdrawals, a discrete withdrawal scenario is relatively much simpler to  tackle. Specifically, between fixed withdrawal (intervention) times, the pricing of GMWB contracts typically involves solving
an either (i) associated linear Partial (Integro)-Differential Equation (P(I)DE) using finite differences \cite{chen08a, Dai, online}, or (ii) an expectation problem using numerical integration \cite{Pavel2015, Pavel2016, Bauer, alonso-garcca_wood_ziveyi_2018, ignatieva_song_ziveyi_2018, Huang2018} or regression-type Monte Carlo
\cite{Bacinello2011, Huang2015}.
Across withdrawal times, an optimization problem needs to be solved to determine the optimal withdrawal amount,
by which the balance of the guarantee account is then adjusted accordingly.
We note that existing numerical integration or regression-type Monte Carlo
are typically not suitable to tackle continuous withdrawals.

\newpage
In light of the current era of wildly fluctuating interest rates and economic turbulence,
it is of enormous importance to apply realistic modelling for popular pension-related products.
In addition, it is also equally important to develop mathematically reliable numerical methods for those products, enabling realistic and useful conclusions to be drawn from the numerical results.
For GMWBs, it is highly desirable to simultaneously incorporate jumps (in the sub-account balance)
and stochastic interest rate dynamics.
Although in practice, only discrete withdrawals are possible,
through no-arbitrage arguments, it is arguable that the prices and insurance fees in the associated continuous withdrawal scenario can serve as worst-case bounds for the respective values in a discrete withdrawal one,
which are important for risk-management purposes.  

Nonetheless, the continuous withdrawal scenario brings about significant mathematical challenges.
As noted earlier, for GMWBs under a low-dimensional model, existing numerical integration and regression-type Monte Carlo methods are computationally expensive. With respect to the PIDE approach, due to the short rate factor,
the no-arbitrage pricing of GMWBs gives  rise to a HJB-QVI of three spatial dimensions with cross derivative terms,
which is very challenging to solve efficiently numerically. In particular, while finite difference methods
can be used to solve this HJB-QVI,  due to cross derivative terms, to ensure monotonicity
through a positive coefficient discretization method, a wide-stencil method based on a local coordinate rotation is needed. However, this is very computationally expensive \cite{MaForsyth2015, DJ12}.

In general, Fourier-based methods, if applicable, offer several important
advantages over finite differences, such as no timestepping error between intervention times,
and the capability of straightforward handling of realistic underlying dynamics,
such as jump diffusion and regime-switching. \red{In particular, the well-known Fourier cosine series expansion method \cite{Fang2008, Ruijter2013} can achieve high order convergence for piecewise smooth problems.  However, optimal control problems are often non-smooth, and hence high order convergence cannot be expected. Convergence issues, especially montonicity considerations are of primary importance.}
A novel Fourier-based method is introduced in our paper \cite{online}
for an impulse control formulation of the GMWB pricing problem
in which the sub-account's balance process follows jump-diffusion dynamics with a constant interest rate.
Central to the method is a combination of (i) the Green's function of an associated multi-dimensional PIDE
and (ii) an $\epsilon$-monotone Fourier method to approximate a pricing convolution integral
through a known closed-form expression of the Fourier transform of the Green function.
Here, the monotonicity of the method is achieved within an $\epsilon$ tolerance, where $\epsilon > 0$,
as opposed to strictly monotone. In this work, a Barles-Souganidis-type analysis in \cite{barles-souganidis:1991} is utilized to rigorously prove the convergence of the scheme the unique viscosity solution of the HJB-QVI
as the discretization parameter and the monotonicity tolerance $\epsilon$ approach zero.
Nonetheless, for the case of jump-diffusion dynamics having a non-trivial correlation
with the short rate, a closed-form expression of the Fourier transform of
the Green function is not know to exist.
Therefore, the approach in \cite{online}, while promising, is not directly applicable.
This mathematical and computational challenge of continuous withdrawals
forms another motivation for our work.

The objective of the paper is (i) to develop a provably convergent and computationally efficient PDE method for
the no-arbitrage GMWB pricing problem with continuous withdrawals under realistic modeling assumptions,
namely jumps and stochastic interest rate, and (ii) to study the combined impacts of these modelling assumptions on  the no-arbitrage prices and fair insurance fees of GMWBs, as well as the holder's optimal withdrawal behaviors.
For clarity of presentation, we focus on the GMWB pricing problem with basic contract features.
We emphasize that we do not to advocate for a specific jump-diffusion and/or stochastic interest rate
model, but rather, we aim to study the impact of realistic modeling on GWMB.
In particular, to model stochastic interest rate, we use the Vasicek short rate dynamics \cite{vasicek1977equilibrium}.
Due to a Gaussian nature, the Vasicek short rate dynamics are often criticized for allowing negative interest rates, which is considered a highly undesirable, and perhaps, also highly improbable, scenario for any economy.
However, in recent times, it has become evident that negative interest rates are employed as
a monetary policy tool by central banks, such as the European Central Bank, against extreme financial crises.
For example, see \cite{jobst2016negative, demiralp2021negative, lopez2020have} and references therein.

The main contribution of the paper are as follows.
\begin{itemize}

\item
We propose a comprehensive and systematic impulse control formulation and pricing approach for GMWBs
when the sub-account process follows a jump-diffusion process \cite{merton1975, kou01}
with the Vasicek short rate dynamics \cite{vasicek1977equilibrium}.
\begin{itemize}
\item We derive and define the pricing problem in a form of an HJB-QVI with three spatial dimensions posed
on an infinite definition domain with appropriate boundary conditions. Through a novel approach built upon a combination of a semi-Lagrangian method and the Green's function of an associated
PIDE, we obtain a properly truncated computational domain for which loss of information in the boundary
is controllably negligible.

\item
Starting from a discrete withdrawal scenario, we develop a semi-Lagrangian $\epsilon$-monotone Fourier method to solve an associated two-dimensional PIDE on a finite computation domain, together with an efficient padding technique to control wrap-around errors.

\item With a provable strong comparison result, we rigorously prove the convergence of our scheme
the unique viscosity solution of the HJB-QVI as the discretization parameter and the
monotonicity tolerance $\epsilon$ approach zero.
That is, our proposed method can be used for discrete withdrawals, and
can also be shown to converge to the viscosity solution of the HJB-QVI arising in the
continuous withdrawal setting.
\end{itemize}

\item
With a provably convergent numerical method, which allows realistic and useful conclusions to be drawn from the numerical results, we carry out a comprehensive study of the
impact of considering jumps  and stochastic short rate.
Our numerical results suggest that, compared to stochastic interest rate dynamics,
using a constant interest rate results in underpricing of fair insurance fees for GMWBs.
Furthermore, the simultaneous application of jumps and stochastic interest rates
results in (i) a much lower fair insurance fee, and (ii) significantly different
optimal withdrawal behaviors than those obtained from a comparable pure-diffusion model with a comparable constant
interest rate. These findings underscore the importance of realistic modelling and mathematically reliable numerical methods in reducing potential underpricing and overpricing of GMWBs,
contributing to the long-term sustainability of the financial markets.
\end{itemize}
The remainder of the paper is organized as follows. Section~\ref{section:model} describes
the impulse control framework and the underlying processes.
We present in Section~\ref{sec:formulation} an impulse control formulation of the GMWB pricing problem
in the form of a three-dimensional HJB-QVI. Also therein, we also prove a strong comparison result.
A numerical method for solving the HJB-QVI is discussed in Section~\ref{sec:num}.
The convergence of the proposed numerical method is demonstrated in Section~\ref{section:conv}.
In Section~\ref{sec:num_test}, we present and discuss extensive numerical results of GMWBs and the combined impact of jumps and stochastic interest rates on the prices, insurance fees, and the holder's optimal withdrawal behaviors.
Section~\ref{section:cc} concludes the paper and outlines possible future work.

\section{Modeling}
\label{section:model}
We consider a complete probability space $(\mathfrak{S}, \mathfrak{F}, \mathfrak{F}_{0 \le t \le T}, \mathfrak{Q})$,
with sample space $\mathfrak{S}$, sigma-algebra $\mathfrak{F}$, filtration $\mathfrak{F}_{0 \le t \le T}$,
where $T > 0$ is a fixed investment maturity, and a risk-neutral measure $\mathfrak{Q}$ defined on $\mathfrak{F}$.
We discuss the underlying dynamics with an impulse control formulation framework in mind \cite{OksendalSulemBook3, Korn1999}.

Broadly speaking, using an impulse control argument \cite{chen08a}, the holder's optimal withdrawal strategy involves choosing either (i) withdraw continuously at a rate determined by the holder,
but no greater than a cap on the maximum allowed continuous withdrawal rate, hereinafter denoted by $C_r$;
or (ii) withdraw finite amounts at specific times, both determined by the holder,
subject to  a penalty charge which is proportional to the withdrawal amount
and is calculated at the rate $\mu$, where $0 < \mu < 1$, as well as a strictly positive fixed cost $c$.
Due to the associated penalty charge, (ii) is only optimal at some stopping times. To this end, let $\{t^\iota\}_{\iota \le \iota_{\mymax}}$, $\iota_{\mymax} \le \infty$, is any sequence of stopping times with respect to the filtration $\mathfrak{F}_{0 \le t \le T}$ satisfying {\numPDEblue{$0\le t \leq t^1 \leq t^2 < \cdots < t^{\iota_{\mymax}}\leq T$}}.

We denote by (i) $\hat{\gamma}(t)$, $\hat{\gamma}(t)  \in [0, C_r]$, a continuous control representing continuous withdrawal rate at time $t$, and by (ii) an impulse control $\{(t^\iota, \gamma^\iota)\}_{\iota \le \iota_{\mymax}}$,
representing withdrawal/intervention times $\{t^\iota\}_{\iota \le \iota_{\mymax}}$
and associated impulses $\{\gamma^\iota\}_{\iota \le \iota_{\mymax}}$, where  $\gamma^\iota$
is a $\mathfrak{F}_{t^\iota}$-measurable random variable. Here, each $t^\iota$ corresponds to a time at which the holder instantaneously withdraws a finite amount, and $\gamma^\iota$, $\gamma^\iota \in [0, A(t_\iota^-)]$, corresponds to the withdrawal amount at that time. The net revenue cash flow provided to the  holder at time $t^\iota$ is $(1-\mu) \gamma^\iota - c$.

We respectively denote by $Z(t)$, $A(t)$, and $R(t)$, $t \in [0, T]$, the time-$t$ balance of the sub-account,
the guarantee account, and the instantaneous short-rate.
Due to continuous withdrawals and withdrawing finite amounts, the dynamics of $A(t)$ are given by
\EQA
\label{eq:A_dynamics}
dA(t) &=& - \hat{\gamma}(t) {\bf{1}}_{\{A(t) >0 \}} dt,
~~ \text{for}~~ t \neq t^{\iota}, \quad \iota = 1, 2, \ldots, \iota_{\mymax},
\nonumber
\\
A(t) &=& A(t^-) - \gamma^{\iota},  \quad ~~~~~~ \text{for}~~ t = t^{\iota}, \quad \iota = 1, 2, \ldots, \iota_{\mymax}.
\ENA
Let the dynamics of $Z(t)$ and $R(t)$ be given by
\begin{linenomath}
 \begin{subequations}\label{eq:dynamics}
\begin{empheq}[left={\empheqlbrace}]{alignat=3}
&\frac{dZ\left(t\right)}{Z\left(t\right)} &~=~&
\left(R(t)-\beta - \lambda \kappa \right) dt+ \sigz \rho d \Wz (t) + \sigz \sqrt{1-\rho^2} d \Wr (t) + 
d J(t)
\nonumber
\\
&&& - \hat{\gamma}(t) {\bf{1}}_{\{Z(t),A(t) >0 \}} dt,
\quad \text{for  } t \neq t^{k}, \quad \iota = 1, 2, \ldots, \iota_{\mymax},
\label{eq:Z_dynamics}
\\
&Z(t) &~=~& \max \left(Z(t^-) - \gamma^{\iota}, 0 \right),
\quad  \text{for } t = t^{\iota}, \quad \iota = 1, 2, \ldots, \iota_{\mymax},
\label{eq:Z_dynamics*}
\\
&dR(t)
&~=~&
\delta \left(\theta - R(t) \right) dt + \sigr d \Wr (t).
\label{eq:R_dynamics}
\end{empheq}
\end{subequations}
\end{linenomath}
We work under the following assumptions for model \eqref{eq:A_dynamics}-\eqref{eq:dynamics}.
\begin{itemize}
 \item Processes $\{\Wz (t)\}_{0\le t \le T}$ and $\{\Wr (t)\}_{0\le t \le T}$ are two independent standard Wiener processes.
 \item
 The process $\{J(t)\}_{0\le t \le T}$, where  $J(t) =  \sum_{k=1}^{\pi(t)}(Y_{k}-1)$, is a compound Poisson process. Specifically, $\{\pi(t)\}_{0\le t \le T}$ is a Poisson process with a constant finite jump intensity $\lambda\geq 0$; and, with $Y$ being a positive random variable representing the jump multiplier,
  $\{Y_{k}\}_{k = 1}^{\infty}$ are independent and identically distributed (i.i.d.) random variables having the same same distribution as $Y$.  In the dynamics \eqref{eq:Z_dynamics},  $\kappa=\mathbb{E}\left[Y-1\right]$
  represents the expected percentage change in the sub-account balance, due to jumps.
  Here,  $\mathbb{E}[\cdot]$ is the expectation operator taken under the risk-neutral measure $\mathfrak{Q}$.

\item
The Poisson process $\{\pi(t)\}_{0\le t \le T}$, and the sequence of random variables $\{Y_{k}\}_{k = 1}^{\infty}$  are mutually independent, as well as independent of the Wiener processes $\{\Wz (t)\}_{0\le t \le T}$ and $\{\Wr (t)\}_{0\le t \le T}$.
\end{itemize}
In \eqref{eq:Z_dynamics}, $\sigz>0$ is the instantaneous volatility of $Z(t)$ and  $\beta>0$ is
the proportional annual insurance rate paid by the policy holder.
The constant $\rho$, where $\left|\rho \right| < 1$, is a correlation coefficient between $Z(t)$ and $R(t)$.\footnote{Through a Cholesky factorization, the correlation coefficient between $\Wr (t)$ and $\rho \Wz (t) + \sqrt{1-\rho^2} \Wr (t)$ is $\left|\rho \right| < 1$.}
In \eqref{eq:R_dynamics}, $\sigr >0$ is the instantaneous volatility of the short rate, $\delta>0$ is the speed of mean-reversion, $\theta$ is the long-term mean level.
For simplicity, model parameters are assumed to be constant in time; however, the results of this paper can be generalized to the case of time-dependent parameters.

As a specific example, we consider two distribution for the jump multiplier $Y$, namely
the log-normal distribution \cite{merton1975}, and the log-double-exponential
distribution \cite{kou01}. Specifically, we denote by $b(y)$ the density function of the random variable $\ln(Y)$. In the former case, $\ln(Y)$ is normally distributed with mean $\nu$ and standard deviation $\varsigma$, and
\EQA
b\left(y\right)=\frac{1}{\varsigma \sqrt{2\pi} }  \exp\left\{ -\frac{\left(y - \nu \right)^{2}}{2 \varsigma^{2}}\right\}.\label{eq: PDF p(psi) Merton}
\ENA
In the latter case, $\ln Y$ has an asymmetric double-exponential distribution with
\EQA
b\left(y\right)= p_u \eta_{1} e^{-\eta_{1}y}
{\bf{1}}_{\{y \geq 0\}} +
\left(1-p_u\right)\eta_{2}e^{\eta_{2}y}
{\bf{1}}_{\{y<0\}}.
\label{eq: PDF p(psi) Kou}
\ENA
Here, $p_u \in\left[0, 1\right]$, $\eta_{1}>1$ and $\eta_{2}>0$. Given that a jump occurs, $p_u$ is the probability of an upward jump, and $(1-p_u)$ is
the probability of a downward jump.

\section{Impulse control formulation}
\label{sec:formulation}
For the controlled underlying process $\left( Z(t), R(t), A(t) \right)$, $t \in [0, T]$, let $(z, r, a)$ be the state of the system.  Let $\tau = T- t$, for $z > 0$,  we apply the change of variable $w = \ln(z) \in (-\infty, \infty)$.
With ${\mathbf{x}} = (w, r, a, \tau)$, we denote by $v({\mathbf{x}}) \equiv v(w, r, a, \tau)$ the time-$\tau$ no-arbitrage price of a GMWB when $Z(t) = e^w$, $R(t) = r$ and $A(t) = a$.
Using dynamic programming, we can show that, under dynamics  \eqref{eq:A_dynamics}-\eqref{eq:dynamics},
$v(w, r, a, \tau)$ satisfy the impulse control formulation \cite{online, chen08a}
\EQA
\label{eq:omega_inf_all}
&&\min\bigg\{
v_{\tau} - \mathcal{L}v - \mathcal{J}v
- \sup_{\hat{\gamma} \in [0, C_r]}
\hat{\gamma}\left(1 - e^{-w}v_w - v_a\right) \mathbf{1}_{\{a > 0\}},
\nonumber
\\
&& \qquad \qquad v - \sup_{\gamma \in [0,a]}
\left[ v\left(\ln\left(\max\left(e^w-\gamma, e^{\winf}\right)\right), a - \gamma, \tau\right)
+ \left(1-\mu\right)\gamma  - c\right]
\bigg\} ~=~ 0,
\ENA
where $(w, r, a, \tau) \in \Oinf ~\equiv~ (-\infty, \infty) \times (-\infty, \infty) \times [a_{\min}, a_{\max}]  \times [0, T)$, with $a_{\min} = 0$ and $a_{\max} = z_0$,
and 
\EQA
\mathcal{L}v\left(\mathbf{x}\right)
&=& \frac{\sigz^{2}}{2} v_{ww} + \rho \sigz \sigr v_{wr} + \frac{\sigr^{2}}{2} v_{rr} +  \left(r-\frac{\sigz^{2}}{2} - \beta  - \lambda \kappa \right) v_w +
\delta \left(\theta - r \right) v_r - (r + \lambda) v,
\nonumber
\\
\mathcal{J}v\left(\mathbf{x}\right) &=&
\lambda\int_{-\infty}^{\infty} v(w + y, r, a, \tau )~b(y)~dy.
\label{eq:Operator_LJ}
\ENA
Here, in \eqref{eq:omega_inf_all}, $\winf \ll 0$ is a constant to avoid the indeterminate case of
of $\ln(0)$, due to condition \eqref{eq:Z_dynamics*}; the constant positive fixed cost $c$ is introduced as a technical tool to ensure uniqueness of the impulse formulation, as commonly done in the impulse control literature \cite{OksendalSulemBook3, Pham, MNIF06};
in \eqref{eq:Operator_LJ}, $b\left(\cdot\right)$  is the probability density function of $\ln Y$.

\subsection{Localization}
The GMWB impulse control formulation \eqref{eq:omega_inf_all} is posed on
the infinite domain $\Oinf$. For problem statement and convergence analysis
of numerical schemes, we define a localized GMWB impulse formulation.
To this end, with $w_{\min} < 0 < w_{\max}$, $r_{\min} < 0 < r_{\max}$, and $\left|w_{\min}\right|$, $w_{\max}$, $\left|r_{\min}\right|$, $r_{\max}$ sufficiently large,
we define the following sub-domains:

\begin{figure}[!ht]
\begin{minipage}{0.45\linewidth}
\begin{linenomath}
\begin{align}
\label{eq:sub_domain_whole}
\Omega_{\text{in}} &=
(w_{\min}, w_{\max}) \times (r_{\min}, r_{\max}) \times (a_{\min}, a_{\max}] \times (0, T],
\nonumber
\\
\Oinf_{\tau_0} &= (-\infty, \infty) \times
(-\infty, \infty)  \times [a_{\min}, a_{\max}] \times \{0\},
\nonumber
\\
\Oinf_{w_{\max}} &= [w_{\max}, \infty) \times
(r_{\min}, r_{\max}) \times [a_{\min}, a_{\max}] \times (0, T],
\nonumber
\\
\Oinf_{w_{\min}}&= (-\infty, w_{\min}] \times
(r_{\min}, r_{\max}) \times (a_{\min}, a_{\max}] \times  (0, T],
\\
\Omega_{a_{\min}} &= (w_{\min}, w_{\max}) \times
(r_{\min}, r_{\max}) \times \{a_{\min}\} \times (0, T],
\nonumber
\\
\Oinf_{wa_{\min}} &= (-\infty, w_{\min}] \times
(r_{\min}, r_{\max}) \times \{a_{\min}\} \times (0, T],
\nonumber
\\
\Oinf_{\myot} &= \Oinf \setminus \Omega_{\text{in}}  \setminus \Oinf_{\tau_0} \setminus \Oinf_{w_{\max}}
\setminus \Oinf_{w_{\min}} \setminus \Omega_{a_{\min}} \setminus \Oinf_{wa_{\min}}.
\nonumber
\end{align}
\end{linenomath}
An illustration of the sub-domains for the localized problem corresponding to a fixed $a \in [a_{\min}, a_{\max}]$
is given in Figure~\ref{fig:domain}.
\end{minipage}
\hspace*{-0.5cm}
\begin{minipage}{0.48\linewidth}
\begin{center}
\begin{tikzpicture}[scale=0.50]
   \draw [thick] [stealth-stealth] (-2,-1) -- (10,-1);
   \draw [thick] [stealth-stealth] (-2,5) -- (10,5);
   \draw [thick] [stealth-stealth] (3,-3) -- (3,7.5);
   \draw [thick](1.25,-1) --(1.25,5);
   \draw [thick](7,-1) --(7,5);
    \draw [dotted, line width=0.5mm] (-1,-1) --(-1,-2.5) -- (9,-2.5) -- (9,-1);
    \draw [dotted, line width=0.5mm] (-1,5) --(-1,6.5) -- (9,6.5) -- (9,5);
   \node at (5,2.5) {\scalebox{1.0}{$\Omega_{\myin}$}};
   \node at (5.2,1.5) {(\scalebox{1.0}{$\Omega_{a_{\min}}$})};

   \node at (-1,2.5) {\scalebox{1.0}{$\Oinf_{w_{\min}}$}};
   \node at (-1,1.5) {(\scalebox{1.0}{$\Oinf_{wa_{\min}}$})};
   \node at (9,2.5) {\scalebox{1.0}{$\Oinf_{w_{\max}}$}};
   \node at (4.5,5.75) {\scalebox{1.0}{$\Oinf_{\myot}$}};
   \node at (4.5,-1.75) {\scalebox{1.0}{$\Oinf_{\myot}$}};
   \node [below] at (-2,0) {\scalebox{0.8}{$-\infty$}};
   \node [below] at (10,0) {\scalebox{0.8}{$\infty$}};
   \node [right] at (3,-3) {\scalebox{0.8}{$-\infty$}};
   \node [right] at (3,7.5) {\scalebox{0.8}{$\infty$}};
   \node [below] at (0.5,0) {\scalebox{0.8}{$w_{\min}$}};
   \node [below] at (8,0) {\scalebox{0.8}{$w_{\max}$}};
   \node [above] at (3.8,-1) {\scalebox{0.8}{$r_{\min}$}};
   \node [below] at (3.8,5) {\scalebox{0.8}{$r_{\max}$}};
\end{tikzpicture}
\end{center}
\vspace*{-0.5cm}
\caption{Spatial computational domain at each $\tau$ and for a fixed $a \in [a_{\min}, a_{\max}]$;
\\
at $a = 0$,  $\Omega_{\myin} \equiv \Omega_{a_{\min}}$ and $\Oinf_{w_{\min}} \equiv \Oinf_{wa_{\min}}$.}
\label{fig:domain}
\end{minipage}
\end{figure}
We now present equations for sub-domains defined in \eqref{eq:sub_domain_whole}.
\begin{itemize}
\item For $(w, r, a, \tau)\in \Omega_{\myin}$, we have \eqref{eq:omega_inf_all}.

\item For $(w, r, a, \tau)\in \Oinf_{\tau_0}$,
we use the initial condition $v(w, a, 0) =  \max(e^{w}, (1-\mu)a-c) \wedge \einf$
for a finite $\wpinf \gg w_{\max}$, where $x\wedge y = \min(x, y)$.

\item For $(w, r, a, \tau) \in \Oinf_{w_{\max}}$, we follow \cite{Dai, chen08a} to impose the Dirichelet-type boundary condition
\EQA
\label{eq:vmax_p2}
v  = e^{-\beta \tau} (e^w \wedge \einf).
\ENA
We note that the theoretical quantity $\wpinf$ is needed to indicate that the
solutions $\Oinf_{\tau_0}$ and $\Oinf_{w_{\max}}$ are bounded as $w \to \infty$, and
it does not need to be numerically specified.

\item
As $w \to -\infty$ (i.e.\
$z = e^w \to 0$),
using the asymptotic forms of the HJB-QVI \eqref{eq:omega_inf_all}, for $(w, r, a, \tau) \in  \Oinf_{w_{\min}}$,
\eqref{eq:omega_inf_all} is reduced to the boundary condition
\EQA
\label{eq:vwmin_p2}
\min \left\{v_{\tau} - \Ld v - \sup_{\hat{\gamma} \in [0, C_r]}
                \left(\hat{\gamma} - \hat{\gamma}v_a\right) {\mathbf{1}}_{\{a>0\}} ,
                 v -  \sup_{\gamma \in [0,a]}
   [v(w, a -\gamma, \tau) + (1-\mu)\gamma -c] \right\}=0,
\ENA
where the degenerated differential operator $\Ld$ is defined by
\EQA
\label{eq:Ltilde}
\Ld v ~:= ~ \frac{\sigr^{2}}{2} v_{rr} + \delta \left( \theta -r \right) v_r - r  v.
\ENA
This is essentially a Dirichlet boundary condition since it
can be solved without using any information from $\Omega_{\myin} \cup \Omega_{a_{\min}}$.

\item For $(w, r, a, \tau) \in \Omega_{a_{\min}}$,
the impulse formulation \eqref{eq:omega_inf_all}
becomes the PIDE
$v_{\tau} - \mathcal{L} v - \mathcal{J}v = 0$.

\item For $(w, r, a, \tau) \in \Oinf_{wa_{\min}}$,
\eqref{eq:vwmin_p2} becomes
$v_{\tau} - \Ld v= 0$.


\item
For $(w, r, a, \tau) \in \Oinf_{\myot}$, we note in this case, significant difficulty arises in choosing a boundary 
condition based on asymptotic forms of the HJB-QVI \eqref{eq:omega_inf_all},
or the holder's optimal withdrawal behaviours.  Since a detailed analysis of the boundary conditions is not the focus of this paper, we leave it as a topic for future research. For simplicity, we follow \cite{DANG2010, Dempster1996} to choose Dirichlet-type ``stopped process'' boundary conditions where we stop the processes $\left(Z(t), R(t), A(t) \right)$ when $R(t)$ hits the boundary. Thus, $(w, r, a, \tau) \in \Oinf_{\myot}$, the value is simply the discounted payoff for the current values of the state variables, i.e.
\EQA
\label{eq:vother}
v(w, r, a, \tau)
&=& \z(w, r, a, \tau)
~=~  p_b(\bar{r}, \tau;T) \max(e^w, (1-\mu) a -c) \wedge \einf,
\ENA
where $\bar{r} := \min( \max( r, r_{\min}), r_{\max} )$. Here, $p_b(r, \tau;T)$ is the price at time $(T-\tau)$ of a zero coupon bond with maturity $T$ given by the closed-form expression \cite{interestratemodels}
\EQA
\label{eq:zero_coupon}
p_b(r, \tau;T) = \exp \left\{ \left(\theta - \frac{\sigr^2}{2 \delta^2} \right) \left(\frac{1}{\delta} \left(1 - e^{-\delta \tau} \right) - \tau \right)  - \frac{\sigr^2}{4 \delta^3} \left(1 - e^{-\delta \tau} \right)^2   - \frac{r}{\delta} \left(1 - e^{-\delta \tau} \right)  \right\}.
\ENA
\end{itemize}

Note that no further information is needed along the boundary $a \to a_{\max}$ due to the hyperbolic
nature of the variable $a$ in the HJB-QVI \eqref{eq:omega_inf_all}.
Although the above-mentioned artificial boundary conditions may induce additional approximation errors in the numerical solutions,  we can make these errors arbitrarily small by choosing sufficiently large values for $|w_{\min}|$, $w_{\max}$, $|r_{\min}|$, and $r_{\max}$.

\subsection{Definition of viscosity solution}
We now write the GMWB pricing problem in a compact form, which includes the terminal
and boundary conditions in a single equation. We define the
intervention operator
\begin{linenomath}
\begin{subequations}
\label{eq:Operator_M}
\begin{empheq}[left={\mathcal{M}(\gamma) v ({\mathbf{x}})= \empheqlbrace}]{alignat=3}
& v(w, r, a -\gamma, \tau)
   + \gamma(1-\mu) -c & &\quad {\mathbf{x}} \in \Oinf_{w_{\min}},
\label{eq:Operator_M_a}
\\
&  v\left(\ln(\max(e^w-\gamma, e^{\winf})), r, a -\gamma, \tau\right)
   + \gamma(1-\mu) -c & &\quad{\mathbf{x}} \in \Omega_{\myin}.
\label{eq:Operator_M_b}
\end{empheq}
\end{subequations}
\end{linenomath}
With ${\mathbf{x}} = (w, r, a, \tau)$, we let $Dv({\mathbf{x}})$ and
$D^2 v( {\mathbf{x}} )$ represent the first-order and second-order partial derivatives of $v\left( {\mathbf{x}} \right)$, and define
\EQA
\label{eq:Fomega_def}
F_{\Oinf} \left({ \mathbf{x}}, v\right) \equiv
F_{\Oinf}\left({ \mathbf{x}},
                    v({\mathbf{x}}),
                    Dv({\mathbf{x}}),
                    D^2 v({\mathbf{x}}),
                    \mathcal{J} v({\mathbf{x}}),
                    \mathcal{M} v({\mathbf{x}})
                \right)
\ENA
where
\EQAS
F_{\Oinf} \left({ \mathbf{x}}, v\right)
~=~
\left\{
\begin{array}{lllll}
F_{\myin} \left({ \mathbf{x}}, v\right) &\equiv&
    F_{\myin}\left( { \mathbf{x}},
                    v({\mathbf{x}}),
                    Dv({\mathbf{x}}),
                    D^2 v({\mathbf{x}}),
                    \mathcal{J} v({\mathbf{x}}),
                    \mathcal{M}v({\mathbf{x}})
                \right),
           &
           \mathbf{x} \in \Omega_{\myin},
\\
F_{a_{\min}}\left({ \mathbf{x}}, v\right) &\equiv& F_{a_{\min}}\left(
                        {\mathbf{x}},
                        v({\mathbf{x}}),
                        Dv({\mathbf{x}}),
                        D^2 v({\mathbf{x}})
                        , \mathcal{J} v({\mathbf{x}})
                        \right),
                    &
                    \mathbf{x} \in \Omega_{a_{\min}},

\\
F_{w_{\min}}\left({ \mathbf{x}}, v\right) &\equiv& F_{w_{\min}}\left(
                 {\mathbf{x}},
                 v({\mathbf{x}}),
                 Dv({\mathbf{x}}),
                 \mathcal{M} v({\mathbf{x}})
                        \right),
                &
                \mathbf{x} \in \Oinf_{w_{\min}},
\\
F_{wa_{\min}} \left({ \mathbf{x}}, v\right) &\equiv& F_{wa_{\min}}\left(
                        {\mathbf{x}},
                        v({\mathbf{x}}),
                        Dv({\mathbf{x}})
                        \right),
                    &
                    \mathbf{x} \in \Oinf_{wa_{\min}},
\\
F_{w_{\max}}\left({ \mathbf{x}}, v\right) &\equiv& F_{w_{\max}}\left(
                 {\mathbf{x}},
                 v({\mathbf{x}})
                \right),
                &
                \mathbf{x} \in \Oinf_{w_{\max}},
\\
F_{\myot}\left({ \mathbf{x}}, v\right) &\equiv& F_{\myot}\left(
                 {\mathbf{x}},
                 v({\mathbf{x}})
                \right),
                &
                \mathbf{x} \in \Oinf_{\myot},
\\
F_{\tau_0}\left({ \mathbf{x}}, v\right) &\equiv& F_{\tau_0}({\mathbf{x}}, v({\mathbf{x}})),
        &
    \mathbf{x} \in \Oinf_{\tau_0},
\end{array}
\right.
\ENAS
with operators
\EQA
F_{\myin}\left({ \mathbf{x}}, v\right) &=&
         \min \left[ v_{\tau} - \mathcal{L} v - \mathcal{J}v
                -\sup_{\hat{\gamma} \in [0, C_r]}
                \left(\hat{\gamma} - \hat{\gamma}e^{-w}v_w - \hat{\gamma}v_a\right){\mathbf{1}}_{\{a>0\}} ,
                 v -  \sup_{\gamma \in [0, a]} \mathcal{M} v  \right],
\label{eq:Finn}
\\
F_{w_{\min}}\left({ \mathbf{x}}, v\right) &=&    \min \left[v_{\tau} - \Ld v
                -\sup_{\hat{\gamma} \in [0, C_r]}
                \left(\hat{\gamma} - \hat{\gamma}v_a\right) {\mathbf{1}}_{\{a>0\}} ,
                 v - \sup_{\gamma \in [0, a]} \mathcal{M} v
                 \right],
\label{eq:fWmin}
\\
F_{a_{\min}}\left({ \mathbf{x}}, v\right) &=&
         v_{\tau} - \mathcal{L} v - \mathcal{J}v,
\label{eq:fAmin}
\\
F_{wa_{\min}}\left({ \mathbf{x}}, v\right) &=&
         v_{\tau} - \Ld v,
\label{eq:fWAmin}
\\
F_{w_{\max}}\left({ \mathbf{x}}, v\right) &=&
         v -  e^{-\beta \tau} (e^w \wedge \einf),
\label{eq:fWmax}
\\
F_{\myot}\left({ \mathbf{x}}, v\right) &=&
             v - \z(w, r, a, \tau),
\label{eq:fother}
\\
F_{\tau_0}\left({ \mathbf{x}}, v\right) &=&
            v - \max(e^w, (1-\mu) a -c) \wedge \einf.
\label{eq:ftau0}
\ENA

\begin{definition}[Impulse control GMWB pricing problem]
\label{def:impulse_def}
The  pricing problem for the GMWB under
an impulse control formulation
is defined as
\EQA
\label{eq:gmwb_def}
F_{\Oinf}\left( { \mathbf{x}},
                    v({\mathbf{x}}),
                    Dv({\mathbf{x}}),
                    D^2 v({\mathbf{x}}),
                    \mathcal{J} v({\mathbf{x}}),
                    \mathcal{M} v({\mathbf{x}})
                \right) ~=~ 0,
\ENA
where the operator $F_{\Oinf}(\cdot)$
is defined in \eqref{eq:Fomega_def}.
\end{definition}
Next, we recall the notions of the upper semicontinuous (u.s.c.\ in short)
and the lower semicontinuous (l.s.c.\ in short) envelops of a function
$u: \mathbb{X} \rightarrow \mathbb{R}$, where $\mathbb{X}$ is a closed
subset of $\mathbb{R}^n$. They are respectively denoted by $u^*(\cdot)$
(for the u.s.c.\ envelop) and $u_*(\cdot)$ (for the l.s.c.\ envelop),
and are given by
\EQAS
\label{eq:envelop}
u^*({\mathbf{\hat{x}}}) = \limsup_{
    \subalign{{\mathbf{x}} &\to {\mathbf{\hat{x}}}
\\
{\mathbf{x}}, {\mathbf{\hat{x}}} &\in\mathbb{X}
}}
u({\mathbf{x}})
\quad
(\text{resp.}
\quad
u_*({\mathbf{\hat{x}}}) = \liminf_{
    \subalign{{\mathbf{x}} &\to {\mathbf{\hat{x}}}
\\
{\mathbf{x}}, {\mathbf{\hat{x}}} &\in\mathbb{X}
}}
u({\mathbf{x}})
).
\ENAS

In general, the solution to impulse control problems are non-smooth, and we seek
the viscosity solution of equation~\eqref{eq:gmwb_def} \cite{davis2010, Seydel2009, Guo2009}.
Since equation~\eqref{eq:gmwb_def} is defined on an infinite domain,
we need to have a suitable growth condition at infinity for the solution
\cite{Barles2008, Seydel2009}. To this end, let $\G$ be the set of  bounded functions
defined by \cite{Barles2008, Seydel2009}
\EQ
\label{eq:G}
\begin{aligned}
\G &= \l\{
u: \Oinf \to \mathbb{R},
~~\sup_{{\bf{x}} \in \Oinf}  |u({\bf{x}})| < \infty
\r\}.
\end{aligned}
\EN
\begin{definition}[Viscosity solution of equation \eqref{eq:gmwb_def}]
\label{def:vis_def_common}
A~locally bounded function $v \in \G$ is a viscosity subsolution (resp.\ supersolution) of \eqref{eq:gmwb_def}
in  $\Oinf$ if for all test function $\phi \in \G \cap \C{\Oinf}$
and for all points ${\bf{\hat{x}}} \in \Oinf$ such that
$v^*-\phi$ has a \emph{global} maximum on $\Oinf$ at ${\bf{\hat{x}}}$
and $v^*({\bf{\hat{x}}}) = \phi({\bf{\hat{x}}})$
(resp.\ $v_*-\phi$ has a \emph{global} minimum on $\Oinf$ at ${\bf{\hat{x}}}$
and $v_*({\bf{\hat{x}}}) = \phi({\bf{\hat{x}}})$), we have
\begin{eqnarray}
\label{eq:Def1}
\left(F_{\Oinf}\right)_* \left({\bf{\hat{x}}}, \phi({\bf{\hat{x}}}), D\phi({\bf{\hat{x}}}), D^2 \phi({\bf{\hat{x}}}),
            \mathcal{J} \phi({\bf{\hat{x}}}),
            \mathcal{M} \phi({\bf{\hat{x}}})
             \right) &\leq  & 0,
    \\
    \big(
\text{resp.\ }
\quad
      \left(F_{\Oinf}\right)^* \l(
              {\bf{\hat{x}}}, \phi({\bf{\hat{x}}}), D\phi({\bf{\hat{x}}}), D^2 \phi({\bf{\hat{x}}}),
             \mathcal{J} \phi ({\bf{\hat{x}}}),
             \mathcal{M} \phi ({\bf{\hat{x}}})
             \r) &\geq & 0,
\big)\nonumber
\end{eqnarray}
where the operator $F_{\Oinf}(\cdot)$ is defined in \eqref{eq:Fomega_def}.

\medskip
(ii) A locally bounded function $v \in \G$ is a viscosity solution of \eqref{eq:gmwb_def}
in  $\Omega_{\myin} \cup \Omega_{a_{\min}}$ if $v$ is a viscosity subsolution and a viscosity supersolution
in $\Omega_{\myin} \cup \Omega_{a_{\min}}$.
\end{definition}

\subsection{A strong comparison result}
\label{ssc:comparison}
In the context of numerical solutions to HJB-QVIs, convergence of numerical methods to the viscosity
typically requires stability, consistency, monotonicity, provided that a strong comparison result \cite{crandall-ishii-lions:1992, MNIF06, Jakobsen2006, Barles2008, Seydel2009, Berestycki2014, barles-souganidis:1991, barles-burdeau:1995}.  
Specifically, using stability, consistency, and monotonicity of a numerical scheme,
the common route is to establish the candidate for u.s.c.\ subsolution (resp.\ l.s.c.\ supersolution) of the HJB-QVI
using $\limsup$ (resp.\ $\liminf$) of the numerical solutions as a discretization parameter
approaches zero. We respectively denote by $\hat{u}$ the subsolution (resp.\ $\hat{v}$ the supersolution)
in a target convergence region $\mathcal{S}$ which is a non-empty subset of $\Oinf$. By construction, we have
$\hat{u}({\mathbf{x}}) \ge \hat{v}({\mathbf{x}})$ for all ${\mathbf{x}} \in \mathcal{S}$.
If a strong comparison result holds in $\mathcal{S}$, it means that for subsolution $\hat{u}({\mathbf{x}})$ and supersolution $\hat{v}({\mathbf{x}})$, we have $\hat{u}({\mathbf{x}}) \le \hat{v}({\mathbf{x}})$ for all ${\mathbf{x}} \in \mathcal{S}$. Therefore, a unique continuous viscosity solution exists in $\mathcal{S}$.
We note that, while stability, consistency and monotonicity are required properties of numerical methods, a strong comparison result is problem dependent.

In our paper {\numPDEblue{\cite[Lemma~B.1 and Theorem~B.1]{online}}}, we present a framework for proving a strong comparison result for HJB-QVIs of a form similar to \eqref{eq:gmwb_def} where jump-diffusion dynamics with a positive constant interest rate
are \blue{considered}. For the HJB-QVI~\eqref{eq:gmwb_def}, using the aforementioned framework, we are able to show a strong comparison result for $\Omega_{\myin} \cup \Omega_{a_{\min}}$, where $\Omega_{a_{\min}} \subset \partial \Omega_{\myin}$.
This result is \blue{presented} in Theorem~\ref{thm:comparison} below.
\begin{theorem}
\label{thm:comparison}
If function $\hat{u}$ (resp.\ $\hat{v}$) is a u.s.c.\ viscosity subsolution (resp.\ l.s.c.\ supersolution)
of the HJB-QVI~\eqref{eq:gmwb_def} in $\Omega$ in the sense of Definition~\ref{def:vis_def_common}, then
we have $\hat{u} \leq \hat{v}$ in $\Omega_{\myin} \cup \Omega_{a_{\min}}$.
\end{theorem}

\begin{proof}[Proof of Theorem~\ref{thm:comparison}]
We follow the framework presented in \cite{online}[Lemma~B.1 and Theorem~B.1].
With the target region being $\mathcal{S} =  \Omega_{\myin} \cup \Omega_{a_{\min}}$,
we rewrite Definition~\ref{def:vis_def_common} into an equivalent definition as follows.
\begin{itemize}
\item[(i)] In the non-local terms $\mathcal{J}(\cdot)$ and $\mathcal{M} (\cdot)$,
the smooth test function $\phi(\hat{\textbf{x}})$ is replaced by $v^*(\hat{\textbf{x}})$ for subsolution
(resp.\ $v_*(\hat{\textbf{x}})$ for supersolution),
\item[(ii)] The envelopes $(F_{\Oinf})_\ast$ (resp.\ $(F_{\Oinf})^\ast$) is eliminated from the definition of subsolution (resp.\ supersolution).
\end{itemize}
We refer to this definition as Def-A, and it is the definition we use to prove a strong comparison result.\footnote{For the purpose of verifying consistency of a numerical scheme, it is convenient to use Definition~\ref{def:vis_def_common}. However, it turns out more convenient to use the equivalent definition to prove a strong comparison result for the HJB-QVI~\eqref{eq:gmwb_def}. Similar arguments can be also referred to \cite{davis2010, Seydel2009, Azimzadeh2018}.}

Unlike the setting in \cite{online}, where a positive constant interest rate is used,
a Gaussian stochastic interest rate is considered in the present paper, which could be negative.
Therefore, the framework in \cite{online} is not directly applicable
without an important preprocessing step (shown below).
\begin{itemize}

    \item Given the HJB-QVI with $F_{\Oinf}(\cdot) = 0$ in \eqref{eq:gmwb_def},
    let $q >  -r_{\min}$ be fixed, implying $r + q > 0$ for all $r\in (r_{\min}, r_{\max})$, we introduce an HJB-QVI
    $F_{\Oinf}(\cdot; q) = 0$ which is similar to $F_{\Oinf}(\cdot) = 0$ except in
    $\Omega_{\myin} \cup \Omega_{a_{\min}}$, where
    $F_{\myin}(\cdot; q)$ and $F_{a_{\min}}(\cdot; q)$  are defined by
    \EQAS
    F_{\myin}(\mathbf{x}, v; q) &=& \min\bigg[
v_{\tau} - \mathcal{L}v + qv - \mathcal{J}v
- \sup_{\hat{\gamma} \in [0, C_r]}
\hat{\gamma}\left(e^{-q \tau} - e^{-w}v_w - v_a\right) \mathbf{1}_{\{a > 0\}},
\nonumber
\\
&& v - \sup_{\gamma \in [0,a]}
\left[ v\left(\ln\left(\max\left(e^w-\gamma, e^{\winf}\right)\right), a - \gamma, \tau\right)
+ (\left(1-\mu\right)\gamma  - c) e^{-q \tau} \right]
\bigg],
\nonumber
\\
F_{a_{\min}}\left({ \mathbf{x}}, v; q\right) &=&
         v_{\tau} - \mathcal{L} v + qv - \mathcal{J}v.
\ENAS

\item It is straightforward to show that: in the sense of Def-A, if $\hat{u}$ is a u.s.c.\ viscosity subsolution (resp.\ $\hat{v}$ is a l.s.c.\ viscosity supersolution) of
    $F_{\Oinf}(\cdot) = 0$ in $\Omega_{\myin} \cup \Omega_{a_{\min}}$, then $e^{-q\tau}\hat{u}$ is a u.s.c.\ viscosity subsolution (resp.\ $e^{-q\tau}\hat{v}$ is a l.s.c.\ viscosity subsolution) of $F_{\Oinf}(\cdot; q) = 0$ in $\Omega_{\myin} \cup \Omega_{a_{\min}}$.
\end{itemize}
Finally, using the same steps as in Lemma~B.1 and Theorem~B.1 of \cite{online} for the
HJB-QVI $F_{\Oinf}(\cdot; q) = 0$, we can prove that a strong comparison results holds for $\Omega_{\myin} \cup \Omega_{a_{\min}}$, i.e.\ $e^{-q\tau}\hat{u} \le e^{-q\tau}\hat{v}$, or equivalently, $\hat{u} \le \hat{v}$ in $\Omega_{\myin} \cup \Omega_{a_{\min}}$, which is the desired outcome.
\end{proof}
We conclude this subsection by noting that, as well-noted in the literature \cite{MNIF06, chen08a, DangForsyth2014, online, huang:2010, Pham}, it is usually the case that a strong comparison result does not hold on the whole definition domain including boundary sub-domains, because this would imply the continuity of the value function across the boundary regions, which is not true for some impulse control problems, including the HJB-QVI~\eqref{eq:gmwb_def}.
In particular, it is possible that loss of boundary data can occur over parts of $\Gamma = \partial \Omega_{\myin}\setminus \Omega_{a_{\min}}$, i.e.\ as  $\tau \to 0$, $w~\to~\{w_{\min}$, $w_{\max}\}$ and $r~\to~\{r_{\min}$, $r_{\max}\}$, hence, we cannot hope that a strong comparison result holds on $\Gamma$. However, these problematic parts of $\Gamma$ are trivial to handle in the sense that either the boundary data is used or is irrelevant. In all cases, we consider the computed solution on those parts of $\Gamma$ as the limiting value approaching $\Gamma$ from the interior.

\section{Numerical methods }
\label{sec:num}

\subsection{Overview}
\label{ssec:over}
Similar to the approach taken in our papers \cite{online, chen08a},
we will tackle the HJB-QVI \eqref{eq:gmwb_def} from a discrete withdrawal \blue{scenario}
which was first suggested in \cite{Dai}. To this end, we first introduce
a set of discrete intervention (withdrawal) times as follows.
Let $\{\tau_m\}$, $m = 0, \ldots, M$, be a partition of $[0, T]$,
where for simplicity,  an uniform spacing is used, i.e.\ $\tau_{m} = m\Delta \tau$
and $\Delta \tau = T/M$. Following \cite{Dai, chen08a}, there is no withdrawal
allowed at time $t = 0$, or equivalently, at $\tau_M = T$;
therefore, the set of intervention times is $\{\tau_m\}$,
$m = 0, \ldots, M-1$.

Broadly speaking, over the time interval $[\tau_{m}, \tau_{m+1}]$, $m = 0, \ldots, M-1$,
our numerical approach consists of two steps, namely intervention in $[\tau_m, \tau_{m}^+]$.
and time-advancement in $[\tau_{m}^+ , \tau_{m+1}]$.
Central to our method is the time-advancement step for the target region of convergence
$\Omega_{\myin} \cup \Omega_{a_{\min}}$. For this step, $a\in [a_{\min}, a_{\max}]$ is fixed,
and our starting point is a linear PIDE in
$(w, r)$ of the form
    \EQ
    \label{eq:dis_pide_ori}
 v_{\tau} - \mathcal{L}v - \mathcal{J}v = 0,
 \quad
 w \in (-\infty, \infty),~r \in (-\infty, \infty),~\tau \in (\tau_{m}^+ , \tau_{m+1}].
    \EN
where the operators $\mathcal{L}$ and $\mathcal{J}$ are given in \eqref{eq:Operator_LJ},
subject to a generic initial condition at time $\tau_{m}^+$ given by
$\vh(w, r, a, \tau_{m}^+)$ obtained from the intervention step above.
Here,
\begin{linenomath}
\begin{subequations}\label{eq:dis_pide_term}
\begin{empheq}[left={\vh (w, r, a, \tau_{m}^+) = \empheqlbrace}]{alignat=3}
&v(w, r, a, \tau_{m}^+) && \qquad (w, r, a, \tau_{m+1}) \in \Omega_{\myin} \cup \Omega_{a_{\min}},
\label{eq:dis_pide_term_a}
\\
&v_{bc}(w, r, a, \tau_{m})
&& \qquad
(w, r, a, \tau_{m+1}) \in
\Oinf \setminus \left( \Omega_{\myin} \cup \Omega_{a_{\min}} \right).
\label{eq:dis_pide_term_b}
\end{empheq}
\end{subequations}
\end{linenomath}
In \eqref{eq:dis_pide_term_a},  $v(w, r, a, \tau_{m}^+)$ is the intermediate results from the intervention step,
and  $v_{bc}(w, r, a, \tau_{m}^+)$ in \eqref{eq:dis_pide_term_b} is the boundary conditions at time-$\tau_{m}$
satisfying \eqref{eq:vwmin_p2}, \eqref{eq:vmax_p2}, \eqref{eq:vother}
in $\Oinf_{w_{\min}} \cup \Oinf_{wa_{\min}} \cup \Oinf_{w_{\max}} \cup \Oinf_{\myot}$.

The key challenge in solving the PIDE \eqref{eq:dis_pide_ori} is that a closed-form expression for \blue{its} Green's function is not known to exist, due to the $v_r$ term arising from the short rate. (Also see  \cite{Jaimungal2013}
for relevant discussions \blue{related to} similar difficulties). To handle the above challenge, we consider a combination of a semi-Lagrangian (SL) method and a Green's function approach. In particular, we consider writing $\mathcal{L} v = \Lg v +  \Ls v - rv$, where
\EQA
\Lg v := \frac{\sigz^{2}}{2} v_{ww} + \rho \sigz \sigr v_{wr} + \frac{\sigr^{2}}{2} v_{rr} - \lambda \kappa v_w - \lambda v,
~
\Ls v := (r -  \frac{\sigz^{2}}{2} - \beta) v_w + \delta(\theta - r) v_r.
\label{eq:operator_Lgs}
\ENA
To solve the PIDE \eqref{eq:dis_pide_ori} in  $\Omega_{\myin} \cup \Omega_{a_{\min}}$,
we first handle the term $\Ls v - rv$ by \blue{an} SL discretization method in $\Omega_{\myin} \cup \Omega_{a_{\min}}$.
(This is discussed in Subsection~\ref{ssc:sld}.).
We then effectively solve the PIDE of the form
\EQA
&& \bvsl_{\tau} - \Lg \vsl  - \mathcal{J} \vsl = 0, \quad
w \in (-\infty, \infty), ~r \in (-\infty, \infty),~ \tau \in (\tau_{m}^+ , \tau_{m+1}],
\label{eq:dis_pide}
\ENA
where $\vsl$ is the  unknown function,  subject to a generic initial condition $\hat{v}_{\scalebox{0.5}{\text{SL}}}(w, r, a, \taus)$ given as follows. Letting ${\bf{x}} = (w, r, a, \tau_{m+1})$,
for ${\bf{x}}~\in~\Omega_{\myin} \cup \Omega_{a_{\min}}$,
$\hat{v}_{\scalebox{0.5}{\text{SL}}}({\bf{x}})$ given by \blue{an} SL discretization method
combined with $\vh (w, r, a, \tau_m^+)$ provided in \eqref{eq:dis_pide_term_a}-\eqref{eq:dis_pide_term_b};
otherwise, $\hat{v}_{\scalebox{0.5}{\text{SL}}}({\bf{x}})$ is given by $v_{bc}({\bf{x}})$ as in \eqref{eq:dis_pide_term_b}.

To numerically solve the PIDE \eqref{eq:dis_pide} for $\vsl\l(w, r, a, \tau_{m+1}\r)$, we start from
a Green's function approach. It is a known fact that the Green's function $g\left(\cdot\right)$ associated with
the PIDE~\eqref{eq:dis_pide} has the form $g(w, w', r, r', \dtau) \equiv  g(w-w', r-r', \dtau)$
\cite{garronigreenfunctionssecond92, Duffy2015}.
Therefore,  the solution $\vsl\l(w, r, a, \tau_{m+1}\r)$ for
$(w, r) \in \D \equiv (w_{\min}, w_{\max}) \times (r_{\min}, r_{\max})$ can be represented as the convolution integral of the Green's function $g\left(\cdot, \Delta \tau\right)$ and the initial condition $\vhsl(w, r, a, \tau_{m}^+)$ as follows \cite{garronigreenfunctionssecond92, Duffy2015}
\EQ
\label{eq:pide_con_int}
\begin{aligned}
\vsl\l(w, r, \cdot, \tau_{m+1}\r)
&= \iint_{\mathbb{R}^2} g\left(w - w', r - r', \Delta \tau\right)~\vhsl (w', r', \cdot, \tau_{m}^+)~dw'~dr',
&\qquad (w, r) \in \D.
\end{aligned}
\EN
The solution $\vsl\l(w, r, \cdot, \tau_{m+1}\r)$ for $(w, r) \not\in \D$
are given by the boundary conditions
\eqref{eq:vwmin_p2}, \eqref{eq:vmax_p2}, \eqref{eq:vother}.

For computational purposes, we truncate the infinite region of integration of \eqref{eq:pide_con_int}
to
\EQ
\label{eq:D_dagger}
\D^{\dagger} \equiv [w_{\min}^{\dagger}, w_{\max}^{\dagger}] \times [r_{\min}^{\dagger}, r_{\max}^{\dagger}],
\EN
where, for $x \in \{w, r\}$,  $x^{\dagger}_{\min}\ll x_{\min}<0<x_{\max}\ll x^{\dagger}_{\max}$
and $|x^{\dagger}_{\min}|$ and $x^{\dagger}_{\max}$ are sufficiently large.
This results in the approximation
\EQA
\label{eq:green_integral_truncated}
\vsl\l(w, r, \cdot, \tau_{m+1}\r) &\simeq&
\iint_{\D^{\dagger}}
g\left(w - w', r - r', \Delta \tau\right)~\vhsl (w', r', \cdot, \tau_{m}^+)~dw'~dr',
\quad
(w, r) \in \D.
\ENA
The error arising from this truncation is discussed in Section~\ref{section:conv}.

With the above discussion in mind, we define a finite domain
$\Omega = [w_{\min}^{\dagger}, w_{\max}^{\dagger}]  \times [r_{\min}^{\dagger}, r_{\max}^{\dagger}]
\times [a_{\min}, a_{\max}] \times [0, T]$, which consists of
\EQA
\label{eq:sub_domain_truncated_fin}
\Omega_{\myin}
& =& \text{defined in \eqref{eq:sub_domain_whole}},
\quad
\Omega_{a_{\min}}  = \text{defined in \eqref{eq:sub_domain_whole}},
\nonumber
\\
\Omega_{\tau_0}
&=& [w_{\min}^{\dagger}, w_{\max}^{\dagger}] \times [r_{\min}^{\dagger}, r_{\max}^{\dagger}] \times [a_{\min}, a_{\max}]  \times \{0\},
\nonumber
\\
\Omega_{w_{\min}} &=& [w_{\min}^{\dagger}, w_{\min}] \times (r_{\min}, r_{\max}) \times (a_{\min}, a_{\max}]  \times (0, T],
\nonumber
\\
\Omega_{wa_{\min}} &=& [w_{\min}^{\dagger}, w_{\min}] \times (r_{\min}, r_{\max}) \times \{a_{\min}\}  \times (0, T],
\nonumber
\\
\Omega_{w_{\max}} &=& [w_{\max}, w_{\max}^{\dagger}] \times (r_{\min}, r_{\max}) \times [a_{\min}, a_{\max}]  \times (0, T],
\nonumber
\\
\Omega_{\myot} &=& \Omega \setminus \Omega_{\myin} \setminus \Omega_{a_{\min}} \setminus \Omega_{w_{\max}} \setminus \Omega_{wa_{\min}} \setminus \Omega_{w_{\min}} \setminus \Omega_{\tau_0}.
\nonumber
\ENA

We stress that the region $\Omega_{w_{\min}} \cup \Omega_{wa_{\min}} \cup \Omega_{w_{\max}}\cup \Omega_{\myot}$ \blue{plays} an important role
in the proposed numerical method.
In particular, the convolution integral \eqref{eq:pide_con_int} is typically approximated
using efficient computation of an associated discrete \blue{convolution} via Fast-Fourier Transform (FFT).
It is well-documented that wraparound error (due to periodic extension) is an important issue for Fourier methods, particularly in the case of control problems (see, for example, \cite{online}).
Therefore, in \eqref{eq:sub_domain_truncated_fin}, the region $\Omega_{w_{\min}} \cup \Omega_{wa_{\min}}\cup \Omega_{w_{\max}}\cup \Omega_{\myot}$ is also set up to serve as padding areas for nodes in $\Omega_{\myin} \cup \Omega_{a_{\min}}$.
For this purpose, we assume that  $|w_{\min}|$, $w_{\max}$, $|r_{\min}|$ and $r_{\max}$ are chosen sufficiently large so that
\EQA
\label{eq:w_choice_green_jump_form}
w^{\dagger}_{\min} = w_{\min} - \frac{w_{\max} - w_{\min}}{2}
~~&\text{and}&~~
w^{\dagger}_{\max} =  w_{\max} + \frac{w_{\max} - w_{\min}}{2},
\nonumber
\\
r^{\dagger}_{\min} = r_{\min} - \frac{r_{\max} - r_{\min}}{2}
~~&\text{and}&~~
r^{\dagger}_{\max} =  r_{\max} + \frac{r_{\max} - r_{\min}}{2}.
\ENA
As elaborated in \cite{online}, this padding technique is
efficient in controlling wraparound error (also Remark~\ref{rm:wrap}).

Due to withdrawals, the non-local impulse operator $\mathcal{M}(\cdot)$ for $\Omega_{\myin}$, defined in \eqref{eq:Operator_M_b}, requires evaluating a candidate value
at point having $w = \ln(\max(e^w-\gamma, e^{\winf}))$ which could be smaller than $w^{\dagger}_{\min}$,
i.e.\ outside the finite computational domain, if $\winf < w^{\dagger}_{\min}$.
Therefore,  with $w^{\dagger}_{\min}$ (and $w^{\dagger}_{\max}$)
selected sufficiently large as above,   we set $\winf = w^{\dagger}_{\min}$.
That is, $\mathcal{M}(\cdot)$ in \eqref{eq:Operator_M_b} becomes
\EQA
\label{eq:Operator_M_b_trun}
\mathcal{M}v ({\bf{x}}) \equiv \mathcal{M}(\gamma) v ({\bf{x}})
= v\left(\ln(\max(e^w-\gamma, e^{w^{\dagger}_{\min}})), r, a -\gamma, \tau\right)
   + \gamma(1-\mu) -c, \quad{\bf{x}} \in \Omega_{\myin}.
\ENA
This is the intervention operator we use in $F_{\myin}$
for computation and  convergence analysis.

Finally, for a semi-Lagrangian discretization in the setting of HJB equations, common computational difficulties lie in the boundary areas, which typically require a special treatment of computational grids and boundary conditions \cite{Reisinger2017, Arto2018}. In our case, a semi-Lagrangian discretization is only applied in the sub-domain $\Omega_{\myin} \cup \Omega_{a_{\min}}$. It may require information from boundary sub-domains, such as $\Omega_{w_{\min}}$ and $\Omega_{w_{\max}}$, which is readily available from the numerical solutions in these boundary sub-domains. With $|r^{\dagger}_{\min}|$, $r^{\dagger}_{\max}$, $|w^{\dagger}_{\min}|$ and $w^{\dagger}_{\max}$ chosen large enough, we can ensure that a semi-Lagrangian discretization never requires information outside the computational domain $\Omega$.

\subsection{Discretization}
\label{section:discretization}
The computational grid is constructed as follows.
We denote by $N$ (resp. $N^{\dagger}$) the number of  points of an uniform partition of $[w_{\min}, w_{\max}]$
(resp. $[w_{\min}^{\dagger}, w_{\max}^{\dagger}]$).
For convenience,
we typically choose $N^{\dagger} = 2N$ so that only one set of $w$-coordinates is needed.
Also let  $P = w_{\max} - w_{\min}$, and $P^{\dagger} = w^{\dagger}_{\max} - w^{\dagger}_{\min}$.
%
We define $\Delta w = \frac{P}{N} = \frac{P^{\dagger}}{N^{\dagger}}$.
We use an equally spaced partition in the $w$-direction, denoted by
$\{w_n\}$, where
\EQA
\label{eq:grid_w}
    w_n &=& \hat{w}_0 + n\Delta w;
    ~~
    n ~=~ -N^{\dagger}/2, \ldots, N^{\dagger}/2, ~~\text{where}
    \\
    \Delta w &=& P/N ~=~ P^{\dagger}/N^{\dagger},~~\text{and}~~
    \hat{w}_0 ~=~ (w_{\min} + w_{\max})/2
   ~=~ (w^{\dagger}_{\min} + w^{\dagger}_{\max})/2.
    \nonumber
\ENA
Similarly, for the $r$-dimension,   with $K^{\dagger} = 2K$,
$Q = r_{\max} - r_{\min}$, and $Q^{\dagger} = r^{\dagger}_{\max} - r^{\dagger}_{\min}$,
we denote by $\{r_k\}$, an equally spaced partition in the $r$-direction, such that
\EQA
\label{eq:grid_r}
    r_k &=& \hat{r}_0 + k\Delta r;
    ~~
    k ~=~ -K^{\dagger}/2, \ldots, K^{\dagger}/2, ~~\text{where}
    \\
    \Delta r &=& Q/K ~=~ Q^{\dagger}/K^{\dagger},
    ~~\text{and}~~
    \hat{r}_0 ~=~ (r_{\min} + r_{\max})/2
   ~=~ (r^{\dagger}_{\min} + r^{\dagger}_{\max})/2.
    \nonumber
\ENA
We use an unequally spaced partition in
the $a$-direction, denoted by $\{a_j\}$, $j = 0,\ldots, J$,
with $a_0=a_{\min}$, and $a_{J}= a_{\max}$. We set
\EQA
\label{eq:grid_a}
\Delta a_{\max} = \max_{0 \leq j \leq J-1} \left(a_{j+1} - a_j\right), ~~
\Delta a_{\min} = \min_{0 \leq j \leq J-1} \left(a_{j+1} - a_j\right).
\ENA
We use the same previously defined equally spaced partition in the $\tau$-dimension with $\Delta \tau = T/M$ and $\tau_{m} = m\Delta \tau$, denoted by $\{\tau_m\}$, $m = 0, \ldots, M$.~\footnote{While it is straightforward to generalized the numerical method to non-uniform partitioning of the $\tau$-dimension, for the purposes of proving convergence, uniform partitioning suffices.}

At each time $\tau_m$, $m = 1, \ldots, M$,  we denote by $v_{n, k, j}^{m}$ an approximation
to the exact solution $v(w_n, r_k, a_j, \tau_m)$ at the reference node
$(w_n, r_k, a_j, \tau_m)$ obtained by our numerical method.
\blue{At time $\tau_m^+$, unless otherwise stated, $v_{n, k, j}^{m+}$
refers to an intermediate value, and not an approximation to the exact solution
at time $\tau_m^+$.}

For subsequent use, we  define the following index sets for the spatial and temporal variables:
\\
$\N = \left\{-N/2+1, \ldots, N/2- 1\right\}$, $\ND = \left\{-N^{\dagger}/2, \ldots, N^{\dagger}/2- 1\right\}$,
$\K = \left\{-K/2+1, \ldots, K/2- 1\right\}$,
\\
$\KD = \left\{-K^{\dagger}/2, \ldots, K^{\dagger}/2- 1\right\}$,
$\J = \left\{0, \ldots, J \right\}$ and $\M = \left\{0, \ldots, M-1 \right\}$,
$\Nl = \left\{-N^{\dagger}/2, \ldots, -N/2\right\}$,
$\Nr = \left\{N/2, \ldots, N^{\dagger}/2- 1\right\}$,
$\Nc = \ND \setminus \N$, and  $\Kc = \KD \setminus \K$.
For fixed $j\in \J$ and $m \in \M$, nodes ${\bf{x}}_{n, j}^{m+1}$ having
(i) $n \in \Nl$ and $k \in \K$ are in $\Omega_{w_{\min}} \cup \Omega_{wa_{\min}}$,
(ii) $n \in \N$ and $k \in \K$ are in $\Omega_{\myin} \cup  \Omega_{a_{\min}}$,
(iii)~$n \in  \Nr$ and $k \in \K$ are in  $\Omega_{w_{\max}}$,
and (iv) $n \in  \ND$ and $k \in \Kc$ are in  $\Omega_{\myot}$.

In subsequent discussion, we denote by $\gamma_{n, k, j}^{m} \in [0, a_j]$ the control representing the withdrawal amount at node $(w_n, r_k, a_j, \tau_{m})$, $n \in \Nl \cup \N$, $k \in \K$, $j \in \J$,
$m \in \M$. 
We also define
\EQ
\label{eq:def_wa}
\tilde{w}_n = \ln(\max(e^{w_n} - \gamma_{n, j,k}^{m}, e^{w^{\dagger}_{\min}})),
\quad
\tilde{a}_j = a_j - \gamma_{n, k, j}^{m}, \quad \gamma_{n, k, j}^{m} \in [0, a_j].
\EN
For a given withdrawal amount $\gamma$, let $f\left(\gamma\right)$ be
the cash amount received by the holder defined as follows
\EQA
\label{eq:f_gamma_k_dis}
f\left(\gamma\right)
= \left\{
\begin{array}{ll}
\gamma & \text{if }~ 0\le \gamma \le C_r\Delta \tau,
\\
\gamma(1 - \mu) + \mu C_r\Delta \tau - c & \text{if }~ C_r\Delta \tau < \gamma.
\end{array}
\right.
\ENA
\begin{remark}[Interpolation]
\label{eq:intp}
Optimal controls are typically decided by comparing candidates obtained via interpolation using
on available relevant discrete values in $\Omega$, i.e.\ including discrete values
are in boundary sub-domains. In this work, we use linear interpolation.
To this end, let $s \in (0, T]$ be fixed.
We denote by $\mathcal{I}\left\{u^{s}\right\}(w, r, a)$
a generic three-dimensional linear interpolation operator acting on the time-$s$
discrete values $\left\{\left(\left(w_l, r_d, a_q \right), u_{l, d, q}^{s}\right)\right\}$,
$l \in \ND$, $d \in \KD$, $q \in \J$. Here, unless otherwise stated, values $u_{l, d, q}^{s}$ corresponding
to points ${\bf{x}}_{l, d, q}^{s}$ in the boundary sub-domains $\Omega_{w_{\min}}$, $\Omega_{wa_{\min}}$,  $\Omega_{w_{\max}}$ or  $\Omega_{\myot}$ are given by the respective time-$s$ boundary values.

In its primary usage, the above interpolation operator degenerates
to a two- or one-dimensional operator respectively when one or two of the following equalities hold:
$w = w_n$, $r = r_k$, and $a = a_j$,  for some $n \in \ND$, $k \in \K$, and $j \in \J$.
Nonetheless, in these cases, to simplify notation, we still use the notation
$\mathcal{I}\left\{u^{s}\right\}(w, r, a)$, with these degenerations being implicitly
understood.

It is straightforward to show that, due to linear interpolation, for any
constant $\xi$, we have
\EQ
\label{eq:interp_xi}
\mathcal{I}\left\{\varphi^{s}+\xi\right\}(w, r, a)
= \mathcal{I}\left\{\varphi^{s}\right\}(w, r, a) + \xi.
\EN
Furthermore, for a smooth test function $\varphi \in \C{\Oinf}$, we have
\EQ
\label{eq:interp_sim}
\mathcal{I}\left\{\varphi^{s}\right\}(w, r, a) = \varphi(w, r, a)
+ \mathcal{O}\l(\l(\Delta w + \Delta r\r)^2\r).
\EN
Finally, we note that linear interpolation is monotone in the viscosity sense.
\end{remark}

\noindent For double summations, we use the short-hand notation:
$\mysum_{d \in \mathcal{D}}^{q \in \mathcal{Q}}(\cdot):=\sum_{d \in \mathcal{D}} \sum_{q \in \mathcal{Q}} (\cdot)$,
unless otherwise noted.
We are now ready to present the complete numerical schemes to solve the HJB-QVI~\eqref{eq:gmwb_def}.
For any point $(w_n, r_k, a_j, \tau_{m+1})$ in $\Omega$, unless otherwise stated, we let $j \in \J$ and $m \in \M$ be fixed, and focus on the index sets of $n$ and $k$ in subsequent discussion.

\subsection{\texorpdfstring{$\boldsymbol{\Omega_{\tau_0}}$, $\boldsymbol{\Omega_{w_{\max}}}$, and $\boldsymbol{\Omega_{\myot}}$}{Omega-tau0, Omega-wmax, Omega-ot}}
For $(w_n, r_k, a_j, \tau_{0}) \in \Omega_{\tau_0}$, we impose the initial condition \eqref{eq:ftau0}.
\EQ
\label{eq:terminal}
v_{n,k,j}^0 = \max(e^{w_n}, (1-\mu)a_j -c), \quad n \in \ND,~ k \in \KD.
\EN
For $(w_n, r_k, a_j, \tau_{m+1})$ in $\Omega_{w_{\max}}$ and $\Omega_{\myot}$,
we respectively apply the Dirichlet boundary condition \eqref{eq:vmax_p2} and \eqref{eq:vother} as follows
\begin{linenomath}
\postdisplaypenalty=0
\begin{align}
\label{eq:omega_max}
v_{n,k,j}^{m+1} &= e^{-\beta \tau_{m+1}} e^{w_n}, \quad n \in \Nr,
~k \in \K,
\\
v_{n,k,j}^{m+1} &= \z(w_n, r_k, a_j, \tau_{m+1}), \quad n \in \ND,
~k \in \Kc,
\label{eq:omega_ot}
\end{align}
\end{linenomath}
where $\z(w_n, r_k, a_j, \tau_{m+1})$ is given in \eqref{eq:vother}.

\subsection[Left padding areas]{$\boldsymbol{\Omega_{w_{\min}} \cup \Omega_{wa_{\min}}}$}
For $(w_n, r_k, a_j, \tau_{m+1})$ in $\Omega_{w_{\min}} \cup \Omega_{wa_{\min}}$,
we let $\tilde{v}_{n,k,j}^{m}$ be an approximation to $v(w_n, r_k,  a_j - \gamma_{n, j}^{m}, \tau_{m})$ computed by linear interpolation as follows
\EQA
\label{eq:vtil_a}
    \tilde{v}_{n,k,j}^{m} =
    \mathcal{I}\left\{v^{m} \right\}
    \left(w_n, r_k,
    a_j - \gamma_{n, k, j}^{m} \right),
    \quad
    n \in \Nl,~
    k \in \K.
\ENA
We compute intermediate results $v_{n,k,j}^{m+}$ by solving the optimization problem
\EQA
\label{eq:Intervention_Operator_wmin}
v_{n,k,j}^{m+} &=& \sup_{\gamma_{n,k,j}^{m} \in [0, a_j]} \left( \tilde{v}_{n,k,j}^{m}
+ f\left(\gamma_{n,k,j}^{m}\right)\right),  \quad
    n \in \Nl,~
    k \in \K.
\ENA
where $\tilde{v}_{n,k,j}^{m}$ is given in \eqref{eq:vtil_a}
and $f\left(\cdot\right)$ is defined in \eqref{eq:f_gamma_k_dis}.
To advance to time $\tau_{m+1}$, we solve the PDE $v_\tau - \Ld v = 0$ with the time-$\tau_{m+}$ initial condition given by $v_{n, k, j}^{m+}$ in \eqref{eq:Intervention_Operator_wmin}.
This step is achieved by applying finite difference methods built upon a fully implicit timestepping scheme together with a positive coefficient discretization
as follows \cite{chen08a, chen06,  huang:2010a, DangForsyth2014, Forsyth08}
\EQA
v_{n,k,j}^{m+1}  &=&  v_{n,k,j}^{m+} + \Delta \tau
(\Ld^h v )_{n,k,j}^{m+1}, \quad \text{where}
\label{eq:scheme_timestep_left}
\\
(\Ld^h v )_{n,k,j}^{m+1} &=& \alpha_k v_{n, k-1, j}^{m+1}
+ \beta_k v_{n, k+1, j}^{m+1} - \left( \alpha_k + \beta_k + r_k \right) v_{n, k, j}^{m+1},
\quad
n \in \Nl,~
k \in \K,~
\nonumber
\\
&&\text{with}\quad \alpha_k ~\geq~ 0, \quad \beta_k ~\geq ~0, \quad k \in \K.
\label{eq:pos_con}
\ENA

\subsection[Main areas: scheme]{$\boldsymbol{\Omega_{\myin} \cup  \Omega_{a_{\min}}}$}
For $(w_n, r_k, a_j, \tau_{m+1})$ in $\Omega_{\myin} \cup  \Omega_{a_{\min}}$
and $\gamma_{n, k,j}^{m} \in [0, a_j]$, we let $\vt_{n,k,j}^{m}$ be an approximation to
$v(\tilde{w}_n, r_k, \tilde{a}_j, \tau_m)$, where $\tilde{w}_n$ and $\tilde{a}_j$ are defined in \eqref{eq:def_wa},
computed by linear interpolation given by
\begin{linenomath}
\postdisplaypenalty=0
\begin{align}
\vt_{n,k,j}^{m} &=
\mathcal{I}\left\{v^{m}\right\}
   \left(\tilde{w}_n, r_k, \tilde{a}_j\right),
        \quad
        \gamma_{n, k,j}^{m} \in [0, a_j],~
    n \in \N,~
    k \in \K.
    \label{eq:vtil_b}
\end{align}
\end{linenomath}
We recall the control formulation \eqref{eq:omega_inf_all}, where the admissible control set is $[0, a]$.
We observe that the $\min\{\cdot\}$ operator of \eqref{eq:omega_inf_all} contains two terms, with the continuous control $\hat{\gamma}$ in the first term having a local nature ($\hat{\gamma}\in [0, C_r]$), while the impulse control $\gamma$ in the second term having a non-local nature ($\gamma\in [0, a]$). Motivated by this observation, as in \cite{online, chen08a}, with the convention that $(C_r \Delta \tau, a_j] = \emptyset$ if $a_j \le C_r \Delta \tau$,
we  partition $[0, a_j]$ into $[0, a_j \wedge C_r\Delta\tau]$ and $(C_r\Delta\tau, a_j]$,
where $x\wedge y = \min(x, y)$.
We compute respective intermediate results $\vl_{n, k, j}^{m+}$ and $\vn_{n, k, j}^{m+}$,
$n \in \N$, $k \in \K$,  by solving the optimization problems
\begin{align}
\label{eq:scheme*}
\vl_{n,k,j}^{m+} =
\sup_{\gamma_{n, k, j}^{m} \in[0, a_j\wedge C_r\Delta\tau]}
(\vt_{n,k,j}^{m} + f(\gamma_{n,k,j}^{m})),\quad
\vn_{n,k,j}^{m+} =
\sup_{\gamma_{n, k, j}^{m} \in(C_r\Delta\tau, a_j]}
(\vt_{n,k,j}^{m} + f(\gamma_{n,k,j}^{m})),
\end{align}
where  $\vt_{n,k,j}^{m}$ is given in \eqref{eq:vtil_b} and
$f(\cdot)$ is defined in \eqref{eq:f_gamma_k_dis}.
\begin{remark}[Attainability of supremum]
\label{rm:sup_exist}
It is straightforward to show that, due to boundedness of nodal values
used in $\mathcal{I}\left\{v^{m}\right\}(\cdot)$ (see Lemma~\ref{lemma:stability} on stability),
the interpolated value $\vt_{n,k,j}^{m}$ in \eqref{eq:vtil_b} is uniformly continuous in $\gamma_{n, k, j}^m$.
As a result, the supremum in the discrete equations for $\vl_{n, k, j}^{m+}$  and $\vn_{n, k, j}^{m+}$
in \eqref{eq:scheme*} can be achieved by a control in $[0, \min(a_j, C_r \Delta \tau)]$
and $(C_r \Delta \tau, a_j]$, respectively, with the latter case being made possible due to $c > 0$ \cite{chen08a}.
\end{remark}

The next step in the numerical scheme for $\Omega_{\myin} \cup  \Omega_{a_{\min}}$
is time advancement from $\tau_{m}^+$ to $\tau_{m+1}$.  As briefly discussed previously, the time advancement step involves (i) \blue{an} SL discretization
for the  term $\Ls v - rv$ of the PIDE \eqref{eq:dis_pide_ori} in
$\Omega_{\myin} \cup \Omega_{a_{\min}}$, (ii) an $\epsilon$-monotone Fourier method based on the Green function associated with the PIDE \eqref{eq:dis_pide}. We now discuss these steps in detail below.

\subsubsection{Intuition of semi-Lagrangian discretization}
\label{ssc:sld}
\noindent We start by providing an intuition of \blue{an} SL discretization method and the Green's function approach utilized
for $\Omega_{\myin} \cup \Omega_{a_{\min}}$.
The main idea employed to construct \blue{an} SL discretization of the PIDE of the form \eqref{eq:dis_pide_ori}
is to integrate the PIDE along \blue{an} SL trajectory, which is to be defined subsequently.
Recall from \eqref{eq:operator_Lgs} that the differential operator $\mathcal{L}$ in the PIDE \eqref{eq:dis_pide_ori}
can be written as $\mathcal{L}~=~\Lg +\Ls - rv$, where the operator $\Ls = (r -  \frac{\sigz^{2}}{2} - \beta) v_w + \delta(\theta - r) v_r$. In subsequent discussion, 
we let $a \in [a_{\min}, a_{\max}]$ be fixed, and also let $x := (w, r)$ be arbitrary in $[w_{\min}, w_{\max}] \times [r_{\min}, r_{\max}]$.
For any $s \in [\tau_{m}^+, \tau_{m+1}]$,
and $\tau \le s$, we consider \blue{an} SL trajectory, denoted by
$\mychi(\tau; s, x) =  \l(\mychi_1(\tau; s, x), \mychi_2(\tau; s, x)\r)$,
which satisfies the ordinary differential equations
\EQA
\ds
\label{eq:trajectory_ode*}
\left\{
\begin{array}{ll}
\ds
\frac{\partial \mychi_1(\tau; s, x)}{\partial \tau} = -(r - \frac{\sigz^2}{2} -\beta),
& \tau < s,
\\
\mychi_1(s; s, x)= w,
& \tau = s,
\end{array}
\right.
~\text{and}~
\left\{
\begin{array}{ll}
\ds
\frac{\partial \mychi_2(\tau; s, x)}{\partial \tau} = - \delta(\theta - r),
& \tau < s,
\\
\mychi_2(s; s, x)= r,
& \tau = s.
\end{array}
\right.
\ENA
Using \eqref{eq:trajectory_ode*}, we have $\frac{D v}{D \tau} = v_\tau + \Ls v$,
and therefore,  the PIDE \eqref{eq:dis_pide_ori} can be written as
\EQA
\label{eq:SL_pide}
\frac{D v}{D \tau} + rv - \Lg v - \mathcal{J}v  = 0, \quad \tau \in (\tau_{m}^+, \tau_{m+1}],
\ENA
subject to a generic initial condition of the form \eqref{eq:dis_pide_term}.
We let $(\bw(s), \br(s))$ be the $(w, r)$-departure point at time-$\tau_{m}$ for the trajectory
$\mychi(\tau; s, x)$, i.e.\
$(\bw(s), \br(s)) = (\mychi_1(\tau = \tau_{m}; s, x), \mychi_2(\tau = \tau_{m}; s, x))$, and hence,
they can be computed by solving \eqref{eq:trajectory_ode*} from $\tau = \tau_{m}$ to $\tau = s$,
i.e.\
\EQ
\label{eq:semi-lag-location*}
\breve{w}(s) =  w + r( e^{s - \tau_{m}}-1) - \left(\frac{\sigz^2}{2} + \beta \right) ( e^{s - \tau_{m}}-1),
~
\breve{r}(s) =  r e^{-\delta (s - \tau_{m})} - \theta \left( e^{-\delta (s - \tau_{m})} - 1 \right).
\EN
We then integrate both sides of the equation \eqref{eq:SL_pide} along the
trajectory $\mychi(\tau; s, x)$ from $\tau = \tau_{m}$ to $\tau = s$
with $a$ being fixed. This gives
\EQA
\label{eq:SL_pide*}
\int_{\tau_{m}}^{s} \l(\frac{D v}{D \tau}\l(\mychi(\tau; s, x), a, \tau\r)
 +  r v\l(w, r, a, \tau\r) - \l(\Lg + \mathcal{J}\r)v\l(w, r, a, \tau\r)\r)
 d\tau
 = 0.
\ENA
In \eqref{eq:SL_pide*}, using the identity
\EQAS
\int_{\tau_{m}}^{s} \frac{D v}{D \tau}\l(\mychi(\tau; s, x), a, \tau\r)  d\tau =
v\l(w, r, a, s\r)
-
v\l(\bw(s), \br(s),  a, \tau_{m}\r),
\ENAS
together with a simple left-hand-side rule for
$\int_{\tau_{m}}^{s} rv\l(w, r, a, \tau\r) d\tau \simeq r(s-\tau_{m}) v\l(w, r, a, \tau_{m}\r)$,
and rearranging, \eqref{eq:SL_pide*} becomes
\EQ
\label{eq:inte}
v\l(w, r, a, s\r)
-
\int_{\tau_{m}}^{s}
 \l(\Lg  + \mathcal{J}\r)v\l(w, r, a, \tau\r)
 d\tau
=
v\l(\bw(s), \br(s),  a, \tau_{m}\r)
-
r(s-\tau_{m}) v\l(w, r, a, \tau_{m}\r).
\EN
Here, $v\l(w, r, a, s\r)$, $\tau_{m} \le s \le \tau_{m+1}$, is the unknown function at time-$s$.
In particular, we are interested in finding $v\l(w, r, a, \tau_{m+1}\r)$.
To this end, we approximate $v\l(w, r, a, \tau_{m+1}\r)$ by $\vsl\l(w, r, a, \tau_{m+1}\r)$
where the function $\vsl\l(w, r, a, s\r)$, $\tau_{m} \le s \le \tau_{m+1}$,
satisfies a variation of equation~\eqref{eq:inte} obtained by fixing its right-hand-side
at $s = \tau_{m+1}$. More specifically, with $(\bw, \br) \equiv (\bw(\tau_{m}^+), \br(\tau_m^+)$,
$\vsl\l(w, r, a, s\r)$ satisfies
\EQ
\label{eq:inte_2}
\vsl\l(w, r, a, s\r)
-
\int_{\tau_{m}}^{s}
 \l(\Lg  + \mathcal{J}\r)\vsl\l(w, r, a, \tau\r)
 d\tau
=
v\l(\bw, \br,  a, \tau_{m}\r)
-
r \Delta \tau v\l(w, r, a, \tau_{m}\r),
\EN
where, on the rhs, $v\l(\cdot, \cdot,  a, \tau_{m}\r)$ is given by a known generic initial condition
at time $\tau_{m}$.
We highlight that equation~\eqref{eq:inte} agrees with equation~\eqref{eq:inte_2} only when $s = \tau_{m+1}$,
at which time we have $\vsl\l(w, r, a, \tau_{m+1}\r) = v\l(w, r, a, \tau_{m+1}\r)$, as wanted.

The form of equation~\eqref{eq:inte_2} suggests that $\vsl\l(w, r, a, s\r)$
satisfies the PIDE of the form \eqref{eq:dis_pide}, i.e.\
\EQ
\label{eq:dis_pide_sl}
(\vsl)_{\tau} - \Lg \vsl  - \mathcal{J} \vsl = 0,
\quad
 w \in (-\infty, \infty),~r \in (-\infty, \infty),~\tau \in (\tau_{m}^+ , \tau_{m+1}],
\EN
subject to the initial condition
\begin{linenomath}
\begin{subequations}\label{eq:LG_pide_term}
\begin{empheq}[left={\vhsl(w, r, a, \tau_{m}^+o) = \empheqlbrace}]{alignat=3}
&v\l(w, r, a, \tau_{m}^+\r) = \frac{v\l(\bw, \br,  a, \tau_{m}^+\r)}{ 1 + \Delta \tau r}
&& \qquad (w, r, a, \taus) \in \Omega_{\myin} \cup \Omega_{a_{\min}},
\label{eq:LG_pide_term_a}
\\
&v_{bc}(w, r, a, \tau_{m}^+)&& \qquad (w, r, a, \taus) \in
\Omega \setminus \left( \Omega_{\myin} \cup \Omega_{a_{\min}} \right),
\label{eq:LG_pide_term_b}
\end{empheq}
\end{subequations}
\end{linenomath}
where, in \eqref{eq:LG_pide_term_a},
$(\bw, \br) \equiv (\bw(\tau_{m+1}), \br(\tau_{m+1}))$ given by \eqref{eq:semi-lag-location*}.
From here, as previously discussed in Subsection~\ref{ssec:over}, the solution
$\vsl\l(w, r, \cdot, \tau_{m+1}\r)$ is approximated by the convolution integral
\eqref{eq:green_integral_truncated}.


For subsequent discussions, we investigate equation \eqref{eq:inte_2} and the initial condition \eqref{eq:LG_pide_term}
from a standpoint that involves discrete grid points.
Specifically, for a Lagrangian trajectory which ends at
$(w_n, r_k)$ at time $\tau_{m+1}$, the departure point $(\bw_n, \br_k)$ at time-$\tau_m^+$, computed by
$\eqref{eq:semi-lag-location*}$  with $w  = w_n$, $r = r_k$, and $s = \tau_{m+1}$,
does not necessarily coincide with a grid point.
Therefore,
to approximate \eqref{eq:LG_pide_term_a}
corresponding to $(w_n, r_k, a_j)$, i.e.\ $\frac{v\l(\bw_n, \br_k,  a_j, \tau_{m}^+\r)}{ 1 + \Delta \tau r}$,
linear interpolation can be used.
Specifically,  we denote by $\bvsl^{m+}_{n,k,j}$ the interpolation result given by
\EQA
\label{eq:semi_lag_results}
    && \bvsl_{n,k,j}^{m+} =
    \frac{\mathcal{I}\left\{v^{m+} \right\}
    \left(\bw_n, \br_k, a_j \right)}{1 + \Delta \tau r_k},
    \quad
    n \in \N, ~k \in \K,~
    \\
    &&
    \text{where } \bw_n =  w_n + r_k \left( e^{\Delta \tau}-1\right) - \left(\frac{\sigz^2}{2} + \beta \right) \left( e^{\Delta \tau}-1\right),
    ~
    \br_k =  r_k e^{-\delta \Delta \tau} - \theta \left( e^{-\delta \Delta \tau} - 1 \right).
    \nonumber
\ENA
Here, $\mathcal{I}\left\{\cdot\right\}$ is the discrete interpolation operator
defined in \eqref{eq:intp}. If the departure point $(\bw_n, \br_k, a_j)$ falls outside $\Omega_{\myin} \cup \Omega_{a_{\min}}$, discrete solutions in the boundary sub-domains are used for interpolation.
We emphasize the SL discretization is not applied to grid points outside $\Omega_{\myin} \cup \Omega_{a_{\min}}$.

\subsubsection{Time advancement scheme: ${\boldsymbol{\tau \in [\tau_{m}^+, \tau_{m+1}]}}$}
\noindent
To prepare for time advancement, we combine the time-$\tau_{m}$ boundary values in $\Omega_{w_{\min}}$, $\Omega_{wa_{\min}}$, $\Omega_{w_{\max}}$, and $\Omega_{\myot}$ with the time-$\tau_{m}^+$
intermediate results obtained by the SL discretization discussed above and results from \eqref{eq:scheme*}.
With a slight abuse of notation, for $(i) \in  \{(1), (2)\}$, this is done as follows
\EQA
\label{eq:v_sl_ext}
\begin{array}{l}
\bvisl^{m+}_{l,d,j}
\end{array}
= \left\{
\begin{array}{llll}
\ds \frac{\mathcal{I}\left\{(v^{\mysup})^{m+} \right\}
    \left(\bw_l, \br_d, a_j \right)}{1 + \Delta \tau r_d} & \text{$\bw_l$ and $\br_d$ defined in \eqref{eq:semi_lag_results}} & l \in \N ~\text{and}~ d \in \K,
\\
v_{l,d,j}^{m} &\text{in \eqref{eq:omega_max}, \eqref{eq:omega_ot}, and \eqref{eq:scheme_timestep_left}},
& \text{otherwise}.
\end{array}
\right.
\ENA
For $\tau \in [\tau_{m}^+, \tau_{m+1}]$, our timestepping method for solving the PIDE \eqref{eq:dis_pide_sl} is built upon the convolution integral \eqref{eq:pide_con_int}, with the initial condition
$\hvisl(w, r, \cdot, \tau_{m}^+)$, $(i) \in \{(1), (2)\}$, approximated by
a projection of discrete values in \eqref{eq:v_sl_ext}.
onto linear basis functions for the $w$- and $r$-dimensions.
Specifically, $\hvisl\left(w, r, \cdot, \tau_{m}^+\right)$, $(i) \in  \{(1), (2)\}$,  is approximated by the projection
\EQA
\label{eq:rhowfunc}
\hvisl\left(w, r, \cdot, \tau_{m}^+\right) \simeq \mysum_{l \in \ND}^{d \in \KD} \varphi_l(w)~ \psi_d(r)
                ~ \bvisl^{m+}_{l,d,j},
\quad
(w, r)  \in \D \equiv (w_{\min}, w_{\max}) \times (r_{\min}, r_{\max}),
\ENA
where $\{\varphi_l(w)\}_{l \in \ND}$ and $\{\psi_{d}(r) \}_{d \in \KD}$  are piecewise linear basis functions defined by
\EQA
\varphi_{l}(w) =
\left \{
\begin{array}{lll}
(w_{l+1} - w)/\Delta w, & w_{l}\le w \le w_{l+1},
\\
(w - w_{l-1})/\Delta w, & w_{l-1}\le w \le w_{l},
\\
0,  & \text{otherwise},
\end{array}
\right.
\psi_{d}(r) =
\left \{
\begin{array}{lll}
(r_{d+1} - r)/\Delta r, & r_{d}\le r \le r_{d+1},
\\
(r - r_{d-1})/\Delta r, & r_{d-1}\le r \le r_{d},
\\
0,  & \text{otherwise}.
\end{array}
\right.
\label{eq:piecewise_w}
\ENA
In the convolution integral \eqref{eq:green_integral_truncated},
we substitute $\hvisl (w, r, \cdot, \tau_{m}^+)$,
$(i) \in  \{(1), (2)\}$, by the projection \eqref{eq:rhowfunc} and rearrange the resulting equation.
We obtain the discrete convolution for $\bvisl_{n,k,j}^{m+1}$,  $(i) \in  \{(1), (2)\}$, as follows
\EQ
\label{eq:V_nG}
\bvisl_{n,k,j}^{m+1} = \Delta w \Delta r
\mysum_{l \in \ND}^{d \in \KD} \tilde{g}_{n-l, k-d}~\bvisl^{m+}_{l,d,j},
\quad n \in \N,~ k \in \K.
\EN
Here, $\bvisl^{m+}_{l,d,j}$ is given by the linear interpolation in  \eqref{eq:semi_lag_results},
and $\tilde{g}_{n-l, k-d}$ is given by
\EQA
\label{eq:Gtilde_ell}
\tilde{g}_{n-l, k-d} &\equiv& \tilde{g}(w_n-w_l, r_k-r_d, \Delta \tau)
\nonumber
\\
&=& \frac{1}{\Delta w} \frac{1}{\Delta r} \iint_{\D^{\dagger}}~ \varphi_{l}(w)~\psi_{d}(r)~g(w_n - w, r_k - r, \Delta \tau)~dw~dr.
\ENA
That is, in the discrete convolution \eqref{eq:V_nG}, the exact weights $\tilde{g}_{n-l, k-d}$,
$n \in \N$, $k \in \K$, $l \in \ND$, $d \in \KD$, are obtained by a projection
of the Green's function $g\left(\cdot, \Delta \tau\right)$ onto the piecewise linear basis functions
$\{\varphi_l(w)\}_{l \in \ND}$ and $\{\psi_{d}(r) \}_{d \in \KD}$.

Finally, we compute the discrete solution $v_{n,k,j}^{m+1}$ by
\begin{linenomath}
\postdisplaypenalty=0
\begin{align}
\label{eq:scheme}
v_{n,k,j}^{m+1} =
\max\l(
\bvlsl_{n,k,j}^{m+1},
\bvnsl_{n,k,j}^{m+1}\r)
\quad
n \in \N,
~
k \in \K,
\end{align}
\end{linenomath}
where  $\bvlsl_{n,k,j}^{m+1}$ and $\bvnsl_{n,k,j}^{m+1}$ are given by \eqref{eq:V_nG}.

\subsubsection{Approximation of exact weights ${\boldsymbol{\tilde{g}}}$ and ${\boldsymbol{\epsilon}}$-monotonicity}
We need to approximate the exact weights $\tilde{g}_{n-l, k-d}$ defined in the convolution integral
\eqref{eq:Gtilde_ell}. To this end, we adapt steps in \cite{ForsythLabahn2017, online} for
two-dimensional Green's functions. We let $G \left( \eta, \xi, \Delta \tau \right)$ be
the Fourier transform of the Green's function $g(w,r, \Delta \tau)$.
A closed-form expression for $G \left( \eta, \xi, \Delta \tau \right)$ is given by
\EQA
\label{eq:small_g}
G(\eta, \xi, \Delta \tau) &=&
\exp \left( \Psi \left( \eta, \xi \right) \Delta \tau \right), \quad \text{with}
\nonumber
\\
\Psi(\eta, \xi) &=& - \frac{ \sigz^2}{2} (2 \pi \eta)^2  - \rho \sigz \sigr (2 \pi \eta) (2 \pi \xi) - \frac{\sigr^2}{2} (2 \pi \xi)^2  - \lambda \kappa (2 \pi i \eta) - \lambda + \lambda \overline{B} (\eta),
\ENA
where, $\overline{B}(\eta)$ is the complex conjugate of the integral $B (\eta) = \int_{-\infty}^{\infty} b(y) e^{-2 \pi i \eta y}~dy$, noting $b(y)$ is the density function of $\ln(Y)$, where $Y$ is the random variable representing the jump multiplier.

The idea in approximating the integral \eqref{eq:Gtilde_ell} is to replace $g(w, r, \Delta \tau)$ therein by its localized, periodic approximation $\hat{g}(w,  r, \Delta \tau)$ given by
\EQA
\label{eq:green_hat}
\hat{g}(w,r,\Delta \tau) ~=~ \frac{1}{P^{\dagger}} \frac{1}{Q^{\dagger}}  \mysum^{z \in \Z}_{s \in \Z}  e^{2\pi i \eta_s w} e^{2\pi i \xi_z r} G(\eta_s, \xi_z, \Delta \tau) ~~\text{with}~~\eta_s = \frac{s}{P^{\dagger}},~ \xi_z = \frac{z}{Q^{\dagger}}.
\ENA
where we denote $\mathbb{Z}$ to be the set of all integers.\footnote{We note that the coefficients $G(\eta_s, \xi_z \Delta \tau)$ in \eqref{eq:green_hat} are the exact coefficients corresponding to the Green's function of the PIDE \eqref{eq:dis_pide} with suitable periodic boundary conditions; hence, $\hat{g}(w,r,\Delta \tau)$ is a valid Green's function, and in particular $\hat{g}(\cdot) \ge 0$.} Then, assuming uniform convergence of Fourier series, we integrate \eqref{eq:Gtilde_ell} to obtain
\EQA
\label{eq:gtilde_infty}
\tilde{g}_{n-1,k-d}
\equiv \tilde{g}_{n-1,k-d} (\infty)
=
\frac{1}{P^{\dagger}} \frac{1}{Q^{\dagger}}
\mysum^{z \in \Z}_{s \in \Z}
e^{2\pi i \eta_s (n-l) \Delta w}
e^{2\pi i \xi_z (k-d) \Delta r}
~
\text{tg}(s,z)~
G(\eta_s, \xi_z, \Delta \tau),
\ENA
where the trigonometry term $\text{tg}(s,z)$ is defined by\footnote{For $\eta_s = 0$ and $\xi_z = 0$, we take the limit $\eta_s \to 0$ and $\xi_z \to 0$.}
\EQA
\label{eq:trigofun}
	\text{tg}(s,z) = \left(
\frac{\sin^2 \pi \eta_s \Delta w}{\left(\pi \eta_s \Delta w\right)^2}
\right) \left(
\frac{\sin^2 \pi \xi_z \Delta r}{\left(\pi \xi_z \Delta r\right)^2}
\right), \quad s \in \Z, z \in \Z.
\ENA
For $\alpha \in \{2, 4, 8, \ldots\}$, \eqref{eq:gtilde_infty} is truncated to
$\alpha N^{\dagger}$ and $\alpha K^{\dagger}$ terms for the outer and the inner summations, respectively,
resulting in an approximation
\EQA
\label{eq:gtilde_truncated}
\tilde{g}_{n-l,k-d} \left(\alpha \right)
= \frac{1}{P^{\dagger}} \frac{1}{Q^{\dagger}}
\mysum^{z \in \Ka}_{s \in \Na}
e^{2\pi i \eta_s (n-l) \Delta w}
e^{2\pi i \xi_z (k-d) \Delta r}
~
\text{tg}(s,z)~
{\numPDEblue{G(\eta_s, \xi_z, \Delta \tau),}}
\ENA
where $\Na = \{ -\alpha N^{\dagger}/2 -1, \ldots, \alpha N^{\dagger}/2 -1 \}$
and $\Ka = \{ -\alpha K^{\dagger}/2 -1, \ldots, \alpha K^{\dagger}/2 -1 \}$.{\footnote{ We can use different numbers of terms
in the \blue{truncation} for the outer and the inner summations, i.e.\ $\alpha_1 N^{\dagger}$ and $\alpha_2 K^{\dagger}$,
respectively. Here, we use a single $\alpha$ to simplify the presentation.}}

As $\alpha \to \infty$, replacing $\tilde{g}_{n-l,k-d}$ by $\tilde{g}_{n-l,k-d} \left(\alpha \right)$ in the
discrete convolution \eqref{eq:V_nG} results in no loss of information.
However, for any finite $\alpha$, there is an error due to the use of a truncated Fourier series,
although, as $\alpha \to \infty$, this error vanishes very quickly due a rapid convergence of
truncated Fourier series. This is discussed in Subsection~\eqref{ssc:error_analysis}.
Due to the above truncation error of Fourier series, strict monotonicity is not guaranteed for a finite $\alpha$.
To control this potential loss of monotonicity for a finite $\alpha$,
as in \cite{ForsythLabahn2017, online}, the selected $\alpha$ must satisfy
\EQA
\label{eq:test1}
\Delta w  \Delta r  \mysum^{d \in \KD}_{l \in \ND}   \big|\min\left(
\tilde{g}_{n-l, k-d}(\alpha),
0\right)\big| ~<~ \epsilon \frac{\Delta \tau}{T},
\quad
\forall n \in \N, ~ k \in \K,
\ENA
where  $0 < \epsilon \ll 1/2$ is an user-defined monotonicity tolerance.
We note that the left-hand-side of the monotonicity test \eqref{eq:test1} is
scaled by $\Delta w$ so that it is bounded as $\Delta w, \Delta \tau \to 0$.
In addition, $\epsilon$ is scaled by $\frac{\Delta \tau}{T}$ in order to
eliminate the number of timesteps from the bounds of potential loss of monotonicity.

\subsubsection{Efficient implementation via FFT and algorithms}
Note that, for a fixed  $\alpha \in \{2, 4, 8, \ldots\}$,
the sequence $\{\tilde{g}_{-N^{\dagger}/2, k}(\alpha),\ldots, \tilde{g}_{N^{\dagger}/2-1, k}(\alpha)\}$
for a fixed $k \in \KD$ is $N^{\dagger}$-periodic, and the sequence $\{\tilde{g}_{n, -K^{\dagger}/2}(\alpha),\ldots, \tilde{g}_{n, K^{\dagger}/2-1}(\alpha)\}$ for a fixed $n \in \ND$ is $K^{\dagger}$-periodic.
With these in mind, we let $p = n-l$ and $q = k-d$ in the discrete convolution
\eqref{eq:gtilde_truncated}, and, for a fixed $\alpha$,  the set of approximate weights in the physical domain
to be determined is $\tilde{g}_{p,q}(\alpha)$, $p \in \ND$, $q \in \KD$.
Using this notation, in \eqref{eq:gtilde_truncated}, with  $p = n-l$ and $q = k-d$,
we rewrite $e^{2\pi i \eta_s (n-l) \Delta w} = e^{2 \pi i s \alpha p/(\alpha N^{\dagger})}$, $e^{2\pi i \xi_z (k-d) \Delta r} = e^{2 \pi i z \alpha q/(\alpha K^{\dagger})}$, and obtain
\EQ
\label{eq:gtilde_truncated_evaluation}
\begin{aligned}
\tilde{g}_{p, q}(\alpha)
&= \frac{1}{P^{\dagger}} \frac{1}{Q^{\dagger}}
\mysum^{z \in \Ka}_{s \in \Na}
e^{2 \pi i s \alpha p/(\alpha N^{\dagger})}e^{2 \pi i z \alpha q/(\alpha K^{\dagger})} ~y_{s,z},
&
p \in \ND, ~ q \in \KD,
\\
& \text{ where } y_{s,z} =
\text{tg}(s,z)~
G(\eta_s, \xi_z \Delta \tau), &
s \in \Na,~ z \in \Ka,
\end{aligned}
\EN
and $\text{tg}(s,z)$ is given in \eqref{eq:trigofun}.
It is observed from \eqref{eq:gtilde_truncated_evaluation} that, given $\{y_{s,z}\}$,
$\{\tilde{g}_{p,q}(\alpha)\}$ can be computed efficiently via a single two-dimensional FFT of size $(\alpha N^{\dagger}, \alpha K^{\dagger})$.
A suitable value for $\alpha$, i.e.\ satisfying the $\epsilon$-monotonicity condition \eqref{eq:test1},
can be determined through an iterative procedure based on formula \eqref{eq:gtilde_truncated_evaluation}.
Let this value be $\alpha_{\epsilon}$.
We also observe that, once $\alpha_{\epsilon}$  is found,
the discrete \blue{convolution} \eqref{eq:V_nG} can also be computed efficiently using an FFT. This suggests that
we only need to compute the weights in the Fourier domain, i.e.\ the DFT of $\{\tilde{g}_{p, q}(\alpha_\epsilon)\}$, only once, and  reuse them for all timesteps.
We define $\{\tilde{G}_{p, q}(\alpha_{\epsilon})\}$ to be the DFT of  $\{\tilde{g}_{p,q}(\alpha_{\epsilon})\}$ given by
\EQA
\tilde{G}(\eta_s,  \xi_z, \Delta \tau, \alpha_{\epsilon})
   =
    \frac{P^{\dagger}}{N^{\dagger}} \frac{Q^{\dagger}}{K^{\dagger}}
    \mysum^{q \in \KD}_{p \in \ND}
	e^{-2\pi i p s/N^{\dagger}} e^{-2\pi i q z/K^{\dagger}}
	~
    \tilde{g}_{p, q}(\alpha_{\epsilon}),
    \quad
    s \in \ND,~ z \in \KD.
\label{eq:bigGtilde}
\ENA
An iterative procedure for computing $\{\tilde{G}_{p, q}(\alpha_{\epsilon})\}$ is
given in Algorithm~\ref{alg:Gtilde}, where we also use the stopping criterion
$\Delta w \Delta r \mysum^{q \in \KD}_{p \in \ND}
\left|
\tilde{g}_{p, q}(\alpha)
-
\tilde{g}_{p, q}(\alpha/2)
\right|< \epsilon_1$, $\epsilon_1>0$.
\begin{algorithm}[!ht]
\caption{
\label{alg:Gtilde}
Computation of weights $\tilde{G}_{p,q}(\alpha_{\epsilon})$, $p \in \ND$, $q \in \KD$, in Fourier domain. 
}
\begin{algorithmic}[1]
\STATE set $\alpha = 1$ and compute $\tilde{g}_{p,q}(\alpha)$, $p \in \ND$, $q \in \KD$ using \eqref{eq:gtilde_truncated_evaluation};
\FOR{$\alpha = 2, 4, \ldots$ until convergence}
    \STATE
    compute
    $\tilde{g}_{p,q}(\alpha)$,
    $p \in \ND$, $q \in \KD$, using \eqref{eq:gtilde_truncated_evaluation};

    \STATE
    compute
    $\text{test}_1 = \Delta w \Delta r \sum_{p \in \ND} \sum_{q \in \KD}
    \min \left(
     \tilde{g}_{p,q}(\alpha), 0 \right) $ for monotonicity test;

    \STATE
    compute
    $\text{test}_2 = \Delta w \Delta r \sum_{p \in \ND} \sum_{q \in \KD}
     \big|\tilde{g}_{p,q}(\alpha)-\tilde{g}_{p,q}(\alpha/2)\big|$ for accuracy test;

    \IF{$|\text{test}_1 | < \epsilon (\Delta \tau/T)$
    and $\text{test}_2 < \epsilon_1$}
    \label{alg:if}

    \STATE
    \label{alg:break}
    $\alpha_{\epsilon} = \alpha$;
    \\
    break from for loop;
    \ENDIF
\ENDFOR
\STATE use \eqref{eq:bigGtilde} to compute and output weights $\tilde{G}_{p,q}(\alpha_{\epsilon})$,
$p \in \ND$, $q \in \KD$, in Fourier domain.
\end{algorithmic}
\end{algorithm}

For simplicity, unless otherwise state, we adopt the notional convention
$\tilde{g}_{n-l, k-d}=\tilde{g}_{n-l, k-d}(\alpha_{\epsilon})$
and
$\tilde{G}(\eta_s, \xi_z,  \Delta \tau) \equiv \tilde{G}(\eta_s, \xi_z,  \Delta \tau, \alpha_{\epsilon})$,
where $\alpha_{\epsilon}$ is selected by Algorithm~\ref{alg:Gtilde}.
The discrete convolutions \eqref{eq:V_nG} can then be implemented efficiently via an FFT as follows
\begin{align}
\label{eq:green_solution_project}
\bvisl_{n,k,j}^{{m+1}}
&\simeq
\mysum^{q \in \KD}_{p \in \ND}
e^{2\pi i p n/N^{\dagger}} e^{2\pi i q n/K^{\dagger}}
~
\bVisl (\eta_p, \xi_q, a_j, \tau_{m}^+)
~
\tilde{G}(\eta_p, \xi_q, \Delta \tau),
\\
\text{with} ~ \bVisl \left(\eta_p, \xi_q, a_j, \tau_{m}^+\right) &= \frac{1}{N^{\dagger}} \frac{1}{K^{\dagger}}
 \mysum^{d \in \KD}_{l \in \ND}
 	e^{-2\pi i p l/N^{\dagger}} e^{-2\pi i q d/K^{\dagger}}
 	 \bvisl_{l,d,j}^{m+},
    ~p \in \ND,~ q \in \KD,
\nonumber
\end{align}
where $(i) \in \{(1), (2)\}$ and $\tilde{G}(\eta_p,\xi_q \Delta \tau)$ is given by \eqref{eq:bigGtilde}.
Putting everything together, an $\epsilon$-monotone Fourier numerical algorithm for the HJB-QVI~\eqref{eq:gmwb_def} on $\Omega$ is presented in Algorithm~\ref{alg:monotone} below.
\begin{algorithm}[!ht]
\caption{
\label{alg:monotone}
An $\epsilon$-monotone Fourier algorithm for GMWB problem defined in Definition~\eqref{def:impulse_def}.
$x \circ y$ is the Hadamard product of matrices $x$ and $y$.
}
\begin{algorithmic}[1]
\STATE
compute matrix
$\tilde{G} = \left\{
\tilde{G}(\eta_{p}, \xi_{q}, \Delta \tau)
\right\}_{p \in \ND, q \in \KD}$,
using Algorithm~\ref{alg:Gtilde};

\STATE
\label{alg:initial}
initialize $\blue{v_{n, k, j}^{0}}
= \max\left(e^{w_n}, (1-\mu) a_j -c\right)$, $n \in \ND$, $k \in \KD$, $j \in \J$;
\hfill //$\Omega_{\tau_0}$

\FOR{$m = 0, \ldots, M-1$}

    \STATE
    \label{alg:step1_a}
    solve \eqref{eq:scheme*} to obtain
    $\vl_{n,k,j}^{m+}$ and $\vn_{n,k,j}^{m+}$,
    $n \in \N$, $k \in \K$, $j \in \J$;
    \hfill //$\Omega_{\myin} \cup  \Omega_{a_{\min}}$

    \STATE
    \label{alg:step1}
    compute  $\bvlsl_{n,k,j}^{m+}$ and $\bvnsl_{n,k,j}^{m+}$, $n \in \N$, $k \in \K$, $j \in \J$;
    using \eqref{eq:semi_lag_results}; \hfill //$\Omega_{\myin} \cup  \Omega_{a_{\min}}$

            \STATE
    \label{alg:step1_c}
    combine results in Line-\ref{alg:step1}
    with $v_{n,k,j}^{m}$ in  $\Omega_{w_{\min}}$, $\Omega_{wa_{\min}}$, $\Omega_{w_{\max}}$ and $\Omega_{\myot}$,
    to obtain
    \\
    $\bvisl_j^{m+} = \left\{ \bvisl_{n,k,j}^{m+} \right\}_{n \in \ND, k \in \KD}$,~
     $(i) \in  \{(1), (2)\}$,
    ~ $j \in \J$;

    \STATE
    \label{alg:step2}
    compute
 $\left\{\bvisl_{n,k,j}^{m+1}\right\}_{n \in \ND, k \in \KD}
    ~=~
    \text{IFFT} \left\{
      \text{FFT}\left\{
   \bvisl_j^{m+} \right\} \circ
    \tilde{G}
    \right\}$,~
    $(i) \in  \{(1), (2)\}$,~
    $j \in \J$;



    \STATE
    \label{alg:step4}
    discard FFT values in $\Omega_{w_{\min}}$, $\Omega_{wa_{\min}}$, $\Omega_{w_{\max}}$, and $\Omega_{\myot}$, namely
    $\bvlsl_{n,k,j}^{m+1}$ and $\bvnsl_{n,k,j}^{m+1}$, \\
     $n \in \Nc$, $k \in \Kc$, $j \in \J$;

    \STATE
    set
    \label{alg:step5}
    $v_{n,k,j}^{m+1} = \max\left(\bvlsl_{n,k,j}^{m+1} , \bvnsl_{n,k,j}^{m+1} \right)$, $n \in \N$, $k \in \K$, $j \in \J$;
     \hfill $//\Omega_{\myin} \cup  \Omega_{a_{\min}}$

    \STATE
    compute $v_{n,k,j}^{m+1}$,
	$n \in \Nc$, $k \in \Kc$, $j \in \J$
    using \eqref{eq:omega_max},  \eqref{eq:omega_ot}
    and \eqref{eq:scheme_timestep_left};
    \hfill// $\Omega \setminus (\Omega_{\myin} \cup  \Omega_{a_{\min}})$



\ENDFOR
\end{algorithmic}
\end{algorithm}

\begin{remark}[Wraparound error]
\label{rm:wrap}
The boundary sub-domains $\Omega_{w_{\min}} \cup \Omega_{wa_{\min}}$, $\Omega_{w_{\max}}$ and $\Omega_{\myot}$
are also set up to act as  padding areas to minimize the wraparound error
in the computation of discrete convolutions \eqref{eq:V_nG} via
an FFT in Line~\ref{alg:step2} of Algorithm~\ref{alg:monotone}.
After an FFT is applied, all results of auxiliary padding nodes in $\Omega_{w_{\min}} \cup \Omega_{wa_{\min}}$, $\Omega_{w_{\max}}$ and $\Omega_{\myot}$ are discarded to minimize the wraparound error at nodes in
 $\Omega_{\myin}\cup \Omega_{a_{\min}}$ (Line~\ref{alg:step4}).
Using similar techniques as in \cite{online} for the case of one-dimensional Green's function, we can show that,
with our choice of  $N^{\dagger} = 2N$ and $K^{\dagger} = 2K$,  where $N$ and $K$ are chosen large enough,
our handling of wraparound described above is sufficiently effective.
The reader is referred to \cite{online}[Section 4.4] for relevant details.
\end{remark}

\subsection{Fair insurance fees}
\label{ssec:fee}

With respect to the insurance fee $\beta$, let $v(\beta; w, r, a, \tau)$ be the exact solution, i.e.\
$v(w,a,r,\tau)$, be parameterised by the insurance fee $\beta$.
Then, the fair insurance fee for $t = 0$, or $\tau_M=T$, denoted by $\beta_f$, solves the equation $v\left(\beta_f; \ln(z_0), r_0, z_0, T \right) ~=~z_0$.
In a numerical setting, with a slight abuse of notation, let $v_{\ln(z_0), r_0, z_0}^M(\beta)$ be the numerical solution parametrized by $\beta$, then we need to solve $v_{\ln(z_0), r_0, z_0}^M(\beta_f)  = z_0$, where
$v_{\ln(z_0), r_0, z_0}^M$ is obtained by Algorithm~\ref{alg:monotone}. Finally, we apply the Newton iteration to solve for $\beta_f$. 
\section{Convergence to the viscosity solution}
\label{section:conv}
In this section, we appeal to a Barles-Souganidis-type analysis \cite{barles-souganidis:1991}
to rigorously study the convergence of our scheme
in $\Omega_{\myin} \cup \Omega_{a_{\min}}$ as $h \to 0$ by verifying three properties:
$\ell_\infty$-stability, $\epsilon$-monotonicity (as opposed to strict monotonicity),
and consistency. We will show that convergence of our scheme is ensured if
the monotonicity tolerance $\epsilon \to 0$ as $h \to 0$.
We note that our proofs share some similarities with those in \cite{online},
but our proof techniques are more involved due to the SL discretization, especially
for consistency of the numerical scheme. We will emphasize these key similarities and
differences where suitable.

For subsequent use,  we introduce several important results
related to relevant properties of the weights $\tilde{g}_{n-l, k-d}$ in the discrete
convolution~\eqref{eq:Gtilde_ell}.
\begin{proposition}
\label{proposition:sum_g}
For any $(n, k) \in \l\{\N \times \K\r\}$, we have
\[
\Delta w \Delta r \mysum^{d \in \KD}_{l \in \ND} \tilde{g}_{n-l, k-d} ~=~ 1,
\quad
{\text{with $\tilde{g}_{n-l, k-d}$ is given by \eqref{eq:gtilde_truncated}}}.
\]
\end{proposition}
\noindent A proof of Proposition~\ref{proposition:sum_g} is given Appendix~\ref{app:sum_g}.
Noting $\tilde{g} =\max(\tilde{g}, 0) + \min(\tilde{g}, 0)$,
Proposition~\ref{proposition:sum_g} and  the monotonicity condition \eqref{eq:test1} give the bound
\EQA
\label{eq:sum_g}
 \Delta w \Delta r
\mysum^{d \in \KD}_{l \in \ND}
 \left(\max \left(\tilde{g}_{n-l, k-d}, 0\right)  +  \left|\min\left(\tilde{g}_{n-l, k-d}, 0\right)
    \right| \right)
~\le~ 1 + 2 \epsilon \frac{\Delta \tau}{T}.
\ENA
Our scheme consists of the following equations:
\eqref{eq:terminal} for $\Omega_{\tau_0}$, \eqref{eq:omega_max} for $\Omega_{w_{\max}}$,
\eqref{eq:omega_ot} for $\Omega_{\myot}$,
\eqref{eq:scheme_timestep_left} for $\Omega_{w_{\min}}\cup\Omega_{wa_{\min}}$,
and finally \eqref{eq:scheme} for $\Omega_{\myin} \cup  \Omega_{a_{\min}}$. We start by verifying $\ell_\infty$-stability of our scheme.
\subsection{Stability}
\begin{lemma}[$\ell_\infty$-stability]
\label{lemma:stability}
Suppose that (i) the discretization parameter $h$ satisfies \eqref{eq:dis_parameter},
and (ii) the discretization \eqref{eq:scheme_timestep_left} satisfies the positive coefficient condition \eqref{eq:pos_con}, (iii) linear interpolation in \eqref{eq:vtil_a}, \eqref{eq:semi_lag_results}, and \eqref{eq:vtil_b}, and (iv) $r_{\min}< 0$ satisfies the condition
\EQA
\label{eq:r_min}
1 + \Delta \tau r_{\min} ~>~ 0.
\ENA
Then scheme \eqref{eq:terminal}, \eqref{eq:omega_max}, \eqref{eq:omega_ot}, \eqref{eq:scheme_timestep_left}, and \eqref{eq:scheme}
satisfies $\ds \sup_{h > 0} \left\| v^{m} \right\|_{\infty} < \infty$ for all $m = 0, \ldots, M$, as the discretization parameter $h \to 0$.
Here, $\left\| v^{m} \right\|_{\infty} = \max_{n, k, j} |v_{n,k, j}^{m}|$, where
$n \in \ND$, $k \in \KD$ and $j \in \J$.
\end{lemma}

\begin{proof}[Proof of Lemma~\ref{lemma:stability}]
For fixed $h >0$, we have $\left\| v^{0} \right\|_{\infty} <\infty$,
and thus, $\sup_{h > 0} \left\| v^{0} \right\|_{\infty} < \infty$.
Motivated by this observation, to demonstrate $\ell_\infty$-stability of our scheme,
we aim to demonstrate that, for a fixed $h > 0$, at any $(w_n, r_k, a_j, \tau_m)$ in $\Omega$,
\EQ
\label{eq:key}
|v_{n, k, j}^{m}| < C'(\left\| v^{0} \right\|_{\infty} + a_j),
\text{ where }
C' =  e^{2m\epsilon  \frac{\Delta\tau}{T}} e^{\R m \Delta \tau}, \text{ with }
\R = |r_{\min}| (1 +  \Delta \tau r_{\min})^{-1},
\EN
where $\epsilon$, $0 < \epsilon <1/2 $, is the monotonicity tolerance used in \eqref{eq:test1}.
Since $m \Delta\tau \le T$, $C'$ is bounded~above.


We now discuss the  important point of how to the constant $C'$ in \eqref{eq:key} is determined.
This choice is motivated by the stability bounds for
$\Omega_{\myin} \cup \Omega_{a_{\min}}$, which primarily depend on the amplification factor
of the time-advancement step. (Boundary sub-domains require smaller stability bounds as shown subsequently).
In our proof techniques, through mathematical induction on $m$, the time-$\tau_m$ accumulative amplification factor
of the time-advancement in $\Omega_{\myin} \cup \Omega_{a_{\min}}$ can be bounded by the product of the respective amplification factors of the SL discretization and of the $\epsilon$-monotone Fourier method. For the SL discretization, from \eqref{eq:semi_lag_results} and the condition \eqref{eq:r_min}, for all $k \in \K$, we have
\EQA
\label{eq:r_k_con}
0 ~< ~ (1+ \Delta \tau r_k)^{-1} ~\le~ (1+ \Delta \tau r_{\min})^{-1} = 1 + \Delta \tau \R,
\text{ where } \R = |r_{\min}| (1 +  \Delta \tau r_{\min})^{-1} > 0,
\ENA
which results in the time-$\tau_m$ accumulative amplification factor bounded by $e^{\R m \Delta \tau}$.
For the $\epsilon$-monotone Fourier method, the bound \eqref{eq:sum_g}
suggests the time-$\tau_m$ amplification factor is bounded by $e^{2m\epsilon  \frac{\Delta\tau}{T}}$.
Putting together, we obtain the constant $C' > 0$ given in \eqref{eq:key}.

We address $\ell$-stability for the boundary and interior sub-domains separately.
For \eqref{eq:terminal}, \eqref{eq:omega_max}, it is straightforward to show
$\max_{n, k, j} |v_{n, k, j}^{m}| \le \left\| v^{0} \right\|_{\infty}$,
$n \in \N \cup \Nr$, \blue{$k \in \K$}, $j \in \J$, and $m = 0, \ldots, M$.
For \eqref{eq:omega_ot}, since the $T$-maturity zero-coupon bond price
$p_b(r_k, \tau_m;T)$ given \blue{in} \eqref{eq:vother} is non-negative, the stability
trivially to show.  For \eqref{eq:scheme_timestep_left}, since the finite difference scheme
is strictly monotone, the $\ell$-stability can be demonstrated using
the induction technique (on $m$) as in \cite{chen08a}.

To prove \eqref{eq:key} for \eqref{eq:scheme},
it is sufficient to show that for all $m \in \{ 0, \ldots, M\}$ and $j \in \J$, we have
\EQA
  \left[v_j^{m}\right]_{\max} &\le&
  e^{2 m \epsilon\frac{\Delta\tau}{T}} e^{\R m \Delta \tau}
\left(
\left\| v^{0} \right\|_{\infty} + a_{j}\right),
    \label{eq:vminusmax}
    \\
  - 2m\epsilon\frac{\Delta\tau}{T}e^{2m\epsilon  \frac{\Delta\tau}{T}} e^{\R m \Delta \tau}
\left(
\left\| v^{0} \right\|_{\infty} + a_{j}\right)
&\le&
\left[v_j^{m}\right]_{\min}.
  \label{eq:vminusmin}
\ENA
where $\big[v_j^{m}\big]_{\max} = \max_{n,k}\big\{v_{n,k,j}^{m}\big\}$ and $\big[v_j^{m}\big]_{\min} = \min_{n,k}\big\{v_{n,k,j}^{m}\big\}$. To prove \eqref{eq:vminusmax}-\eqref{eq:vminusmin},
motivated by the above reasoning regarding the choice $C'$,  we use mathematical induction on $m = 0, \ldots, M$, similar to the technique developed in \cite{online}[Lemma~5.1]. The details for this step are provided in
Appendix~\ref{app:induction}.
\end{proof}

\subsection{Error analysis results}
\label{ssc:error_analysis}
In this subsection, we identify errors arising in our numerical scheme and make assumptions needed
for subsequent proofs.
\begin{enumerate}
\item Truncating the infinite region of integration in the convolution integral
\eqref{eq:pide_con_int} to $\D^{\dagger}$ (defined in \eqref{eq:D_dagger}) results in a
boundary truncation error, denoted by $\errorb$, where
\EQA
\label{eq:green_integral_truncated_e}
\errorb  = \iint_{\mathbb{R}^2 \setminus \D^{\dagger}}
g(w - w', r - r', \dtau)~\vhsl(w', r', \cdot,\taus)~dw'~dr',
\quad (w, r) \in \D.
\ENA
Similar to the discussions in \cite{online}, 
we can show that $\errorb$ is bounded by
\EQAS
\label{eq:truncation}
        \left| \errorb \right| &\le&  K_1 \Delta \tau e^{-K_2 \l(P^{\dagger} \wedge Q^{\dagger}\r)},
    \quad
    \text{$\forall (w,r) \in \D$},
     \quad
     \text{$K_1, K_2 > 0$ independent of $\Delta \tau$, $P^{\dagger}$ and $Q^{\dagger}$},
\ENAS
where $P^{\dagger} = w^{\dagger}_{\max} - w^{\dagger}_{\min}$ and $Q^{\dagger} = r^{\dagger}_{\max} - r^{\dagger}_{\min}$. For fixed $P^{\dagger}$ and $Q^{\dagger}$, \eqref{eq:truncation} shows $\errorb \to 0$, as $\Delta \tau \to 0$.
However,  as typically required for showing consistency, one would need to ensure $\frac{\errorb}{\Delta \tau}  \to 0$ as $\Delta \tau \to 0$.
Therefore, from \eqref{eq:truncation}, we need $P^{\dagger} \to \infty$ and $Q^{\dagger} \to \infty$ as $\Delta \tau \to 0$, which can be achieved by letting $P^{\dagger} = C/\Delta \tau$ and $Q^{\dagger} = C'/\Delta \tau$,
for finite $C>0$ and $C' >0$.

\item The next error arises in approximating the Green's function $g(w, r,  \Delta \tau)$ by its localized, periodic approximation $\hat{g}(w, r, \Delta \tau)$ defined in \eqref{eq:green_hat}. We denote this error by
    $\errorg$.
While $\hat{g}(w, r, \Delta \tau) \neq g(w, r, \Delta \tau)$  for $(w, r) \in \D$. Nonetheless,
if $P^{\dagger} = C_5/\Delta \tau$ and $Q^{\dagger} = C'_5/\Delta \tau$ as discussed above,
then, as $\Delta \tau \to 0$, we have
\begin{linenomath}
\postdisplaypenalty=0
\begin{align*}
\hat{g}(w,r, \Delta \tau)
\overset{\text{(i)}}{=}
\iint_{\mathbb{R}^2} e^{2\pi i \eta w} e^{2\pi i \xi r} G(\eta, \xi, \Delta \tau) d\eta d\xi + \mathcal{O}\l(1/\l(P^{\dagger} \wedge Q^{\dagger}\r)^2\r)
\overset{\text{(ii)}}{=}
g(w, r, \Delta \tau) + \mathcal{O}(\Delta \tau ^2).
\end{align*}
\end{linenomath}
Here, (i) is due to $P^{\dagger} \to \infty$ and $Q^{\dagger} \to \infty$ as $\Delta \to 0$, ensuring in an
$\mathcal{O}\l(1/\l(P^{\dagger} \wedge Q^{\dagger}\r)^2\r) \sim \mathcal{O}((\Delta \tau) ^2)$ error for the traperzoidal rule approximation of the integral, and (ii) is due to that $G(\cdot)$ is the Fourier transform of $g(\cdot)$. Therefore, $\errorg = \mathcal{O}(\Delta \tau ^2)$ as $\Delta \tau \to 0$.

\item Truncating $\tilde{g}_{n-l}(\infty)$, defined in \eqref{eq:gtilde_infty}, to
to $\tilde{g}_{n-l}(\alpha)$, for a finite $\alpha \in \{2, 4, 8, \ldots\}$, in \eqref{eq:gtilde_truncated},
\blue{gives rise} to a Fourier series truncation error, denote by $\errorf$.
As shown in Appendix~\ref{app:truncation}, as $\Delta \tau$, $\Delta w$ and $\Delta r \to 0$,  this error is
\[
\errorf = \mathcal{O}\l(e^{-\frac{\Delta \tau}{(\Delta w)^2}}/(\Delta w \wedge \Delta r)^2\r)
+ \mathcal{O}\l(e^{-\frac{\Delta \tau}{(\Delta r)^2}}/(\Delta w \wedge \Delta r)^2\r), \quad
\text{ as } \Delta \tau,~\Delta w,~\Delta r \to 0.
\]


\item Approximating a function in $\G \cap \C{\Oinf}$ by its projection on the piecewise linear basis
     functions $\varphi_l(\cdot)$ and $\psi_d(\cdot)$, $l \in \ND$ ad $d \in \KD$,
     as in \eqref{eq:rhowfunc}, as well as by linear interpolation, as in Remark~\eqref{eq:intp},
     gives rise to a projection/interpolation error, collectively denoted by
     $\errorp$. Generally $ \errorp  = \mathcal{O}\l(\max (\Delta w, \Delta r, \Delta a)^2\r)$,
     as $\Delta w, \Delta r, \Delta a \to 0$.
\end{enumerate}
Motivated by the above discussions, for convergence analysis, we make an assumption about
the discretization parameter.
\begin{assumption}
\label{as:dis_parameter}
We assume that there is a discretization parameter $h$
such that
\EQA
\label{eq:dis_parameter}
&& \Delta w = C_1 h,
\quad
\Delta r = C_2 h,
\quad
\Delta a_{\max} = C_3 h,
\quad
\Delta a_{\min} = C_3'h,
\nonumber
\\
&& \qquad\qquad\qquad\qquad\qquad\qquad\qquad
\Delta \tau = C_4 h,
\quad
P^{\dagger} = C_5/h,
\quad
Q^{\dagger} = C'_5/h,
\ENA
where the positive constants $C_1$, $C_2$, $C_3$, $C'_3$, $C_4$, $C_5$ and $C_5'$ are independent of $h$.
\end{assumption}
Under Assumption~\ref{as:dis_parameter}, it is straightforward to obtain
\EQ
\label{eq:err_h}
\errorb = \mathcal{O}(he^{-\frac{1}{h}}),
\quad
\errorg = \mathcal{O}(h^2),
\quad
\errorf  = \mathcal{O}(e^{-\frac{1}{h}}/h^2),
\quad
\errorp  = \mathcal{O}(h^2).
\EN
It is also straightforward to ensure the theoretical requirement $P^{\dagger}, Q^{\dagger} \to \infty$ as $h \to 0$.
For example, with $C_5 = C'_5 = 1$ in \eqref{eq:dis_parameter},
we can quadruple $N^{\dagger}$ and $K^{\dagger}$ as we halve $h$.
We emphasize that, for practical purposes, if $P^{\dagger}$ and $Q^{\dagger}$ are chosen sufficiently large,
both can be kept constant for all $\Delta \tau$ refinement levels (as we let $\Delta \tau  \to 0$).
The effectiveness of this practical approach is demonstrated through numerical experiments
in Section~\ref{sec:num_test}. Also \blue{see} relevant discussions in \cite{online}.

To show convergence of the numerical scheme to the viscosity solution,
our starting point is discrete convolutions of the form \eqref{eq:V_nG} which typically involve a generic function $\varphi \in \G$. There are two cases: (i) $\varphi$ is not necessarily smooth,
which corresponds to the SL discretization or non-local impulses,
and (ii) $\varphi$ is a test function in $\G \cap \C{\Oinf}$,
which corresponds to local impulses. In subsequent discussions,
we present results relevant to these two cases in
Lemma~\ref{lemma:ar} below. For differential and jump operators, we use the notation
$[\cdot]_{n,k,j}^m := [\cdot]({\bf{x}}_{n,k,j}^m)$.

\begin{lemma}
\label{lemma:ar}
Suppose the discretization parameter $h$ satisfies Assumption~\ref{as:dis_parameter}.
Let $\phi$  and $\chi$ be in $\G \cap \C{\Oinf}$ and $\G$, respectively.
For ${\bf{x}}_{n,k,j}^{m}$, $n \in \N$, $j \in \J$, $k \in \K$, $m \in \{0, \ldots, M\}$,
we~have
\begin{align}
\Delta w \Delta r
\mysum^{d \in \KD}_{l \in \ND}
    \tilde{g}_{n-l, k-d}~
    \phi_{l,d,j}^{m}
&=
\phi_{n,k,j}^{m}
+ \Delta \tau \left[ \Lg \phi + \mathcal{J} \phi\right]_{n,k,j}^{m} + \mathcal{O}( h^2),
\label{eq:error_analysis_1}
\\
\Delta w \Delta r
\mysum^{d \in \KD}_{l \in \ND}
    \tilde{g}_{n-l, k-d}~
    \chi_{l,d,j}^{m}
&=
\chi_{n,k,j}^{m} + \mathcal{O}(h^2)
+  \myerrm{\chi},
\text{ where $\myerrm{\chi} \to 0$ as $h \to 0$}.
\label{eq:error_analysis_2}
\end{align}
\end{lemma}
\begin{proof}[Proof of Lemma~\ref{lemma:ar}]
Lemma~\ref{lemma:ar} can be proved using similar techniques in
\cite{online}[Lemmas~5.3 and 5.4] for the one-dimensional Greens' function case.
For completeness, we provide the key steps below.
We let $a = a_j$ and $\tau_ = \tau_m$ be fixed, and with a slight abuse of notation,
we view $\phi$ and $\chi$ as functions of $(w, r)$.
Let $\xi  \in \{\phi, \chi\}$.
Starting from the discrete convolutions on the left-hand-side of \eqref{eq:error_analysis_1}-\eqref{eq:error_analysis_2},
we need to recover an associated convolution integrals of the form \eqref{eq:pide_con_int}
which is posed on an infinite integration region. Since  $\xi  \in \{\chi, \phi\}$
is not necessarily in $L^1(R^2)$, standard mollification techniques can be used
to obtain $\xi' \in L^1(R^2)$ which agrees with  $\xi$ on $\D^{\dagger}$.
Then, with $\xi  \in \{\phi, \chi\}$, using error analysis, we have
\EQ
\label{eq:integral}
\Delta w \Delta r
\mysum^{d \in \KD}_{l \in \ND}
    \tilde{g}_{n-l, k-d}~
    \xi_{l,d,j}^{m} = \iint_{R^2}\xi''(w, r) g(w_n - w, r_k - r, \Delta \tau) dwdr +
    \errorb + \errorg + \errorf + \errorp.
\EN
where $\xi''$ is a projection of $\xi'$ onto the piecewise linear basis functions $\varphi_l(\cdot)$ and $\psi_d(\cdot)$, $l \in \ND$ ad $d \in \KD$.
By Assumption~\ref{as:dis_parameter} and \eqref{eq:err_h},
$ \errorb + \errorg + \errorf + \errorp = \mathcal{O}(h^2)$.

For $\xi  = \phi$, and since $\phi$ is smooth, we then apply the Fourier Transform and inverse Fourier Transform to
$\iint_{R^2}\xi''(w, r) g(w_n - w, r_k - r, \Delta \tau) dwdr$ in  \eqref{eq:integral} to recover the differential and jump operators.

For $\xi  = \chi$ which is not smooth,  we write the convolution integral in \eqref{eq:integral} as
\[
\iint_{R^2}\chi''(w, r) g(w_n - w, r_k - r, \Delta \tau)
= \chi''(w_n, r_k) + \iint_{\mathbb{R}^2}  g(w_n -w, r_k - r, \Delta \tau)\l(\chi''(w, r) - \chi''(w_n, r_k)\r)~dwdr.
\]
Note that $\chi''(w_n, r_k) = \chi_{l,d,j}^{m}$, and letting
$\myerrm{\chi} =  \iint_{\mathbb{R}^2} (\cdot) dwdr$ gives \eqref{eq:error_analysis_2},
due to the ``cancelation properties'' of the Green's function
\cite{garronigreenfunctionssecond92, Duffy2015}.
This concludes the proof.
\end{proof}
We now consider a special case of the discrete convolution \eqref{eq:V_nG}
that involves interpolation of values of a smooth test function evaluated at the departure points of the SL trajectory presented in Subsection~\ref{ssc:sld}. Specifically, given $\phi \in \G \cap \C{\Oinf}$,
for ${\bf{x}}_{l,d,q}^{m+1} \in \Omega$, $0 <\tau_{m+1} \le T$, we define discrete values $\l(\pb\r)_{l,d,q}^{m}$ as follows
\EQA
\label{eq:phi_term}
\begin{array}{l}
(\pb)_{l,d,q}^{m}
= \left\{
\begin{array}{lll}
\mathcal{I} \{ \phi^{m}\} (\bw_l, \br_d, a_q) (1 + \Delta \tau r_d)^{-1}
& {\bf{x}}_{l,d,q}^{m+1} \in \Omega_{\myin} \cup \Omega_{a_{\min}},
\\
\phi_{l,d,q}^{m}
& \text{otherwise}.
\end{array}
\right.
\end{array}
\ENA
Here, as described in Remark~\ref{eq:intp}, $\mathcal{I} \left\{ \phi^{m} \right\} (\cdot)$
is the linear interpolation operator acting
on discrete data $\left\{(w_l, r_d, a_q), \phi_{l,d,q}^{m} \right\}$
and $(\bw_l, \br_d)$ is given by \eqref{eq:semi_lag_results}, while $a_q$ is fixed.

\begin{lemma}
\label{lemma:error_smooth_sl}
Let $\phi \in \G \cap \C{\Oinf}$ and $\{(w_l, r_d, a_q), (\pb)_{l,d,q}^{m}\}$ be given by \eqref{eq:phi_term}.
For any fixed ${\bf{x}}_{n,k,j}^{m} \in \Omega_{\myin} \cup \Omega_{a_{\min}}$, i.e.\ $n \in \N$, $j \in \J$, $k \in \K$,
and $m \in \{1, \ldots, M\}$, we have
\EQ
\label{eq:error_smooth_sl}
\Delta w \Delta r
\mysum_{l \in \ND}^{d \in \KD}
        \tilde{g}_{n-l, k-d}~
   (\pb)_{l,d,j}^{m}
=
\phi_{n,k,j}^{m}+ \Delta \tau \left[ \mathcal{L} \phi + \mathcal{J} \phi
\right]_{n,k,j}^{m} + \mathcal{O}( h^2 ) + \Delta \tau \myerrm{}.
\EN
Here, $\tilde{g}_{n-l, k-d}$ is given by \eqref{eq:gtilde_truncated},
$\mathcal{L}$ and $\mathcal{J}$ are defined in \eqref{eq:Operator_LJ},
and $\myerrmm{} \to 0$ as $h \to 0$..
\end{lemma}
\begin{proof}[Proof of Lemma~\ref{lemma:error_smooth_sl}]
We let $j \in \J$ be fixed in this proof. We start by investigating the interpolation result $\mathcal{I} \left\{ \phi^{m}\right\} (\bw_l, \br_d, a_j)$ for ${\bf{x}}_{l,d,j}^{m} \in \Omega_{\myin} \cup \Omega_{a_{\min}}$ in \eqref{eq:phi_term}.
Remark~\ref{eq:intp}
\EQA
\label{eq:phi_interp_sl}
\mathcal{I} \left\{ \phi^{m} \right\} (\bw_l, \br_d, a_j)
&\overset{(i)}{=} &\phi\left(\bw_{l},  \br_{d}, a_j, \tau_{m}\right)
 + \mathcal{O}(h^2)
\nonumber
\\
&\overset{(ii)}{=} &\phi_{l,d,j}^{m} + \Delta \tau \l[(r_d-\frac{\sigz^2}{2} - \beta) (\phi_w)_{l,d,j}^{m} + \delta (\theta - r_d) (\phi_r)_{l,d,j}^{m} \r] + \mathcal{O}(h^2)
\nonumber
\\
&= & \phi_{l,d,j}^{m} + \Delta \tau \left[ \Ls \phi \right]^m_{l,d,j} + \mathcal{O}(h^2).
\ENA
Here, (i) follows from Remark~\ref{eq:intp}[equation \eqref{eq:interp_sim}], noting $\phi \in \C{\Oinf}$;
in (ii), we apply a Taylor series to expand
the term $\phi\left(\bw_{l}, \br_{d}, a_j, \tau_{m}\right)$ about the point $\left(w_{l}, r_{d}, a_{j}, \tau_{m}\right)$,
and then use $e^{\Delta \tau} = 1 + \Delta \tau + \mathcal{O}(h^2)$ and $e^{-\delta \Delta \tau} = 1 - \delta \Delta \tau + \mathcal{O}(h^2)$.
We note that, for ${\bf{x}}_{l,d,q}^{m} \in \Omega_{\myin} \cup \Omega_{a_{\min}}$, we have
\EQ
\label{eq:simple}
(1+ \Delta \tau r_d)^{-1} = 1 - \Delta \tau r_d + \mathcal{O}\left((\Delta \tau)^2\right),
\quad r_d \in [r_{\min}, r_{\max}].
\EN
Using \eqref{eq:simple} and \eqref{eq:phi_interp_sl}, we arrive at
\begin{align}
\label{eq:phi_interp_sl_2}
\mathcal{I} \left\{ \phi^{m} \right\} (\bw_l, \br_d, a_j) (1 + \Delta \tau r_d)^{-1}
=
\phi_{l,d,j}^{m} + \Delta \tau \left[ \Ls \phi - r \phi \right]_{l,d,j}^{m} + \mathcal{O}(h^2),\quad
{\bf{x}}_{l,d,j}^{m} \in \Omega_{\myin} \cup \Omega_{a_{\min}}.
\end{align}

\noindent Next, letting ${\bf{x'}} = (w', a', r', \tau')$,
we define a function $\psi\left({\bf{x'}}\right): \Oinf \to \mathbb{R}$ by
\EQA
\label{eq:phi_interp_sl_3}
\psi\left({\bf{x'}}\right)
=
\left\{
\begin{array}{ll}
(r' -\frac{\sigz^2}{2} - \beta) \phi_w\left({\bf{x'}}\right)    +
\delta (\theta - r' )\phi_r\left({\bf{x'}}\right) - r' \phi\left({\bf{x'}}\right),
& {\bf{x'}} \in \Omega_{\myin} \cup \Omega_{a_{\min}},
\\
0 & {\text{otherwise}}.
\end{array}
\right.
\ENA
Note that $\psi \in \G$, and that $\psi_{l,d,j}^{m} = [ \Ls \phi - r \phi]_{l,d,j}^{m}$ for ${\bf{x}}_{l,d,j}^{m} \in  \Omega_{\myin} \cup \Omega_{a_{\min}}$.
Now, we consider the discrete convolution on the rhs of \eqref{eq:error_smooth_sl}:
$\Delta w \Delta r
\mysum_{l \in \ND}^{d \in \KD}
        \tilde{g}_{n-l, k-d}~
  (\pb)_{l,d,j}^{m} = \ldots$
\begin{linenomath}
\postdisplaypenalty=0  
\begin{align*}
\ldots
&\overset{(i)}{=}
\Delta w \Delta r
\mysum_{l \in \ND}^{d \in \KD}
        \tilde{g}_{n-l, k-d}~\phi_{l,d,j}^{m}
        + \Delta \tau
        \big(\Delta w \Delta r
        \mysum_{l \in \ND}^{d \in \KD}
        \tilde{g}_{n-l, k-d}~
        \psi_{l,d,j}^{m}\big)
        + \mathcal{O}(h^2)
\\
&
\overset{(ii)}{=}
\phi_{n,k,j}^{m} + \Delta \tau \left[ \Lg \phi + \mathcal{J} \phi
\right]_{n,k,j}^{m} + \Delta \tau \left[ \Ls \phi - r \phi \right]_{n,k,j}^{m} +
  \Delta \tau \myerrm{} + \mathcal{O}(h^2)
\\
&
\overset{(iii)}{=}
\phi_{n,k,j}^{m} + \Delta \tau \left[ \mathcal{L} \phi + \mathcal{J} \phi
\right]_{n,k,j}^{m} + \mathcal{O}(h^2) + \Delta \tau \myerrm{}.
\end{align*}
\end{linenomath}
Here, (i) is due to the definition of $(\pb)_{l,d,j}^{m}$ given in \eqref{eq:phi_term},
together with \eqref{eq:phi_interp_sl_2}-\eqref{eq:phi_interp_sl_3}, and Proposition~\ref{proposition:sum_g} to get $\mathcal{O}(h^2)$. In (ii), we use Lemma~\ref{lemma:ar}[equation~\eqref{eq:error_analysis_2}] on the discrete convolution involving $ \psi_{l,d,j}^{m}$,
noting its definition $\eqref{eq:phi_interp_sl_3}$ and $\myerrm{}\to 0$ as $h \to 0$; and in (iii), we use $\mathcal{L} \phi = \Lg \phi +  \Ls \phi - r\phi$. This concludes the proof.
\end{proof}

\subsection{Consistency}
\label{section:consistency}
While equations \eqref{eq:terminal}, \eqref{eq:omega_max}, \eqref{eq:omega_ot},
\eqref{eq:scheme_timestep_left}, and \eqref{eq:scheme} are convenient for computation,
they are not in a form amendable for analysis.
For purposes of proving consistency, it is more convenient to rewrite them in
a single equation.  To this end, we recall that we  partition $[0, a_j]$ into $[0, a_j\wedge C_r\Delta\tau]$ and $(C_r\Delta\tau, a_j]$, with the convention that $(C_r \Delta \tau, a_j] = \emptyset$ if $a_j \le C_r \Delta \tau$.
Subsequently in this subsection, the aforementioned partition of $[0, a_j]$ is used  to write
\eqref{eq:terminal}, \eqref{eq:omega_max}, \eqref{eq:omega_ot},
\eqref{eq:scheme_timestep_left}, and \eqref{eq:scheme} into an equivalent single equation
convenient for analysis. Unless noted otherwise, in the following, let $j \in \J$ and $m \in \M$ be fixed.

For $(w_n,r_k,a_j, \tau_{m+1}) \in \Omega_{w_{\min}} \cup \Omega_{wa_{\min}}$,
i.e.\ $n \in \Nl$ and $k \in \K$, we define the following operators:
$\mathcal{A}_{n,k,j}^{m+1}
\left(h, v_{n,k,j}^{m+1},
\left\{v_{l,d,p}^{m}\right\}_{\subalign{p \le j}}\right)
\equiv
\mathcal{A}_{n,k,j}^{m+1}\left(\cdot\right)$
and
$\mathcal{B}_{n,k,j}^{m+1}
\left(h, v_{n,k,j}^{m+1},
\left\{v_{l,d,p}^{m}\right\}_{\subalign{p \le j}}
\right) \equiv \mathcal{B}_{n,k,j}^{m+1}\left(\cdot\right)$, where
\EQA
\label{eq:scheme_AB}
\mathcal{A}_{n,k,j}^{m+1}\left(\cdot\right)
&=&
\frac{1}{\Delta \tau}\bigg[ v_{n,k,j}^{m+1}
-
    \sup_{\gamma_{n,k,j}^{m} \in \left[0, a_j \wedge C_r \Delta \tau\right]}
   \left(
   \tilde{v}_{n,k,j}^{m}
   +
    f\left(\gamma_{n,k,j}^{m}\right)
   \right)
+ \Delta \tau (\Ld^h  v )_{n,k,j}^{m+1}
\bigg],
\nonumber
\\
\mathcal{B}_{n,k,j}^{m+1}\left(\cdot\right)
&=&
v_{n,k,j}^{m+1}
-
     \sup_{\gamma_{n,k,j}^{m} \in \left(C_r \Delta \tau, a_j\right]}
    \left(
     \tilde{v}_{n,k,j}^{m}
+           f\left(\gamma_{n,k,j}^{m}\right)
        \right)
 + \Delta \tau (\Ld^h  v )_{n,k,j}^{m+1},
\ENA
where $\tilde{v}_{n,k,j}^{m}$, $n \in \Nl$ and $k \in \K$,
is given in \eqref{eq:vtil_a}, and $f\left(\cdot\right)$
is defined in \eqref{eq:f_gamma_k_dis}.

For $(w_n,r_k,a_j, \tau_{m+1}) \in \Omega_{\myin} \cup \Omega_{a_{\min}}$,
i.e.\ $n \in \N$ and $k \in \K$, we define the following operators:
$\mathcal{C}_{n,k,j}^{m+1}
\left(h, v_{n,k,j}^{m+1},
\left\{v_{l,d,p}^{m}\right\}_{\subalign{p \le j}}
 \right)\equiv \mathcal{C}_{n,k,j}^{m+1}\left(\cdot\right)$
and
$\mathcal{D}_{n,k,j}^{m+1}
\left(h, v_{n,k,j}^{m+1},
\left\{v_{l,d,p}^{m}\right\}_{\subalign{p \le j}}
 \right) \equiv \mathcal{D}_{n,k,j}^{m+1}\left(\cdot\right)$,
where
\begin{linenomath}
\postdisplaypenalty=0
\label{eq:scheme_CD}
\begin{align}
\mathcal{C}_{n,k,j}^{m+1}\left(\cdot \right)
&=
\frac{1}{\Delta \tau}\l[ v_{n,k,j}^{m+1}
-
\Delta w \Delta r \mysum_{l \in \N}^{d \in \K}
\tilde{g}_{n-l, k-d}~\bvlsl^{m+}_{l,d,j}
-
\Delta w \Delta r \mysum_{l \in \Nc}^{d \in \Kc}
\tilde{g}_{n-l, k-d}~
v_{l,d,j}^{m} \r],
\nonumber
\\
\mathcal{D}_{n,k,j}^{m+1}\left(\cdot\right)
&=
v_{n,k,j}^{m+1}
-
\Delta w \Delta r \mysum_{l \in \N}^{d \in \K}
\tilde{g}_{n-l, k-d}
\bvnsl^{m+}_{l,d,j}
-
\Delta w \Delta r \mysum_{l \in \Nc}^{d \in \Kc}
\tilde{g}_{n-l, k-d}~
v_{l,d,j}^{m}.
\end{align}
\end{linenomath}
Here, for $(i) \in  \{(1), (2)\}$,
$\bvisl_{l,d,j}^{m+} = \frac{\mathcal{I}\left\{(v^{\mysup})^{m+} \right\} \left(\bw_l, \br_d, a_j \right)}{1 + \Delta \tau r_d}$, $l \in \N$ and $d \in \K$, are defined in \eqref{eq:semi_lag_results},
and $\mathcal{I}\{(v^{\mysup})^{m+}\}$, a  linear operator discussed in Remark~\ref{eq:intp}.

In order to show local consistency, we split the sub-domains $\Omega_{\myin}$ and $\Omega_{w_{\min}}$
as follows: $\Omega_{\myin} = \Omega^{\mydn}_{\myin} \cup \Omega^{\myup}_{\myin}$ and
$\Omega_{w_{\min}} = \Omega^{\mydn}_{w_{\min}} \cup \Omega^{\myup}_{w_{\min}}$, where
\EQ
\label{eq:new_dom}
\begin{aligned}
\Omega^{\mydn}_{\myin} &= (w_{\min}, w_{\max}) \times
(r_{\min}, r_{\max}) \times
(a_{\min}, C_r \Delta \tau] \times
(0, T],
\\
\Omega^{\myup}_{\myin} &= (w_{\min}, w_{\max}) \times
(r_{\min}, r_{\max}) \times
(C_r \Delta \tau, a_{\max}] \times
(0, T],
\\
\Omega^{\mydn}_{w_{\min}} &= [w^{\dagger}_{\min}, w_{\min}] \times
(r_{\min}, r_{\max}) \times
(a_{\min}, C_r \Delta \tau] \times
(0, T],
\\
\Omega^{\myup}_{w_{\min}} &= [w^{\dagger}_{\min}, w_{\min}] \times
(r_{\min}, r_{\max}) \times
(C_r \Delta \tau, a_{\max}] \times
(0, T].
\end{aligned}
\EN
Using $\mathcal{A}_{n,k,j}^{m+1}\left(\cdot\right)$,
$\mathcal{B}_{n,k,j}^{m+1}\left(\cdot\right)$,
$\mathcal{C}_{n,k,j}^{m+1}\left(\cdot\right)$ and $\mathcal{D}_{n,k,j}^{m+1}\left(\cdot\right)$
defined \eqref{eq:scheme_AB}-\eqref{eq:scheme_CD},
our scheme at  the reference node ${\bf{x}} = (w_n,r_k,a_j, \tau_{m+1})$
can be rewritten in an equivalent form as follows
\EQA
\label{eq:scheme_G}
0=
\mathcal{H}_{n,k,j}^{m+1}
\left(h, v_{n,k,j}^{m+1},
\left\{v_{l,d,p}^{m}\right\}_{\subalign{p \le j}}
  \right)
\equiv
\left\{
\begin{array}{lllllllllllll}
\mathcal{A}_{n,k,j}^{m+1}
\left(\cdot\right)
& \quad{\bf{x}} \in \Omega^{\mydn}_{w_{\min}} \cup \Omega_{wa_{\min}},
\\
\min\left\{ \mathcal{A}_{n,k,j}^{m+1} \left(\cdot\right), \mathcal{B}_{n,k,j}^{m+1} \left(\cdot\right) \right\}
&
\quad {\bf{x}} \in \Omega^{\myup}_{w_{\min}},
\\
\mathcal{C}_{n,k,j}^{m+1}
\left(\cdot\right)
&
\quad {\bf{x}} \in \Omega^{\mydn}_{\myin} \cup \Omega_{a_{\min}},
\\
\min\left\{ \mathcal{C}_{n,k,j}^{m+1}\left(\cdot\right), \mathcal{D}_{n,k,j}^{m+1} \left(\cdot\right) \right\}
&
\quad {\bf{x}} \in \Omega^{\myup}_{\myin},
\\
v_{n,k,j}^{m+1} - e^{-\beta \tau_{m+1}} e^{w_n}
&
\quad {\bf{x}} \in \Omega_{w_{\max}},
\\
v_{n,k,j}^{m+1} - \max(e^{w_{n}}, (1 - \mu)a_j -c)
&
\quad {\bf{x}} \in \Omega_{\tau_0},
\\
v_{n,k,j}^{m} - \z(w_n,r_k,a_j, \tau_{m})
&
\quad {\bf{x}} \in \Omega_{\myot},
\end{array}
\right.
\ENA
where the sub-domains are defined in \eqref{eq:sub_domain_whole} and \eqref{eq:new_dom}.

To demonstrate the consistency in viscosity sense of \eqref{eq:scheme_G}, we need some intermediate results
on local consistency of our scheme. To this end, motivated by the aforementioned partitioning
of $[0, a_j]$, we define operators $F_{{\myin}'}$ and $F_{w'_{\min}}$, respectively
associated with $F_{{\myin}}$ and $F_{w_{\min}}$,  for the case  $0 \le a_j \le C_r \Delta \tau$,
i.e.\  $0 \le a/\Delta \tau \le C_r$, as follows
\EQA
\label{eq:extra}
F_{{\myin}'} \left({ \bf{x}}, v\right) &=&
v_{\tau} - \mathcal{L} v - \mathcal{J}v
- \sup_{\hat{\gamma} \in \left[0, a/\Delta \tau \right]}
\hat{\gamma}\left(1 - {e^{-w}}v_w  - v_a\right) {\bf{1}}_{\{a>0 \}},
\quad
0 \le a/\Delta \tau \le C_r,
\nonumber
\\
F_{w'_{\min}} \left({ \bf{x}}, v\right) &=&
v_{\tau} - \Ld v
- \sup_{\hat{\gamma} \in \left[0, a/\Delta \tau \right]}
\hat{\gamma}\left(1 - v_a\right){\bf{1}}_{\{a>0 \}},
\quad
0 \le a/\Delta \tau \le C_r.
\ENA
%
Below, we state the key supporting lemma related to local consistency of scheme \eqref{eq:scheme_G}.
\begin{lemma} [Local consistency]
\label{lemma:consistency}
Suppose that (i) the discretization parameter $h$ satisfies Assumption~\ref{as:dis_parameter},
(ii)~linear interpolation in \eqref{eq:vtil_a}, \eqref{eq:semi_lag_results}, and \eqref{eq:vtil_b}
is used, and (iii)~$w_{\min}$ satisfies
\EQA
\label{eq:min_e_w}
e^{w_{\min}} - e^{w^{\dagger}_{\min}} \geq C_r \Delta \tau.
\ENA
Then, for any function $\phi\in \G \cap \C{\Oinf}$,
with $\phi_{n,k,j}^{m} = \phi\left({\bf{x}}_{n,k,j}^{m}\r)$
and ${\bf{x}} =(w_n,r_k,a_j, \tau_{m+1})$, and for a sufficiently small $h$,  we have
\EQA
\label{eq:lemma_1}
\mathcal{H}_{n,k,j}^{m+1}
\l(h, \phi_{n,k,j}^{m+1} + \xi,
\{\phi_{l,d,p}^{m}+\xi\}_{p \le j}\r)
=
\left\{
\begin{array}{llllllllllr}
F_{\myin}(\cdot, \cdot)
&\!\!\!\!\!+~ c({\bf{x}})\xi
+ \mathcal{O}(h)
+ \err
& {\bf{x}} \in \Omega^{\myup}_{\myin};
\\
F_{{\myin}'}(\cdot, \cdot)
&\!\!\!\!\!+~ c({\bf{x}})\xi
+ \mathcal{O}(h)
+ \err
& {\bf{x}} \in \Omega^{\mydn}_{\myin};
\\
F_{a_{\min}}(\cdot, \cdot)
&\!\!\!\!\!+~ c({\bf{x}})\xi
+ \mathcal{O}(h)
& {\bf{x}} \in \Omega_{a_{\min}}
\\
F_{w_{\min}}(\cdot, \cdot)
&\!\!\!\!\!+~ c({\bf{x}})\xi
 + \mathcal{O}(h)
& {\bf{x}} \in \Omega^{\myup}_{w_{\min}};
\\
F_{w'_{\min}}(\cdot, \cdot)
&\!\!\!\!\!+~ c({\bf{x}})\xi
+ \mathcal{O}(h)
& {\bf{x}} \in \Omega^{\mydn}_{w_{\min}};
\\
F_{wa_{\min}}(\cdot, \cdot)
&\!\!\!\!\!+~ c({\bf{x}})\xi
+  \mathcal{O}(h)
& {\bf{x}} \in \Omega_{wa_{\min}};
\\
F_{w_{\max}}(\cdot, \cdot)
&\!\!\!\!\!+~ c({\bf{x}})\xi
& {\bf{x}} \in \Omega_{w_{\max}};
\\
F_{\tau_0}(\cdot, \cdot)
&\!\!\!\!\!+~ c({\bf{x}})\xi
& {\bf{x}} \in \Omega_{\tau_0};
\\
F_{\myot}(\cdot, \cdot)
&\!\!\!\!\!+~c({\bf{x}})\xi
& {\bf{x}} \in \Omega_{\myot}.
\end{array}
\right.
\ENA
Here, $\xi$ is a constant and $c(\cdot)$ is a bounded function
satisfying $|c({\bf{x}})| \le \max(|r_{\min}|, r_{\max}, 1)$
for all ${\bf{x}}~\in~\Omega$,
and $\err \to 0 $ as $h \to 0$.
The operators $F_{\myin}\left(\cdot, \cdot\right)$, 
$F_{a_{\min}}\left(\cdot, \cdot\right)$, $F_{w_{\min}}\left(\cdot, \cdot\right)$,
$F_{wa_{\min}}\left(\cdot, \cdot\right)$, $F_{w_{\max}}\left(\cdot, \cdot\right)$
$F_{\tau_0}\left(\cdot, \cdot\right)$, $F_{\myot}\left(\cdot, \cdot\right)$, defined in \eqref{eq:Finn}-\eqref{eq:fother},
as well as $F_{{\myin}'}$ and $F_{w'_{\min}}$ defined in \eqref{eq:extra},
are function of $\left({\bf{x}}, \phi\left({\bf{x}}\right)\right)$.
\end{lemma}

\begin{proof}[Proof of Lemma~\ref{lemma:consistency}]
Since  $\phi \in \C{\Oinf}$ and the computational domain $\Omega$ is bounded, $\phi$ has continuous and bounded  derivatives of up to second-order in $\Omega$.
Given the smooth test function $\phi$, with $j\in \J$ and $m \in \M$ being fixed
and $(i) \in \{(1), (2)\}$, we define discrete values $(\phi^{\mysup})_{l, d, j}^{m+}$,
$l \in \ND$ and $d \in \KD$, as follows
\begin{linenomath}
\postdisplaypenalty=0
\begin{align}
l \in \N \text{ and } d \in \K: & \quad (\phi^{\myloc})_{l, d, j}^{m+} = \sup_{\gamma_{l, d, j}^m \in [0, C_r \Delta \tau]} \pt_{l, d, j}^{m} + f(\gamma_{l, d, j}^m),~
(\phi^{\mynlc})_{l, d, j}^{m+} = \sup_{\gamma_{l, d, j}^m \in (C_r \Delta \tau, a_j]} \pt_{l, d, j}^{m} + f(\gamma_{l, d, j}^m),
\nonumber
\\
l \in \Nc \text{ or  } d \in \Kc: & \quad (\phi^{\myloc})_{l, d, j}^{m+} = (\phi^{\mynlc})_{l, d, j}^{m+} = \phi_{l, d, j}^m + \xi,
\label{eq:phi_tilde}
\end{align}
\end{linenomath}
where $\pt_{l, d, j}^{m}$ is given by
\EQ
\label{eq:phi_tilde_def}
\pt_{l, d, j}^{m} = \mathcal{I}\{\phi^{m} + \xi \}(\tilde{w}_l, r_d, \tilde{a}_j),
~\tilde{w}_l = \ln(\max(e^{w_l} - \gamma_{l, d, j}^{m}, e^{w^{\dagger}_{\min}})),
~ \tilde{a}_j = a_j - \gamma_{l, d, j}^{m}.
\EN
Given the discrete data $\left\{\left(\left(w_l, r_d, a_j \right), (\phi^{\mysup})_{l, d, j}^{m+}\right)\right\}$,
$(i) \in \{(1), (2)\}$, where $(\phi^{\mysup})_{l, d, j}^{m+}$, is given in \eqref{eq:phi_tilde}-\eqref{eq:phi_tilde_def}, we define associated discrete values $(\pisl)_{l,d,j}^{m}$ as follows
\begin{linenomath}
\begin{subequations}
\label{eq:psi_sl_plus}
\begin{empheq}[left={(\pisl)^{m+}_{l,d,j} = \empheqlbrace}]{alignat=3}
&\mathcal{I}\left\{(\phi^{\mysup})^{m+} \right\}(\bw_l, \br_d, a_j)
(1 + \Delta \tau r_d)^{-1} && \quad
l \in \N \text{ and } d \in \K
\label{eq:psi_sl_plus_aa}
\\
&\phi_{l, d, j}^m + \xi&& \quad \text{otherwise},
\label{eq:psi_sl_plus_bb}
\end{empheq}
\end{subequations}
\end{linenomath}
where the departure point $(\bw_l, \br_d)$ of \blue{an} SL trajectory are defined in \eqref{eq:semi_lag_results}.

We now show that the first equation of \eqref{eq:lemma_1} holds, that is, for ${\bf{x}} =(w_n,r_k,a_j, \tau_{m+1})$,
\EQAS
&&\mathcal{H}_{n,k,j}^{m+1}(\cdot) = \min\left\{ \mathcal{C}_{n,k,j}^{m+1} \left(\cdot\right), \mathcal{D}_{n,k,j}^{m+1} \left(\cdot\right) \right\}
= F_{\myin}\left({\bf{x}}, \phi\left({\bf{x}}\right)\right)
+ c\left({\bf{x}}\right)\xi
+ \mathcal{O}(h)
+ \errm
\\
&&\qquad\qquad
\text{if}~
w_{\min} < w_n < w_{\max},~ r_{\min} < r_k < r_{\max},
C_r \Delta \tau < a_j  \le a_J,~
0 < \tau_{m+1} \le T,
\ENAS
where operators $\mathcal{C}_{n,k,j}^{m+1}(\cdot)$ and $\mathcal{D}_{n,k,j}^{m+1}(\cdot)$ are defined in \eqref{eq:scheme_CD}.  First, we consider operator $\mathcal{C}_{n,k,j}^{m+1}(\cdot)$
which can be written as
\begin{linenomath}
\postdisplaypenalty=0
\begin{align}
\label{eq:C_operator}
\mathcal{C}_{n,k,j}^{m+1}\left(\cdot \right)
&=
\frac{1}{\Delta \tau}\bigg[ \phi_{n,k,j}^{m+1} + \xi
-
\Delta w \Delta r \mysum_{l \in \ND}^{d \in \KD}
\tilde{g}_{n-l, k-d}~(\plsl)^{m+}_{l,d,j} \bigg],
\end{align}
\end{linenomath}
where the discrete values $(\plsl)^{m+}_{l,d,j}$ are defined in \eqref{eq:psi_sl_plus} with $(i) = (1)$.

The key challenge in \eqref{eq:C_operator} is the discrete convolution $\mysum^{}_{}\tilde{g}~(\plsl)^{m+}_{l,d,j}$.
Our approach is to decompose it into the sum of two simpler discrete convolutions
of the forms $\mysum^{}_{} \tilde{g}~ (\pb)_{l,d,j}^{m}$ and $\mysum^{}_{} \tilde{g}~ (\vpsl)_{l,d,j}^{m}$
for which Lemmas \ref{lemma:error_smooth_sl} and \ref{lemma:ar}  are respectively applicable.
Here, $(\pb)_{l,d,j}^{m}$ is given in \eqref{eq:phi_term} and $(\vpsl)_{l,d,j}^{m}$ is to be defined subsequently.
To this end,  we  will start with the interpolated values $\pt_{l, d, j}^{m}$
in \eqref{eq:phi_tilde_def}.

For operator $\mathcal{C}_{n,k,j}^{m+1}(\cdot)$, the admissible control set is $\gamma_{l, d, j}^{m} \in [0, C_r\Delta \tau]$. In this case,  condition \eqref{eq:min_e_w} implies that, for $w_l \in (w_{\min}, w_{\max})$, $e^{w_l} - \gamma_{l, d, j}^{m} > e^{w^{\dagger}_{\min}}$ for all $\gamma_{l, d, j}^{m} \in [0, C_r\Delta \tau]$. Therefore,   we can eliminate the $\max(\cdot)$ operator in the linear interpolation operator
in \eqref{eq:phi_tilde_def} when $\gamma_{l, d, j}^{m} \in [0, C_r\Delta \tau]$.
Consequently, when $\gamma_{l, d, j}^{m} \in [0, C_r\Delta \tau]$, using \eqref{eq:phi_tilde} and
recalling the cash flow function $f(\cdot)$ defined in \eqref{eq:f_gamma_k_dis}, we have
\begin{linenomath}
\postdisplaypenalty=0
\begin{align}
\label{eq:sum_sum1}
\tilde{\phi}_{l, d, j}^{m} + f\l(\gamma_{l, d, j}^{m}\r)
&\overset{\text{(i)}}{=}
\phi\l(\ln\left(e^{w_l} - {\gamma}_{l, d,  j}^{m}\right),
            a_j - {\gamma}_{l, d, j}^{m}, \tau_{m}\r) + \xi
+
\mathcal{O}\l(h^2\r)
+ \gamma_{l, d, j}^{m}
\nonumber
\\
&\overset{\text{(ii)}}{=}
\phi_{l, d, j}^{m}
+
\xi
+ \gamma_{l, d, j}^{m}  \l(1- e^{-w_l} (\phi_w)_{l, d, j}^{m}
- (\phi_a)_{l, d, j}^{m}\r)
+ \mathcal{O}\l(h^2\r).
\end{align}
\end{linenomath}
Here, (i) follows from Remark~\ref{eq:intp}[eqns~\eqref{eq:interp_xi} and \eqref{eq:interp_sim}],
and $f\l(\gamma_{l, d, j}^{m}\r) = \gamma_{l, d, j}^{m}$ as defined in
\eqref{eq:f_gamma_k_dis};
and in~(ii), we  apply a Taylor series to expand
$\phi\l(\ln\left(e^{w_l} - {\gamma}_{l, d, j}^{m}\right), r_d,
            a_j - {\gamma}_{l, d, j}^{m}, \tau_{m}\r)$
about $(w_l, r_d, a_j, \tau_{m})$, noting
$\gamma_{l, d, j}^{m} = \mathcal{O}(\Delta \tau)$.
Therefore, using \eqref{eq:sum_sum1}, $\sup_{\gamma_{l, d, j}^{m} \in [0, C_r \Delta \tau]}
\tilde{\phi}_{l, d, j}^{m} + f(\gamma_{l, d, j}^{m}) = \ldots$
\begin{linenomath}
\postdisplaypenalty=0
\begin{align}
\label{eq:sum_sum}
\ldots &=
\phi_{l, d, j}^{m}
+
\xi
+  \mathcal{O}\l(h^2\r)
+
\sup_{\gamma_{l, d, j}^{m} \in [0, C_r \Delta \tau]}
\gamma_{l, d, j}^{m}  (1- e^{-w_l} (\phi_w)_{l, d, j}^{m}
- (\phi_a)_{l, d, j}^{m})
\nonumber
\\
&\overset{\text{(i)}}{=}\phi_{l, d, j}^{m}
+
\xi
+  \mathcal{O}(h^2)
+
\Delta \tau \sup_{\hat{\gamma}_{l, d, j}^{m} \in [0, C_r]}
\hat{\gamma}_{l, d, j}^{m}  (1- e^{-w_l} (\phi_w)_{l, d, j}^{m}
- (\phi_a)_{l, d, j}^{m}).
\end{align}
\end{linenomath}
Here, in (i) of \eqref{eq:sum_sum},  since the control $\gamma_{l, d, j}^{m}$ can be factored
out completely from the objective function
$\gamma_{l, d, j}^{m}  (1- e^{-w_l} (\phi_w)_{l, d, j}^{m}
- (\phi_a)_{l, d, j}^{m})$,  we define
a new control variable $\hat{\gamma}_{l, d, j}^{m} =
\gamma_{l, d, j}^{m}/\Delta \tau$ where $\hat{\gamma}_{l, d, j}^{m} \in [0, C_r]$.
We also note that, as a result of this change of control variable,
there is a factor of $\Delta \tau$ in front of the term
$\sup_{\hat{\gamma}_{l, d, j}^{m} \in [0, C_r]}(\cdot)$ in (i) of \eqref{eq:sum_sum}.

For subsequent use, letting ${\bf{x'}} = (w', r', a', \tau') \in \Oinf$,
we define a function $\varphi\l({\bf{x'}}\r)$ as follows
\EQA
\label{eq:phiprime}
\varphi\l({\bf{x'}}\r)
&=&
\l\{
\begin{array}{ll}
\ds \sup_{\hat{\gamma} \in [0, C_r]} \varphi'(\hat{\gamma}, {\bf{x'}}),  &
w_{\min} < w' < w_{\max},~r_{\min} < r' < r_{\max},
\\
~~~\text{where } \varphi'(\hat{\gamma}, {\bf{x'}}) = \hat{\gamma}(1 - e^{-w} \phi_w({\bf{x'}})
- \phi_a({\bf{x'}})) & C_r \Delta \tau < a'  \le a_J,~0 \le \tau' < T,
\\
0 & {\text{otherwise}}.
\end{array}
\right.
\ENA
Using \eqref{eq:sum_sum}-\eqref{eq:phiprime}, and recalling from  \eqref{eq:phi_tilde} that
$(\phi^{\myloc})_{l, d, j}^{m+} = \sup_{\gamma_{l, d, j}^{m} \in [0, C_r \Delta \tau]}
\tilde{\phi}_{l, d, j}^{m} + f(\gamma_{l, d, j}^{m})$, we have
\EQA
\label{eq:sum_sum_2}
(\phi^{\myloc})_{l, d, j}^{m+}
=
\phi_{l, d, j}^{m} + \xi + \Delta \tau  \varphi_{l, d, j}^{m}
+  \mathcal{O}\l(h^2\r), \quad  l \in \N, ~d \in \K.
\ENA
The decomposition formula \eqref{eq:sum_sum_2} allows us to write $(\plsl)^{m+}_{l,d,j}$,
defined in \eqref{eq:psi_sl_plus}, as follows
\EQ
\label{eq:sum_sum_2_dec}
(\plsl)^{m+}_{l,d,j} = (\pb)_{l,d,j}^{m} + (\vpsl)_{l,d,j}^{m} + \mathcal{O}\l(h^2\r),
\quad l \in \ND,~d \in \KD,
\EN
where $(\pb)_{l,d,q}^{m}$ is given in \eqref{eq:phi_term} and $(\vpsl)_{l,d,q}^{m}$ is given by
\begin{linenomath}
\begin{subequations}
\label{eq:psi_sl}
\begin{empheq}[left={(\vpsl)_{l,d,q}^{m} = \empheqlbrace}]{alignat=3}
&(\xi + \Delta \tau \mathcal{I} \{ \varphi^{m} \} (\bw_l, \br_d, a_j))
(1 + \Delta \tau r_d)^{-1} && \quad
l \in \N \text{ and } d \in \K,
\label{eq:psi_sl_a}
\\
&\xi&& \quad {\text{otherwise}},
\label{eq:psi_sl_b}
\end{empheq}
\end{subequations}
\end{linenomath}
where $\varphi$ is defined in \eqref{eq:phiprime}.
Using \eqref{eq:sum_sum_2_dec}-\eqref{eq:psi_sl}, we rewrite operator $\mathcal{C}_{n,k,j}^{m+1}(\cdot)$, previously given in \eqref{eq:C_operator}, into a convenient form below
\begin{linenomath}
\postdisplaypenalty=0
\begin{align}
\mathcal{C}_{n,k,j}^{m+1}(\cdot) &= \frac{1}{\Delta \tau}\bigg[ \phi_{n, k, j}^{m+1} + \xi
-
\Delta w \Delta r \mysum_{l \in \ND}^{d \in \KD} \tilde{g}_{n-l, k-d}
\l((\pb)_{l,d,j}^{m}
+  (\vpsl)_{l,d,j}^{m} + \mathcal{O}\l(h^2\r)\r)
\bigg].
\label{eq:term_cc_new}
\end{align}
\end{linenomath}
From here,  respectively applying Lemma~\ref{lemma:error_smooth_sl} and Lemma~\ref{lemma:ar}[equation~\eqref{eq:error_analysis_2}]
on discrete convolutions involving $(\pb)_{l,d,j}^{m}$ and $(\vpsl)_{l,d,j}^{m}$
gives
\begin{linenomath}
\postdisplaypenalty=0
\begin{align}
\Delta w \Delta r
\mysum^{d \in \KD}_{l \in \ND}
\tilde{g}_{n-l, k-d}~ (\pb)_{l,d,j}^{m} &=
\phi_{n,k,j}^{m}+ \Delta \tau \left[ \mathcal{L} \phi + \mathcal{J} \phi
\right]_{n,k,j}^{m} + \mathcal{O}( h^2 ) + \Delta \tau \myerrm{\phi},
\label{eq:term_c_1}
\\
\Delta w \Delta r
\mysum^{d \in \KD}_{l \in \ND}
\tilde{g}_{n-l, k-d}~ (\vpsl)_{l,d,j}^{m}
&= (\vpsl)_{n,k,j}^{m}  + \mathcal{O}( h^2 ) + \Delta \tau \myerrm{\varphi},
\label{eq:term_c_2}
\end{align}
\end{linenomath}
where $\myerrm{\phi}, \myerrm{\varphi}\to 0$ as $h\to 0$.

\noindent We now investigate the rhs of \eqref{eq:term_c_2}. By the definition of $(\vpsl)_{n,k,j}^{m}$ in \eqref{eq:psi_sl}, and since linear interpolation is used, together with \eqref{eq:simple},
we can further write the term  $(\vpsl)_{n,k,j}^{m}$ for the case \eqref{eq:psi_sl_a} as
\begin{linenomath}
\postdisplaypenalty=0
\begin{align}
(\xi + \Delta \tau \mathcal{I} \{ \varphi^{m} \} (\bw_n, \br_k, a_j)) (1 + \Delta \tau r_k)^{-1}
&= (\xi + \Delta \tau \mathcal{I} \{ \varphi^{m} \} (\bw_n, \br_k, a_j))
(1 - \Delta \tau r_k) + \mathcal{O}(h^2)
\nonumber
\\
&= \xi + \Delta \tau \mathcal{I} \{ \varphi^{m} \} (\bw_n, \br_k, a_j)
- \Delta \tau \xi r_k + \mathcal{O}(h^2).
\label{eq:term_xi_interp}
\end{align}
\end{linenomath}
Suppose that $w_{n'} \le \bw_n \le w_{n'+1}$ and $r_{k'} \le \br_k \le r_{k'+1}$.
Then, $\mathcal{I} \{ \varphi^{m} \} (\bw_n, \br_k, a_j)$ can be written into
\begin{linenomath}
\postdisplaypenalty=0
\begin{align}
\mathcal{I} \left\{ \varphi^{m} \right\} (\bw_n, \br_k, a_j)
&\overset{(i)}{=}
 x_r (x_w\varphi_{n', k', j}^{m} + (1-x_w)\varphi_{n'+1, k', j}^{m})
+ (1-x_r) (x_w\varphi_{n', k'+1, j}^{m} + (1-x_w)\varphi_{n'+1, k'+1, j}^{m}),
\nonumber
\\
&\overset{(ii)}{=}
\bigg[\sup_{\hat{\gamma} \in [0, C_r]} \hat{\gamma}(1 - e^{-w} \phi_w - \phi_a)\bigg]_{n,k,j}^{m}
 + \mathcal{O}(h).
\label{eq:varphi_interp_sl}
\end{align}
\end{linenomath}
Here, in (i), $0\le x_r \le 1$ and $0\le x_w \le 1$ are linear interpolation weights.
For (ii), we replace $\{\varphi_{n', k', j}^{m}, \ldots, \varphi_{n'+1, k'+1, j}^{m}\}$
by $\varphi_{n, k, j}^{m}$, resulting in an overall error of size $\mathcal{O}(h)$.
Specifically, as an example, replacing  $\varphi_{n', k', j}^{m}$ by $\varphi_{n, k, j}^{m}$
gives rise to an error bounded as follows
\begin{linenomath}
\postdisplaypenalty=0
\begin{align}
|\varphi_{n, k, j}^{m} - \varphi_{n', k', j}^{m}|
&\le\sup_{\hat{\gamma} \in [0, C_r]} \hat{\gamma}| e^{-w_n} (\phi_w)_{n, k, j}^m - e^{-w_{n'}} (\phi_w)_{n', k', j}^m
+ (\phi_a)_{n', k', j}^m) - (\phi_a)_{n, k, j}^m| = \mathcal{O}(h),
\label{eq:varphi_resl}
\end{align}
\end{linenomath}
due to smooth test function $\phi$ and boundedness of $\hat{\gamma} \in [0, C_r]$,
independently of $h$.

Substituting \eqref{eq:term_c_1}-\eqref{eq:term_c_2} and \eqref{eq:varphi_interp_sl} into \eqref{eq:term_cc_new},
and simplifying gives $\mathcal{C}_{n,k,j}^m(\cdot) = \ldots $
\begin{linenomath}
\postdisplaypenalty=0
\begin{align}
\ldots &=
\frac{\phi_{n, k, j}^{m+1} -  \phi_{n, k, j}^{m}}{\Delta \tau}
- \left[ \mathcal{L} \phi + \mathcal{J} \phi
+\sup_{\hat{\gamma} \in [0, C_r]} \hat{\gamma}(1 - e^{-w} \phi_w - \phi_a)
\right]_{n,k,j}^{m}
+ \xi r_k
+  \myerrm{}
+ \mathcal{O}(h)
\nonumber
\\
&\overset{(i)}{=}
\l[\phi_{\tau} - \mathcal{L} \phi - \mathcal{J} \phi - \sup_{\hat{\gamma} \in [0, C_r]} \hat{\gamma}\l(1 - e^{-w} \phi_w - \phi_a\r)\r]_{n,k,j}^{m+1}
+\xi r_k
+  \myerrm{} + \mathcal{O}(h).
\end{align}
\end{linenomath}
Here, in (i),  $\myerrm{}\to 0$ as $h \to 0$, and we use
\EQAS
(\phi_{\tau})_{n,k, j}^{m} = (\phi_{\tau})_{n,k, j}^{m+1} + \mathcal{O}\l(h\r),~ (\phi_w)_{n, k, j}^{m} = (\phi_w)_{n, k, j}^{m+1} + \mathcal{O}\l(h\r),~ (\phi_a)_{n,k,j}^{m} = (\phi_a)_{n,k,j}^{m+1} + \mathcal{O}\l(h\r).
\ENAS
This step results in an $\mathcal{O}\l(h\r)$ term inside $\sup_{\hat{\gamma}}\l(\cdot\r)$,
which can be moved out of the $\sup_{\hat{\gamma}}\l(\cdot\r)$, because it
has the form $C(\hat{\gamma})h$, where $C(\hat{\gamma})$ is bounded independently of $h$,
due to boundedness of $\hat{\gamma}\in [0, C_r]$ independently of $h$.

\medskip
\noindent We now consider operator $\mathcal{D}_{n, k, j}^{m+1}(\cdot)$ which can be written as
\begin{linenomath}
\postdisplaypenalty=0
\begin{align}
\label{eq:D_operator}
\mathcal{D}_{n,k,j}^{m+1}\left(\cdot \right)
&=
\phi_{n,k,j}^{m+1} + \xi
-
\Delta w \Delta r \mysum_{l \in \ND}^{d \in \KD} \tilde{g}_{n-l, k-d}~(\pnsl)^{m+}_{l,d,j},
\end{align}
\end{linenomath}
where the discrete values $(\pnsl)^{m+}_{l,d,j}$ are defined in \eqref{eq:psi_sl_plus} with $(i) = (2)$. Adopting a similar approach as the one utilized for  $\mathcal{C}_{n, k, j}^{m+1}(\cdot)$,
we aim to decompose $\mysum^{}_{}\tilde{g}~(\pnsl)^{m+}_{l,d,j}$ into
$\mysum^{}_{} \tilde{g}~ (\psisl)_{l,d,j}^{m}$
for which Lemma~\ref{lemma:error_smooth_sl} is applicable.
Here, $(\psisl)_{l,d,j}^{m}$ is to be defined subsequently.

We first start from the interpolated value $\pt_{l, d, j}^{m}$ in \eqref{eq:phi_tilde}.
In this case, since $\gamma_{l, d, j}^{m} \in (C_r\Delta \tau, a_j]$, we cannot eliminate the $\max(\cdot)$ operator in $\tilde{w}_l$
of the linear interpolation in \eqref{eq:phi_tilde}. Therefore,
as noted in Remark~\ref{eq:intp}[\eqref{eq:interp_xi}-\eqref{eq:interp_sim}],
for $\gamma \in (C_r\Delta \tau, a_j]$, we have $\sup_{\gamma_{l, d, j}^{m}\in (C_r\Delta \tau, a_j]} \pt_{l, d, j}^{m} + f(\gamma_{l, d, j}^{m}) = \ldots$
\begin{linenomath}
\postdisplaypenalty=0
\EQ
\label{eq:Dphi}
\ldots =
\sup_{\gamma_{l, d, j}^{m}\in (C_r\Delta \tau, a_j]}
(\phi(\tilde{w}_l, r_d,\tilde{a}_j, \tau_m)
+ \gamma_{l, d, j}^{m} (1-\mu))
+ \xi
+ \mu C_r \Delta \tau - c
+ \mathcal{O}(h^2).
\EN
\end{linenomath}
Here, $(\tilde{w}_l, \tilde{a}_j)$ is given in \eqref{eq:phi_tilde},
and $f(\gamma)$ is replaced by $\gamma(1 - \mu) + \mu C_r\Delta \tau - c$, as per \eqref{eq:f_gamma_k_dis}
for $\gamma \in (C_r\Delta \tau, a_j]$.

Recalling operator $\mathcal{M}(\cdot)$ defined in \eqref{eq:Operator_M_b},
we define a function $\psi\l({\bf{x'}}\r)$ as follows
\begin{linenomath}
\begin{subequations}\label{eq:phidoubleprime}
\begin{empheq}[left={\psi\l({\bf{x'}}\r) = \empheqlbrace}]{alignat=3}
&\sup_{\gamma \in [0, a']} \psi'(\gamma, {\bf{x'}}) &&&& \quad
w_{\min} < w' < w_{\max},~ r_{\min} < r' < r_{\max},
\label{eq:phidoubleprime_a}
\\
&\quad \text{where } \psi'(\gamma, {\bf{x'}}) = \mathcal{M} (\gamma) \phi ({\bf{x'}}) + \mu C_r \Delta \tau &&&&\quad C_r \Delta \tau < a'  \le a_J,~
0 \le \tau' < T,
\nonumber
\\
&\phi ({\bf{x'}}) &&&& \quad {\text{otherwise}}.
\label{eq:phidoubleprime_b}
\end{empheq}
\end{subequations}
\end{linenomath}
We note that in \eqref{eq:phidoubleprime_a}, the admissible control set is
$\gamma \in [0, a']$. It is straightforward to show that, for
a fixed ${\bf{x'}} \in \Omega$ satisfies \eqref{eq:phidoubleprime_a},
function $\psi'\l(\gamma; {\bf{x'}}\r)$ defined in \eqref{eq:phidoubleprime_a}
is (uniformly) continuous in $\gamma \in [0, a']$. Hence,
for the case \eqref{eq:phidoubleprime_a}
\EQA
\label{eq:imed}
\sup_{\gamma \in \l(C_r \Delta \tau, a' \r] }
\psi' \l(\gamma, {\bf{x'}}\r)
-
\sup_{\gamma \in \l(0, a' \r] }
\psi' \l(\gamma, {\bf{x'}}\r)
=
\max_{\gamma \in \l[C_r \Delta \tau, a' \r] }
\psi' \l(\gamma, {\bf{x'}}\r)
-
\max_{\gamma \in \l[0, a' \r] }
\psi' \l(\gamma, {\bf{x'}}\r)
=
\mathcal{O}\l(h\r),
\ENA
since the difference of the optimal values of $\gamma$ for the two $\max(\cdot)$ expressions
is bounded by $C_r\Delta \tau~=~\mathcal{O}(h)$.

Using \eqref{eq:phidoubleprime_a} and \eqref{eq:imed}, and
recalling from  \eqref{eq:phi_tilde} that
$(\phi^{\mynlc})_{l, d, j}^{m+} = \sup_{\gamma_{l, d, j}^{m} \in (C_r \Delta \tau, a_j]}
\tilde{\phi}_{l, d, j}^{m} + f(\gamma_{l, d, j}^{m})$, we have
\EQA
\label{eq:sum_sum_2_D}
(\phi^{\mynlc})_{l, d, j}^{m+}
=
\xi + (\psi)_{l,d,j}^{m} + \mathcal{O}(h), \quad  l \in \N, ~d \in \K,
\ENA
where $\psi$ is given in \eqref{eq:phidoubleprime_a}.
Equation \eqref{eq:sum_sum_2_D} allows us to write $(\pnsl)^{m+}_{l,d,j}$,
defined in \eqref{eq:psi_sl_plus}, as follows
\EQ
\label{eq:sum_sum_2_dec_d}
(\pnsl)^{m+}_{l,d,j} =  (\psisl)_{l,d,j}^{m} + \mathcal{O}\l(h\r),
\quad l \in \ND \text{ and } d \in \KD,
\EN
where
\EQA
\label{eq:phi_term_D}
\begin{array}{l}
(\psisl)_{l,d,q}^{m}
= \left\{
\begin{array}{lll}
(\xi + \mathcal{I} \{ (\psi)^{m}\} (\bw_l, \br_d, a_q)) (1 + \Delta \tau r_d)^{-1}
& l \in \N \text{ and } d \in \K,
\\
\phi_{l,d,q}^{m} + \xi
& \text{otherwise},
\end{array}
\right.
\end{array}
\ENA
where $\psi$ is defined in \eqref{eq:phidoubleprime}.
Using \eqref{eq:sum_sum_2_dec_d},
we rewrite operator $\mathcal{D}_{n,k,j}^{m+1}(\cdot)$, previously given in \eqref{eq:D_operator},
into a convenient form below
\begin{linenomath}
\postdisplaypenalty=0
\begin{align}
\mathcal{D}_{n,k,j}^{m+1}(\cdot) &= \phi_{n, k, j}^{m+1} + \xi
-
\Delta w \Delta r \mysum_{l \in \ND}^{d \in \KD} \tilde{g}_{n-l, k-d}
(\psisl)_{l,d,j}^{m} + \mathcal{O}\l(h\r).
\label{eq:D_oper}
\end{align}
\end{linenomath}
Then, for the above discrete convolution,
applying Lemma~\ref{lemma:ar}[eqn~\eqref{eq:error_analysis_2}],
noting \eqref{eq:simple}, gives
\begin{linenomath}
\postdisplaypenalty=0
\begin{align}
\Delta w \Delta r
\mysum^{d \in \KD}_{l \in \ND}
\tilde{g}_{n-l, k-d}~ (\psisl)_{l,d,j}^{m} &=
(\psisl)_{n,d,j}^{m}  +  \myerrm{\psi} + \mathcal{O}(h),
\nonumber
\\
&= \xi + \mathcal{I} \{ (\psi)^{m}\} (\bw_n, \br_k, a_j) + \myerrm{\psi} + \mathcal{O}(h),
\label{eq:term_c_1_varphiprime}
\end{align}
\end{linenomath}
where we used the  definition of $(\psisl)_{n,k,j}^{m}$ in \eqref{eq:phi_term_D},
and $\myerrm{\psi}\to 0$ as $h\to 0$.

For the term $\mathcal{I} \{ (\psi)^{m}\} (\bw_n, \br_k, a_j)$ in
\eqref{eq:term_c_1_varphiprime}, following the same arguments as those for
\eqref{eq:varphi_interp_sl}-\eqref{eq:varphi_resl}, noting the definition of $\psi$ in \eqref{eq:phidoubleprime},
we obtain
\begin{linenomath}
\postdisplaypenalty=0
\begin{align}
\mathcal{I} \{ (\psi)^{m}\} (\bw_n, \br_k, a_j)
&=
\sup_{\gamma\in [0, a_j]} \mathcal{M} (\gamma) \phi ({\bf{x}}_{n, k, j}^m) + \mu C_r \Delta \tau
+ \mathcal{O}(h) + \myerrm{\psi}
\nonumber
\\
&=\sup_{\gamma\in [0, a_j]} \mathcal{M} (\gamma) \phi ({\bf{x}}_{n, k, j}^{m+1}) + \mathcal{O}(h) + \myerrm{\psi}.
\label{eq:phi_d}
\end{align}
\end{linenomath}
Here, $\mathcal{M}(\gamma) \phi \l({\bf{x}}_{n, k, j}^{m}\r)  = \mathcal{M}(\gamma) \phi \l({\bf{x}}_{n, k, j}^{m+1}\r) +  \mathcal{O}\l(h\r)$, which is combined with $\mu C_r\Delta \tau = \mathcal{O}\l(h\r)$.
Substituting \eqref{eq:term_c_1_varphiprime} and \eqref{eq:phi_d} into \eqref{eq:D_oper} gives
\EQA
\label{eq:D}
\mathcal{D}_{n, j}^{m+1} \l(\cdot\r)
=
\phi_{n, k, j}^{m+1}
-
\sup_{\gamma \in \l[ 0, a \r]}\mathcal{M}(\gamma) \phi \l({\bf{x}}_{n, k,  j}^{m+1}\r) + \mathcal{O}(h) + \myerrm{}.
\ENA


Overall, recalling ${\bf{x}} = {\bf{x}}_{n,k,j}^{m+1}$,  we have
\EQA
&&\mathcal{H}_{n,k,j}^{m+1}
\left(h, \phi_{n,k,j}^{m+1}+ \xi,
\left\{\phi_{l,d,p}^{m}+\xi \right\}_{p \le j}\right)
-
F_{\myin}\left(
{\bf{x}},
\phi\left({\bf{x}}\right),
D\phi\left({\bf{x}}\right),
D^2\phi\left({\bf{x}}\right),
\mathcal{J} \phi \left({\bf{x}}\right),
\mathcal{M} \phi \left({\bf{x}}\right)
\right) \nonumber
\\
& &
\qquad\qquad
=~
c\left({\bf{x}}\right)\xi
+
\mathcal{O}(h)
+
\err,
\quad
\text{if}~ {\bf{x}} \in \Omega^{\myup}_{\myin},
\nonumber
\ENA
where $c({\bf{x}})$ is a bounded function satisfying $r_{\min} \le c({\bf{x}}) \le r_{\max}$
and
$\err  \to 0$ as $h \to 0$.
This proves the first equation in (\ref{eq:lemma_1}).
The remaining equations in (\ref{eq:lemma_1}) can be proved using similar arguments with the first equation.
\end{proof}

\begin{remark}
\label{rm:consistency_pre_con}
We impose the condition \eqref{eq:min_e_w} to ease the presentation of the proof, that is, we
make sure the term $\max(e^{w_l} - \gamma_{l,d,j}^{m}, e^{w^{\dagger}_{\min}})$ in the operator $\mathcal{C}^{m+1}_{n,k,j}(\cdot)$ will never be triggered. However,
we can avoid this condition by the similar procedures presented in \cite{online}.
\end{remark}

\begin{lemma} [Consistency]
\label{lemma:consistency_viscosity}
Assuming all the conditions in Lemma~\ref{lemma:consistency} are satisfied, then the scheme \eqref{eq:scheme_G} is consistent in the viscosity sense to the impulse control problem \eqref{def:impulse_def} in $\Oinf$. That is,
for all ${\bf{\hat{x}}} = (\hat{w}, \hat{r}, \hat{a}, \hat{\tau}) \in \Oinf$,
and for any $\phi\in \G \cap \C{\Oinf}$
with $\phi_{n,k,j}^{m+1} = \phi\left(w_n,r_k,a_j, \tau_{m+1}\right)$
and {\bf{x}}~=~$(w_n,r_k,a_j, \tau_{m+1})$, we have both of the following
\EQA
\limsup_{\subalign{h \to 0, & ~  {\bf{x}} \to {\bf{\hat{x}}} \\ \xi &\to 0}}
\mathcal{H}_{n,k,j}^{m+1}
\!\left(h, \phi_{n,k,j}^{m+1}\!+\! \xi,
\left\{\phi_{l,d,p}^{m}\!+\!\xi \right\}_{p \le j} \right)
\leq
\left(F_{\Oinf}\right)^* \!\left(
              {\bf{\hat{x}}}, \phi({\bf{\hat{x}}}), D\phi({\bf{\hat{x}}}), D^2 \phi({\bf{\hat{x}}}),
             \mathcal{J} \phi ({\bf{\hat{x}}}),
             \mathcal{M} \phi ({\bf{\hat{x}}})
             \right),
\label{eq:consistency_viscosity_1}
\\
\liminf_{\subalign{h \to 0, & ~ {\bf{x}} \to {\bf{\hat{x}}} \\ \xi &\to 0}}
\mathcal{H}_{n,k,j}^{m+1}
\!\left(h, \phi_{n,k,j}^{m+1}\!+\! \xi,
\left\{\phi_{l,d,p}^{m}\!+\!\xi \right\}_{p \le j} \right)
\geq
\left(F_{\Oinf}\right)_*\left(
              {\bf{\hat{x}}}, \phi({\bf{\hat{x}}}), D\phi({\bf{\hat{x}}}), D^2 \phi({\bf{\hat{x}}}),
             \mathcal{J} \phi ({\bf{\hat{x}}}),
             \mathcal{M} \phi ({\bf{\hat{x}}})
             \right).
\label{eq:consistency_viscosity_2}
\ENA
\end{lemma}
\begin{proof}[Proof of Lemma~\ref{lemma:consistency_viscosity}]
Lemma~\ref{lemma:consistency_viscosity} can be proved using similar steps in
Lemma~5.5 in \cite{online}. For brevity, we outline key steps to prove
\eqref{eq:consistency_viscosity_1} for $\Omega_{\myin}$ and $\Omega_{a_{\min}}$;
other sub-domains can be treated similarly. We note the continuity in their parameters of operators defined in \eqref{eq:ftau0}-\eqref{eq:Finn},
which is needed for this proof.

Consider $\hat{\bf{x}} \in \Omega_{\myin}$. There exist sequences of discretization parameter
$\{h_i\}_i \to 0$, constants $\{\xi_i\}_i \to 0$, and gridpoints $\{(w_{n_i}, r_{k_i}, a_{j_i}, \tau_{m_i+1})\}_i
\equiv {\bf{x}}_i \to \hat{\bf{x}}$, as $i \to \infty$.
For sufficiently small $\{\Delta \tau_i\}_i$, we assume $a_{j_i} \in (C_r\Delta \tau_i, a_{\max}]$ for each $i$,
and hence, the sequence $\{{\bf{x}}_i\}_i$ is contained in $\Omega_{\myin}^{\myup}$, defined in \eqref{eq:new_dom}.
Therefore, lhs of \eqref{eq:consistency_viscosity_1} $= \limsup_{i \to \infty} \mathcal{H}_{n_i, k_i, j_i}^{m_i+1}
\!(h_i, \phi_{n_i, k_i, j_i}^{m_i+1}+\!\xi_i,
\{\phi_{l_i, d_i,  p_i}^{m_i}+\!\xi_i\}_{p_i\le j_i}
)\ldots$
\begin{linenomath}
\postdisplaypenalty=0
\begin{align*}
\ldots\underset{\text{(i)}}{\le}
\limsup_{i \to \infty}
F_{\myin} ({\bf{x}}_i, \phi({\bf{x}}_i))
+ \limsup_{i \to \infty} [ c({\bf{x}}_i)\xi_i + \mathcal{O}( h_i )
+\error({\bf{x}}_{n_i,j_i}^{m_i}, h_i)]
\underset{\text{(ii)}}{=}
F_{\myin}({\bf{\hat{x}}}, \phi({\bf{\hat{x}}}))
=
\text{rhs of \eqref{eq:consistency_viscosity_1}},
\end{align*}
\end{linenomath}
as wanted. Here, (i) is due to the local consistency result for $\Omega_{\myin}^{\myup}$ in
the first equation of \eqref{eq:lemma_1} (Lemma~\ref{lemma:consistency}), and properties
of $\limsup$; (ii) is because of continuity of $F_{\myin}$.

For $\hat{\bf{x}} \in \Omega_{a_{\min}}$, complications arise because $\{{\bf{x}}\}_i$ could converge
to $\hat{\bf{x}}$ from two different sub-domains,
$\Omega^{\myin} = \Omega^{\myup}_{\myin} \cup \Omega^{\mydn}_{\myin}$ and $\Omega_{a_{\min}}$;
however, on $\Omega^{\mydn}_{\myin}$, the second equation of \eqref{eq:lemma_1} (Lemma~\ref{lemma:consistency}) indicates local consistency with $F_{\myin}' ({\bf{x}}_i, \phi({\bf{x}}_i))$, defined in \eqref{eq:extra}
but is not part of $F_{\Oinf}$. Nonetheless, since
$\sup_{\hat{\gamma} \in [0, a/\Delta \tau]} \hat{\gamma}\l(1  - e^{-w}\phi_w - \phi_a \r) \geq 0$,
$F_{\myin}' ({\bf{x}}_i, \phi({\bf{x}}_i)) \le F_{a_{\min}} ({\bf{x}}_i, \phi({\bf{x}}_i))$,
we can eliminate $F_{\myin}' ({\bf{x}}_i, \phi({\bf{x}}_i))$ when considering $\limsup$.
Thus, lhs of \eqref{eq:consistency_viscosity_1} $= \limsup_{i \to \infty} \mathcal{H}_{n_i, k_i, j_i}^{m_i+1}
\!(h_i, \phi_{n_i, k_i, j_i}^{m_i+1}+\!\xi_i,
\{\phi_{l_i, d_i,  p_i}^{m_i}+\!\xi_i\}_{p_i\le j_i}
) \ldots$
\EQAS
\ldots\le
\limsup_{i \to \infty}
F_{\Oinf} ({\bf{x}}_i, \phi({\bf{x}}_i)) + \limsup_{i \to \infty} [ c({\bf{x}}_i)\xi_i
+~\error({\bf{x}}_{n_i,j_i}^{m_i}, h_i) ]
\le
\left(F_{\Oinf}\right)^* ({\bf{\hat{x}}}, \phi({\bf{\hat{x}}}))
= \text{rhs of \eqref{eq:consistency_viscosity_1}}.
\ENAS
\end{proof}

\subsection{Monotonicity}
We present a result on the monotonicity of scheme \eqref{eq:scheme_G}.
\begin{lemma} [$\epsilon$-monotonicity]
\label{lemma:ep_mo}
Suppose that (i) the discretization \eqref{eq:scheme_timestep_left} satisfies the positive coefficient condition \eqref{eq:pos_con}, and (ii) linear interpolation
in \eqref{eq:vtil_a}, \eqref{eq:vtil_b}
and (ii) the weight $\tilde{g}_{n-l, k-d}$
satisfies the monotonicity condition \eqref{eq:test1};
and (iii) $r_{\min}$ satisfies condition~\eqref{eq:r_min}.
Then scheme \eqref{eq:scheme_G} satisfies
\EQA
\label{eq:eps_h}
\mathcal{H}_{n,k,j}^{m+1}
\left(h, v_{n,k,j}^{m+1},
\left\{x_{l,d,p}^{m}\right\}_{p \le j}
  \right)
~\le~
\mathcal{H}_{n,k,j}^{m+1}
\left(h, v_{n,k,j}^{m+1},
\left\{y_{l,d,p}^{m}\right\}_{p \le j}
  \right)
~+~
K' \epsilon
\ENA
for bounded $\{x_{l,d,p}^{m}\}$ and $\{y_{l,d,p}^{m}\}$ having
$\{x_{l,d,p}^{m}\}  ~\geq~ \{y_{l,d,p}^{m}\}$, where
the inequality is understood in the component-wise sense, and $K'$ is a positive constant independent of $h$.
\end{lemma}
\noindent A proof of Lemma~\ref{lemma:ep_mo} is similar to that of Lemma~5.6 in \cite{online}, and hence
omitted for brevity.

\subsection{Convergence to viscosity solution}
\label{ssc:conv}
We have demonstrated that the scheme~\eqref{eq:scheme_G} satisfies the three key properties in $\Omega$:
(i) $\ell_{\infty}$-stability (Lemma~\ref{lemma:stability}), (ii) consistency (Lemma~\ref{lemma:consistency_viscosity})
and (iii) $\epsilon$-monotonicity (Lemma~\ref{lemma:ep_mo}).
With a strong comparison result in
$\Omega_{\myin} \cup \Omega_{a_{\min}}$, we now present the main convergence result of the paper.

\begin{theorem} [Convergence in $\Omega_{\myin} \cup \Omega_{a_{\min}}$]
\label{thm:convergence}
Suppose that all the conditions for Lemmas~\ref{lemma:stability}, \ref{lemma:consistency_viscosity} and
\ref{lemma:ep_mo} are satisfied.
Under the assumption that the monotonicity tolerance $\epsilon \to 0$ as $h \to 0$,
scheme~\eqref{eq:scheme_G}  converges locally uniformly in $\Omega_{\myin} \cup \Omega_{a_{\min}}$
to the unique bounded viscosity solution of the GMWB pricing problem in the sense of Definition~\ref{def:vis_def_common}.
\end{theorem}
\begin{proof}[Proof of Theorem~\ref{thm:convergence}]
To highlight the importance of the discretization parameter $h$, we let
${\bf{x}}_{n,k,j}^m(h) = (w_n, r_k, a_j, \tau_m; h)$, and denote by $v_{n,k,j}^m(h)$
the numerical solution at this node.
The candidate for the viscosity subsolution (resp.\ supersolution)
the GMWB pricing problem is given by
the u.s.c function $\overline{v}: \Oinf \to \mathcal{\mathbb{R}}$ (resp.\
the l.s.c function $\underline{v}: \Oinf \to \mathcal{\mathbb{R}}$)
defined as follows
\EQA
\label{eq:def_v}
\overline{v} \l( {\bf{x}} \r)
= \limsup_{\subalign{h &\to 0 \\ {\bf{x}}_{n, k, j}^{m+1}(h) & \to {\bf{x}}}}v_{n, k, j}^{m+1}(h)
\qquad
(\text{resp.}~~\underline{v} ({\bf{x}})
=
\ds\liminf_{\subalign{h &\to 0 \\ {\bf{x}}_{n, k, j}^{m+1}(h) &\to {\bf{x}}}} v_{n, k, j}^{m+1}(h))
\qquad{\bf{x}} \in \Oinf.
\ENA
Here,  $\limsup$ and $\liminf$ are finite due to stability of our scheme in $\Omega$ established in  Lemma~\ref{lemma:stability}.

We appeal to a Barles-Souganidis-type analysis in \cite{barles-burdeau:1995, barles-souganidis:1991} to show that $\overline{v}$ (resp.\ $\underline{v}$) is a viscosity subsolution (resp.\ supersolution) of the HJB-QVI~\eqref{eq:gmwb_def} in $\Oinf$ in the sense of Definition~\ref{def:vis_def_common}.
In this step, we use (i) $\ell_{\infty}$-stability (Lemma~\ref{lemma:stability}), (ii) consistency (Lemma~\ref{lemma:consistency_viscosity}) and (iii) $\epsilon$-monotonicity (Lemma~\ref{lemma:ep_mo})
of the numerical scheme, noting the requirement $\epsilon \to 0$ as $h \to 0$.
By \ref{eq:def_v}, $\overline{v} \geq \underline{v}$ in $\Oinf$.
By a strong comparison result in Theorem~\ref{thm:comparison},  $\overline{v} \leq \underline{v}$ in $\Omega_{\myin} \cup \Omega_{a_{\min}}$. Therefore, $ v({\bf{x}}) = \overline{v}({\bf{x}})= \underline{v}({\bf{x}})$
is the unique viscosity solution in $\Omega_{\myin} \cup \Omega_{a_{\min}}$
in the sense of Definition~\ref{def:vis_def_common}.  The fact that convergence is locally
uniform is automatically implied. This concludes the proof.
\end{proof}

\section{Numerical experiments}
\label{sec:num_test}
In this section, we present selected numerical results for the no-arbitrage pricing problem \eqref{eq:gmwb_def}.
In addition to validation examples, we particularly focus on investigating the impact of jump-diffusion dynamics
and stochastic interest rates on the prices/the fair insurance fees, as well as  on the holder's
optimal withdrawal behaviors.

A set of GMWB parameters commonly used for subsequent experiments is given in Table~\ref{tab:parameter_gmwb}.
These include expiry time $T$, the maximum allowed withdrawal rate $C_r$ (for continuous withdrawals),
the  proportional penalty rate $\mu$ (for withdrawing finite amounts),
the premium $z_0$ which is also the initial balance of the guarantee account
and of the personal sub-account.

For experiments in this section, the computational domain is constructed with
$w_{\min} = \ln(z_0)-10$, $w_{\max} = \ln(z_0)+10$, $r_{\min} = -0.2$, $r_{\max} = 0.3$,
together with $w^{\dagger}_{\min}$, $w^{\dagger}_{\max}$, $r^{\dagger}_{\min}$, and $r^{\dagger}_{\max}$
computed as discussed in Section~\ref{sec:num}. Unless otherwise stated, relevant details about the refinement levels are given in Table~\ref{tab:grid_size}.  Here, the timestep $M=20$ (resp.\ $M = 40$) corresponds to the case of $T=5$ (resp.\ $T = 10$) in Table~\ref{tab:parameter_gmwb}. Based on the choices of $N$ and $K$,
we have $N^{\dagger}  = 2N$ and $K^{\dagger} = 2K$ as in \eqref{eq:grid_w} and \eqref{eq:grid_r}, respectively.
We emphasize that, increasing $|w_{\min}|$, $w_{\max}$, $|r_{\min}|$, or $r_{\max}$ virtually does not change the
no-arbitrage prices/fair insurance fees. Therefore, for practical purposes, with $P^{\dagger} \equiv w^{\dagger}_{\max} - w^{\dagger}_{\min} $ and $K^{\dagger} \equiv r^{\dagger}_{\max} - r^{\dagger}_{\min}$ chosen sufficiently large as above, they can be kept constant for all refinement levels (as we let $h \to 0$).

\medskip
Similar to \cite{chen08a, huang:2010, online}, a sufficiently small fixed cost $c = 10^{-8}$ is used all numerical tests.  For user-defined tolerances $\epsilon$ and $\epsilon_{1}$ in Algorithm~\ref{alg:Gtilde},
we use $\epsilon = \epsilon_1 = 10^{-6}$ for all experiments and all refinement levels.
We note that using smaller $\epsilon$ or $\epsilon_1$ produces virtually identical numerical results.

\vspace{+1em}
\noindent
\begin{minipage}{0.52\textwidth}
\begin{tabular}{ll}
\hline
Parameter                         & Value      \\ \hline
Expiry time ($T$)                     & \{5, 10\} years \\
Maximum withdrawal rate ($C_r$)     & $1/T$    \\
Withdrawal penalty rate ($\mu$)           & 0.10       \\
Init.\ lump-sum premium ($z_0$)    & 100        \\
Init.\ balance of guarantee a/c ($=z_0$) & 100        \\
Init.\ balance value of sub-a/c ($=z_0$)         & 100        \\ \hline
\end{tabular}
\captionof{table}{GMWB parameters for numerical experiments.}
\label{tab:parameter_gmwb}
\end{minipage}
\hfill
\begin{minipage}{0.47\textwidth}
\center
\begin{tabular}{ccccl}
\hline
Refinement & $N$      & $K$   & $J$  &        $M$ \\
level      & ($w$)  & ($r$) & ($a$) & ($\tau$) \\ \hline
0     & $2^{9}$  & $2^{5}$   & 26        & \{20, 40\}       \\
1     & $2^{10}$  & $2^{6}$   & 51       & \{40, 80\}    \\
2     & $2^{11}$  & $2^{7}$  & 101       & \{80, 160\}       \\
3     & $2^{12}$  & $2^{8}$  & 201       & \{160, 320\}       \\
4     & $2^{13}$  & $2^{9}$   & 401       & \{320, 640\}       \\ \hline
\end{tabular}
\captionof{table}{Grid and timestep
refinement levels for numerical experiments.}
\label{tab:grid_size}
\end{minipage}

\bigskip
\noindent Unless otherwise stated, representative parameters to jump-diffusion dynamics
and Vasicek short rate dynamics are respectively given in Tables~\ref{tab:jump_parameters_vasicek}
(taken from \cite{online}) and \ref{tab:jump_parameters_validation} (from \cite{Pavel2015}).

\medskip
\begin{minipage}{0.52\textwidth}
\center
\setlength{\tabcolsep}{4pt} 
\renewcommand{\arraystretch}{0.9} 

\medskip
\begin{tabular}{lrrr}
\hline
Parameters                                 & Merton  & Kou      \\ \hline
$\sigz$ (risky asset volatility)            & 0.3  & 0.3  \\
$\lambda$ (jump intensity)                 & 0.1  & 0.1  \\
$\nu$(log jump multiplier mean)            & -0.9   & n/a     \\
$\varsigma$ (log jump multiplier std)     & 0.45  & n/a     \\
$p_u$ (probability of up-jump)           & n/a     & 0.3445  \\
$\eta_u$ (exp.\ parameter up-jump)    & n/a     & 3.0465  \\
$\eta_d$ (exp.\ parameter down-jump)  & n/a     & 3.0775  \\ \hline
\end{tabular}
\captionof{table}{Parameters for the jump-diffusion dynamics~\eqref{eq:Z_dynamics}.
Values are taken from \cite{online}.}
\label{tab:jump_parameters_validation}
\vfill
\end{minipage}
\hfill
\begin{minipage}{0.40\textwidth}
\smallskip
\center
\setlength{\tabcolsep}{4pt} 
\renewcommand{\arraystretch}{0.9} 
\begin{tabular}{ll}
\hline
Parameters  & Vasicek \\ \hline
$r_0$  & 0.05 \\
$\theta$ & 0.05  \\
$\delta$  & 0.0349 \\
$\sigr$  & 0.02 \\ \hline
\end{tabular}
\captionof{table}{Parameters for
the Vasicek short rate dynamics~\eqref{eq:R_dynamics}.
Values are taken from \cite{Pavel2015}.
}
\label{tab:jump_parameters_vasicek}
\vspace{+0.90cm}
\end{minipage}

\medskip
The correlation coefficient $\rho$ is chosen from
$\{-0.2, 0.2\}$. The value for $\rho$  will be specified
for each experiment subsequently.

\subsection{Validation through Monte Carlo simulation}
\label{ssc:valid}
As previously mentioned, the no-arbitrage pricing of GMWB with continuous withdrawals under a jump-diffusion dynamics with with stochastic interest rate  has not been previously studied in the literature,
hence, reference prices/insurance fees are not available for the dynamics considered in this work.
Therefore,  for validation purposes, we compare no-arbitrage prices obtained by
the proposed numerical method, hereafter referred to as ``$\epsilon$-mF'',
with those obtained by MC simulation.
\begin{table}[htb!]
\center
\begin{tabular}{cccccccccc}
\hline
\multirow{3}{*}{Method}        & \multirow{3}{*}{Level} & \multicolumn{4}{c}{Merton}                                                                                    & \multicolumn{4}{c}{Kou}                                                                                       \\ \cline{3-10}
                               &                        & \multicolumn{2}{c}{$\rho=-0.2$}                       & \multicolumn{2}{c}{$\rho=0.2$}                        & \multicolumn{2}{c}{$\rho=-0.2$}                       & \multicolumn{2}{c}{$\rho=0.2$}                        \\ \cline{3-10}
                               &                        & \multicolumn{1}{c}{price} & \multicolumn{1}{c}{ratio} & \multicolumn{1}{c}{price} & \multicolumn{1}{c}{ratio} & \multicolumn{1}{c}{price} & \multicolumn{1}{c}{ratio} & \multicolumn{1}{c}{price} & \multicolumn{1}{c}{ratio} \\ \hline
\multirow{5}{*}{$\epsilon$-mF} & 0                      & 115.4845                  &                           & 116.4466                  &                           & 109.1908                  &                           & 110.1039                  &                           \\
                               & 1                      & 114.2267                  &                           & 114.8675                 &                           & 109.1608                  &                           & 109.7832                  &                           \\
                               & 2                      & 113.6613                  & 2.22                      & 114.1549                 & 2.22                      & 109.1517                  & 3.29                      & 109.6396                  & 2.23                      \\
                               & 3                      & 113.3921                  & 2.10                      & 113.8171                  & 2.11                      & 109.1483                  & 2.62                      & 109.5719                  & 2.12                      \\
                               & 4                      & 113.2601                  & 2.04                      & 113.6524                  & 2.05                      & 109.1467                  & 2.27                      & 109.5388                  & 2.05                      \\ \hline
MC                             & 95\%-CI                & \multicolumn{2}{c}{{[}112.61,   113.47{]}}        & \multicolumn{2}{c}{{[}112.95,   113.79{]}}        & \multicolumn{2}{c}{{[}108.64,   109.48{]}}        & \multicolumn{2}{c}{{[}109.31,   110.15{]}}        \\ \hline
\end{tabular}
\caption{Validation example with jump-diffusion and Vasicek short rate dynamics
with parameters from Tables~\ref{tab:jump_parameters_validation} and~\ref{tab:jump_parameters_vasicek};
expiry time $T=5$, the insurance fee $\beta = 0.02$.}
\label{tab:sto_int}
\end{table}

To carry out Monte Carlo validation,
we proceed in two steps outlined below.
\begin{itemize}
\item Step~1: we solve the GMWB pricing problem using the ``$\epsilon$-mF'' method on
a relatively fine computational grid (Refinement Level 2 in Table~\ref{tab:grid_size}).
During this step, the optimal control $\gamma_{l, d, q}^m$ is stored for each
computational gridpoint $\x_{l, d,q}^m \in  \Omega_{\myin} \cup \Omega_{a_{\min}}
\cup \Omega_{w_{\min}} \cup \Omega_{aw_{\min}}$.

\item Step~2: we carry out Monte Carlo simulation of dynamics \eqref{eq:A_dynamics},
and \eqref{eq:dynamics}, and \eqref{eq:R_dynamics}, for $A(t)$,  $Z(t)$, and $R(t)$, respectively,
following the stored PDE-computed optimal strategies
$\{(\x_{l, d,q}^m, \gamma_{l, d, q}^m)\}$ obtained in Step~1.

Specifically, let $t_{m'} = T - \tau_m$, $m' = M-m$, $m = M-1, \ldots, 0$,
and $\hat{Z}_{m'}$, $\hat{R}_{m'}$ and  $\hat{A}_{m'}$ be simulated values.
Across each $t_{m'}$,  if necessary, linear interpolation
$\mathcal{I}\left\{\gamma^{m}\right\}(\ln(\hat{Z}_{m'}), \hat{R}_{m'}, \hat{A}_{m'})$
is applied to determine the optimal controls for simulated state values.
(No linear interpolation across time is used.)  For $t\in [t_{m'-1}, t_{m'}]$, a smaller timestep size than
$\Delta \tau$ is utilized for MC simulation. For Step~2, a total of $10^{5}$ paths and a timestep size
$\Delta \tau/20$ is used. The antithetic variate technique is also employed to reduce the variance of MC simulation.
\end{itemize}
In Table~\ref{tab:sto_int}, we present the no-arbitrage prices (in dollars) obtained by the ``$\epsilon$-mF'' method and by the above-described MC simulation.
These prices indicate indicate excellent agreement
with those obtained by MC simulation.
In addition, first-order convergence is observed for ``$\epsilon$-mF''.

\subsection{Modeling impact}
In this subsection, we investigate the (combined) impact of  jumps and stochastic
interest rate dynamics on quantities of central importance to GMWBs, namely
no-arbitrage prices and fair insurance fees, as well as on the holder's optimal withdrawal behaviors.
In this study, we typically compare the aforementioned quantities obtained from different model types:
(i) pure-diffusion (GBM) dynamics with a constant interest rate,
(i) pure-diffusion (GBM) dynamics with Vasicek short rate,
(ii) jump-diffusion dynamics with a constant interest rate,
and (iii) jump-diffusion dynamics with Vasicek short rate.
Hereinafter, these model types are respectively referred to as
``GBM-C'', ``GBM-V'',  ``JD-C'' and ``JD-V''. As an illustrative example, we only consider
the case of the Merton jump-diffusion dynamics;
using the Kou jump-diffusion dynamics yield qualitatively similar conclusions,
and hence omitted for brevity. We note that, the Merton jump parameters in Table~\eqref{tab:jump_parameters_validation}
result in $\kappa = - 0.5501$, indicating a bear stock market scenario,
which is typical in an elavated interest rate setting.

With respect to interest rates, for fair comparisons,  we establish an effective constant interest
rate which is  ``comparable'' to stochastic short rate dynamics. Hereinafter, this comparable rate is
denoted by $r_c$. Inspired by \cite{Ballotta2015}, the comparable constant interest rate $r_c$
is chosen to be the $T$-year Yield-to-Maturity (YTM)
corresponding to the Vasicek dynamics~\eqref{eq:R_dynamics}.
The comparable constant rate $r_c$ is obtained simply
by solving $e^{-r_c T} = p_b(r_0, T; T)$, where $p_b(r_0, T; T)$
given by the formula \eqref{eq:zero_coupon}.
This gives
\EQ
\label{eq:pb}
r_c = -\ln(p_b(r_0, T; T))/T, \quad
p_b(r_0, \cdot; T) \text{ is given in \eqref{eq:zero_coupon}}.
\EN
With respect to jumps, we consider an effective constant instantaneous volatility which approximates
the behavior of the Merton jump-diffusion dynamics by pure-diffusion dynamics \cite{navas_2000}.
It is interesting to include this case as conventional wisdom asserts that over long times, jump-diffusions can be
approximated by diffusions with enhanced volatility. In our experiments, the effective (enhanced)
constant instantaneous volatility, denoted by $\sigma_{c}$,
is computed by \cite{navas_2000}
\EQ
\label{eq:jump}
\sigma_{c} = \sqrt{\sigz^2 + \lambda(\nu^2 + \varsigma^2)}.
\EN
In Table~\ref{tab:models}, numerical values of parameters relatvant to
different models are given.

Regarding numerical methods for different model types, we note that the propsed SL $\epsilon$-monotone Fourier method
can be modified in a straightfoward manner to handle the GBM-V model.
Concerning the GBM-C and JD-C models, the $\epsilon$-monotone Fourier method
for jump-diffusion dynamics with a constant interes rate
proposed in our paper \cite{online} is used.

\begin{table}[htb!]
\center
\begin{tabular}{r|c|cc|c|c}
\hline
        &   & \multicolumn{2}{c|}{$r_c$}&   &  \\
Model  &  $\sigma_c$            & $T = 5$ & $T = 10$        & Merton   &  Vasicek \\ \hline
GBM-C  & 0.437 & 0.0485 & 0.0448 &  n/a &n/a\\
GBM-V  & 0.437 & \multicolumn{2}{c|}{n/a} & n/a & Table~\ref{tab:jump_parameters_vasicek}\\
JD-C   & n/a & 0.0485 & 0.0448 & Table~\ref{tab:jump_parameters_validation} &n/a\\
JD-V  &  n/a & \multicolumn{2}{c|}{n/a} & Table~\ref{tab:jump_parameters_validation} & Table~\ref{tab:jump_parameters_vasicek}\\
\hline
\end{tabular}
\caption{Parametes for different models considered;
$r_c$ and $\sigma_c$ are computed using
\eqref{eq:pb} and \eqref{eq:jump}, respectively.
\label{tab:models}}
\end{table}

In subsequent discussions, to compare no-arbitrage prices ($v$) and fair insurance fees ($\beta_f$)
across different model types, with $x \in \{v, \beta_f\}$, we denote by
$\% \Delta x (\text{Model}_1, \text{Model}_2)$ the relative change
in the quantity $x$ between $\text{Model}_1$ and $\text{Model}_2$. It is defined by
$\ds \% \Delta x (\text{Model}_1, \text{Model}_2) = \frac{|x_1 - x_2|}{x_2}$,
where $x_1$ and $x_2$ are respecitive $x$-values for $\text{Model}_1$ and $\text{Model}_2$.

\subsubsection{No-arbitrage prices and fair insurance fees}
In this experiment, we compare the no-arbitrage prices
and the fair insurance fees obtained from different model types described above
with parameters specified in Table~\ref{tab:models}
and the correlation coefficient $\rho = 0.2$ for the GBM-V and the JD-V models.
In Table~\ref{tab:table_com04}, we present selected selected results
obtained from four different models.
Here, the no-arbitrage prices
(obtained with the insurance fee $\beta = 0.02$), and
the fair insurance fees are numerically estimated as described in Subsection~\ref{ssec:fee}.
\begin{table}[htb!]
\center
\begin{tabular}{r|cc|cc}
\hline
\multirow{2}{*}{Model} & \multicolumn{2}{c}{no-arbitrage price ($v$)} & \multicolumn{2}{|c}{fair insurance fee ($\beta_f$)} \\ \cline{2-5}
                       & $T=5$         & $T=10$        & $T=5$              & $T=10$            \\ \hline
GBM-C                    & 116.1926    & 115.1230    & 0.1070           & 0.0610          \\
GBM-V                    & 116.2775    & 115.7670    & 0.1079          & 0.0647        \\
JD-C                   & 113.0806    & 111.9754    & 0.0801           & 0.0487          \\
JD-V                   & 114.1549    & 114.4837    & 0.0841           & 0.0550          \\ \hline
\end{tabular}
\caption{
No-arbitrage prices and fair insurance fees obtained from different model types;
parameters specified in Table~\ref{tab:models}; the insurance fee $\beta = 0.02$ used for
no-arbitrage prices; for GBM-V and JD-V, the correlation is $\rho = 0.2$; refinement level 2.
}
\label{tab:table_com04}
\end{table}
The numerical results in Table~\ref{tab:table_com04} suggest that
jumps and stochastic short rate have substantial combined impact on both no-arbitrage prices and fair insurance fees, with the impact being more pronounced on the latter (the fees) than on the former (prices).
Also, the fair insurance fees under the GBM-C/V models are considerably more expensive than
those obtained under JD-C/V models.
Specifically, with GBM-C being the reference model,  when $T= 5$, $\%\Delta \beta_f(\cdot,\text{GBM-C})$ ranges from
$0.8\%$ ($=\%\Delta \beta_f (\text{GBM-V}, \text{GBM-C})$) to $25.1\%$ ($=\%\Delta \beta_f (\text{JD-C}, \text{GBM-C})$), which is much large than $\%\Delta v(\cdot,\text{GBM-C})$ ranging from $0.1\%$ ($=\%\Delta v (\text{GBM-V}, \text{GBM-C})$) to $2.7\%$, which is $\%\Delta v (\text{JD-C}, \text{GBM-C})$.
Similarly,  for $T= 10$:
$\%\Delta \beta_f(\cdot,\text{GBM-C})$ ranges from  $6.0\%$ ($=\%\Delta \beta_f(\text{GBM-V}, \text{GBM-C})$) to
$20.1\%$ ($=\%\Delta \beta_f (\text{JD-C}, \text{GBM-C})$), whereas, $\%\Delta v(\cdot,\text{GBM-C})$
is only from  $0.6\%$ ($=\%\Delta v (\text{GBM-V}, \text{GBM-C})$) to $2.7\%$ ($=\%\Delta v (\text{JD-C}, \text{GBM-C})$).

We also observe that, all else being equal,
the price and the fair insurance fee obtained with a constant interest rate
(GBM-C, JD-C) are also smaller than those obtained from
the Vasicek dynamics counterpart (resp.\ GBM-V, JD-V).
For example, when $T = 10$, compare JD-C (0.0801) vs JD-V (0.0841),
and GBM-C(0.1070) vs GBM-V(0.1079).
On the other hand, application of jumps, all else being equal,
results in a lower fair insurance fee.
For example,  when $T = 10$, compare JD-C (0.0801) vs GBM-C (0.1070)
and JD-V (0.0841) vs GBM-V (0.1079)).
We also observe that, all else being equal,
the impact of jumps on the fair insurance fee (and the price)
reduces as the maturity $T$ increases, but that of stochastic interest rate
appears to be more pronounced over a longer investment horizon.
For example, regarding jumps, $\%\Delta\beta_f(\text{JD-C}, \text{GBM-C})$ is $25.1\%$
when $T = 5$ (years), but reduces to $20.1\%$ when $T = 10$ (years);
regarding interest rate, $\%\Delta\beta_f(\text{JD-C}, \text{JD-V})$ is $4.7\%$
when $T = 5$ (years), but is $11.4\%$ when $T~=~10$~(years) .

A possible explanation for the above observation is as follows.
Stochastic interest rate constitutes an additional source of risk uncaught by using a constant interest rate,
resulting in the fair insurance fee (and
the no-arbitrage price) underpriced using a constant interest rate
than using stochastic interest rate dynamics.
Furthermore, using an effective volatility ($\sigma_c$)
does not fully capture risk caused by (substantial) downward jumps,
hence resulting in the fair insurance fee underpriced.
To investigate further the combined impact of jumps
and stochastic interes rates, in the following subsection,
we study the holder's optimal withdrawal behaviors.

\subsubsection{Optimal withdrawals}
In this study, we use the fair insurance fees for the GBM-C, GBM-V, JD-C and JD-V models, respectively
denoted by $\beta_f^{gc}$, $\beta_f^{gV}$, $\beta_f^c$ and $\beta_f^{V}$.
We use $T = 10$ and $\rho = 0.2$. As reported in Table~\ref{tab:table_com04},
$\beta_f^{gc} = 0.0610$, $\beta_f^{gV} = 0.0647$, $\beta_f^c = 0.0487$ and $\beta_f^{V} = 0.0550$.
In Figure~\ref{fig:result05}, we present plots of optimal withdrawals
for (calendar) time $t = 5$ (years) obtained using different models:
the GBM-C in Figure~\ref{fig:result05}(a), the JD-C model in Figure~\ref{fig:result05}(b), the GBM-V in Figure~\ref{fig:result05}(c), and the JD-V model in Figure~\ref{fig:result05}(d).
For the GBM-V and JD-V models, the control plots correspond to the spot rate $R(t= 5) = r_c = 0.0448$.

\begin{figure}[htb!]
\begin{center}
$\begin{array}{c@{\hspace{0.08in}}c}
\includegraphics[width=0.5\textwidth]{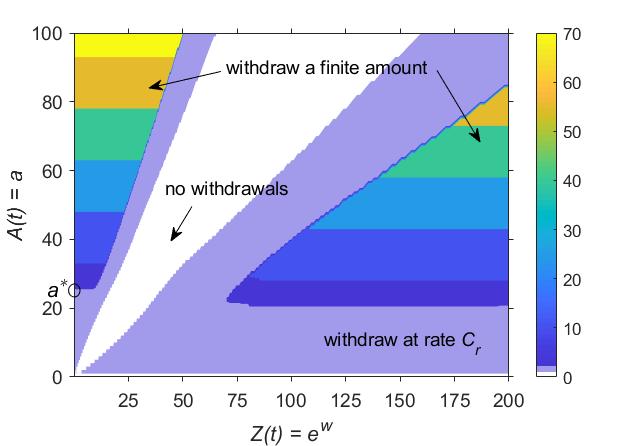}
& \includegraphics[width=0.5\textwidth]{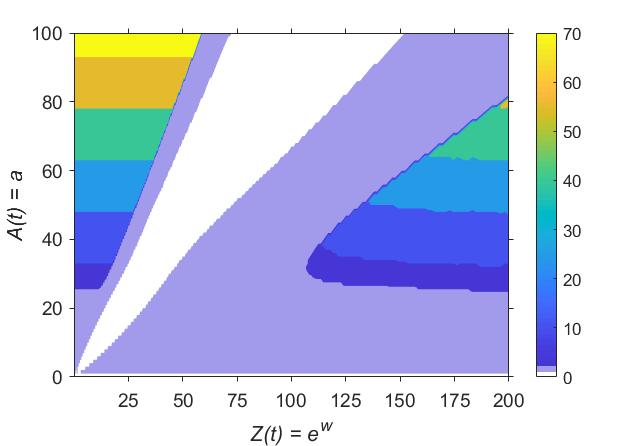}
\\
\mbox{(a) GBM-C,  $t= 5$, $\sigma_c = 0.4373$, $r_c = 0.0448$}
& \mbox{(b) JD-C, $t= 5$, $r_c = 0.0448$}
\\
\includegraphics[width=0.5\textwidth]{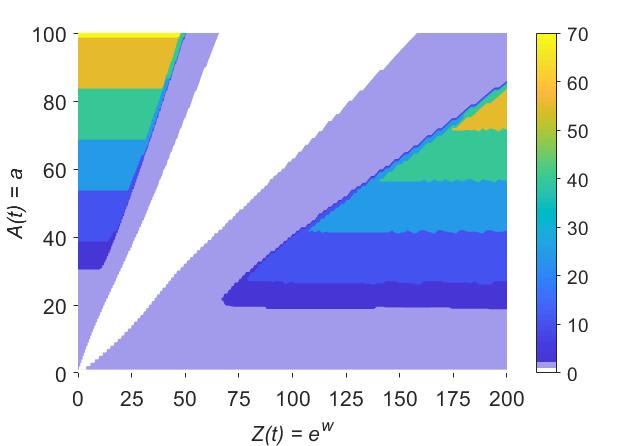}
& \includegraphics[width=0.5\textwidth]{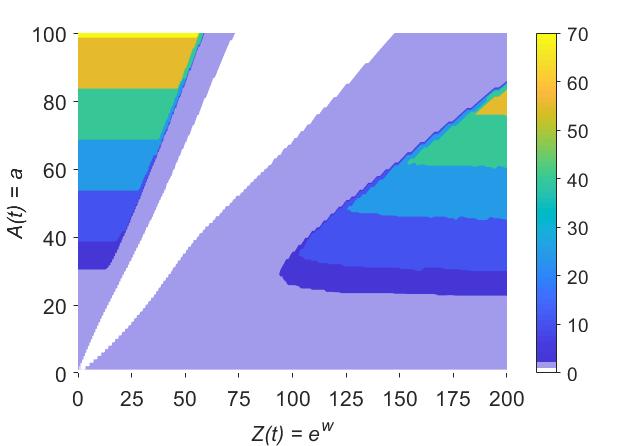}
\\
\mbox{(c) GBM-V, $t= 5$, $R(t) = r_c=  0.0448$}
& \mbox{(d) JD-V, $t= 5$, $R(t) = r_c = 0.0448$}
\end{array}$
\end{center}
\caption{The holder's optimal withdrawals at (calendar) time $t= 5$ (years);
parameters specified in Table~\ref{tab:models};
$T = 10$, $\rho = 0.2$;
fair insurance fee $\beta_f^{gc} = 0.0610$, $\beta_f^{gV} = 0.0647$, $\beta_f^c = 0.0487$, $\beta_f^{V} = 0.0550$;
refinement level 2.}
\label{fig:result05}
\end{figure}

\begin{figure}[htb!]
\begin{center}
$\begin{array}{c@{\hspace{0.08in}}c}
\includegraphics[width=0.5\textwidth]{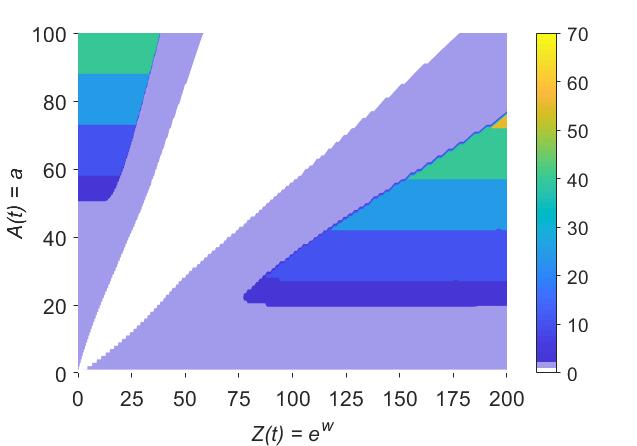}
& \includegraphics[width=0.5\textwidth]{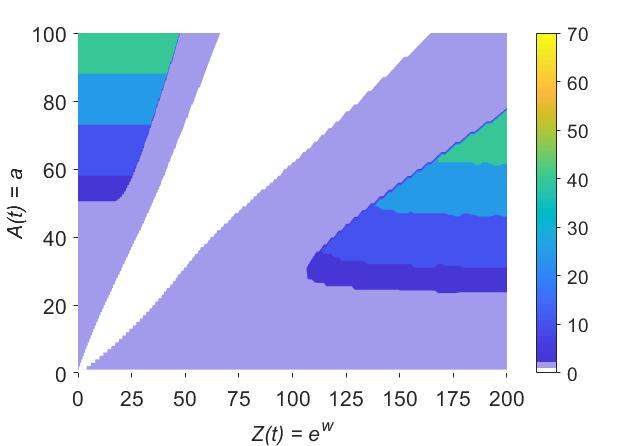}
\\
\mbox{(a) GBM-V,  $t= 5$, $R(t) = 0.03$}
& \mbox{(b) JD-V, $t= 5$, $R(t) = 0.03$}
\\
\includegraphics[width=0.5\textwidth]{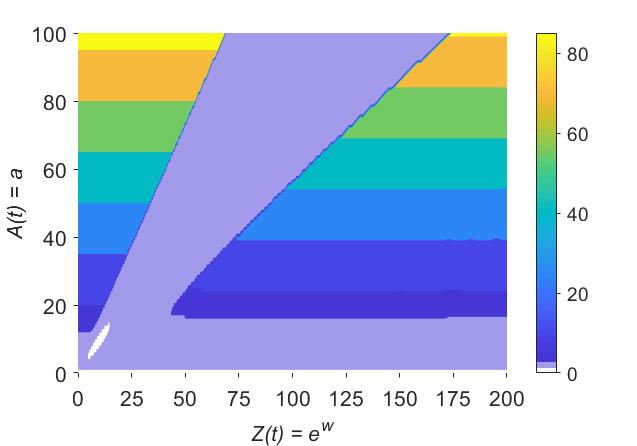}
& \includegraphics[width=0.5\textwidth]{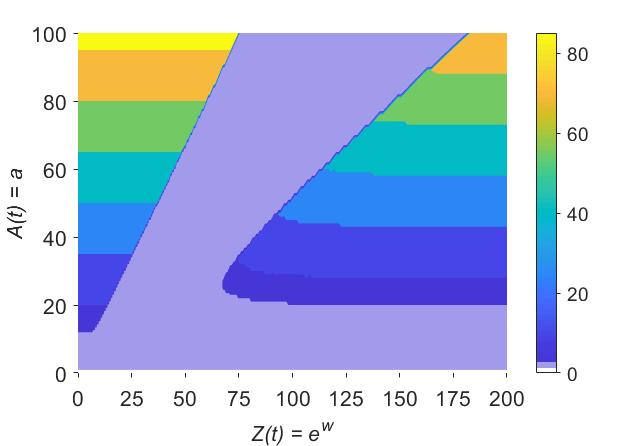}
\\
\mbox{(c) GBM-V,  $t= 5$, $R(t) = 0.1$}
& \mbox{(d) JD-V, $t= 5$, $R(t) = 0.1$}
\\
\includegraphics[width=0.5\textwidth]{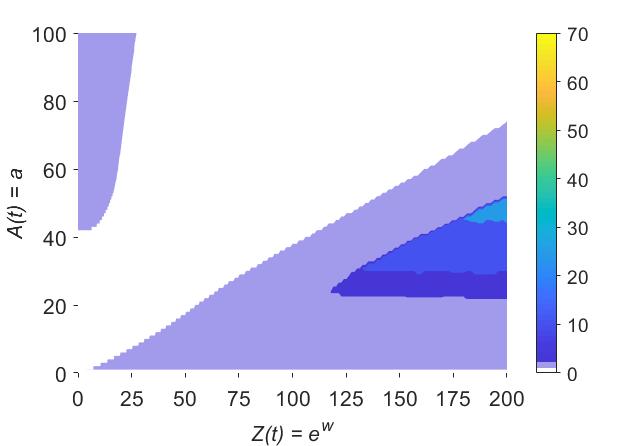}
&
\includegraphics[width=0.5\textwidth]{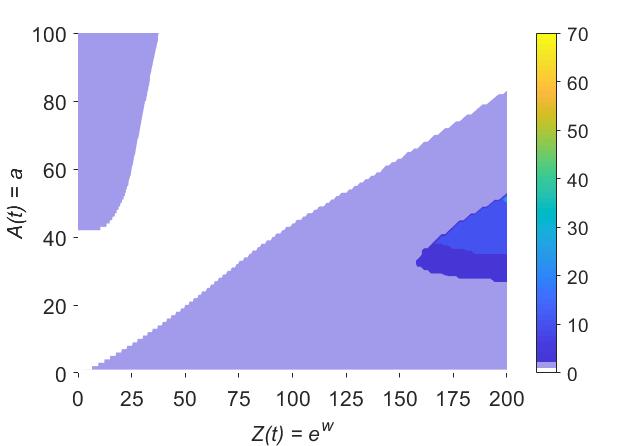}
\\
\mbox{(e) GBM-V,  $t= 5$, $R(t) = -0.0125$}
&
\mbox{(f) JD-V, $t= 5$, $R(t) =  -0.0125$}
\end{array}$
\end{center}
\caption{The holder's optimal withdrawals at $t= 5$ (years) for different spot rates;
parameters are from Table~\ref{tab:jump_parameters_validation}[Merton]
and Table~\ref{tab:jump_parameters_vasicek};
$T = 10$, correlation coefficient $\rho = 0.2$,
effective volatility $\sigma_c = 0.4373$, fair insurance fee
$\beta_f^{gV} = 0.0647$, $\beta_f^{V} = 0.0550$;
refinement level 2.
}
\label{fig:result09}
\end{figure}

From Figure~\ref{fig:result05}, we observe several key qualitative similarities across different models.
Specifically, in the lower-right region, where $A(t) \ll z_0$ and $Z(t)\gg A(t)$,
all optimal controls suggest the holder should withdraw continuously at rate $C_r$;
however, withdrawing a finite amount becomes optimal when $A(t)$ becomes sufficiently
large (upper-right region). Also, in the lower-left region, when both $A(t)$ and $Z(t)$ are small,
optimal controls suggest to either withdrawal nothing or to withdraw continuously at rate $C_r$;
however, in the upper-left region of Figure~\ref{fig:result05}, where $A(t)\gg Z(t)$,
optimal controls suggest withdraw a finite amount.

Nonetheless, significant quantitative differences are also observed,
most notably in  the upper-right and in the lower-left regions.
For example, consider the upper-right region
in Figure~\ref{fig:result05}(a)-(d). At $(Z(t), A(t)) = (200, 80)$, our numerical results in
Figure~\ref{fig:result05}(b), indicate that, when the JD-C model is used,
it is optimal to withdraw continuously at rate $C_r = 1/T = 0.1$;
however,  using other model, as shown in Figure~\ref{fig:result05}(a), (c) and (d),
it is suggested that withdrawing a finite amount (about \$60) is optimal.

In Figure~\ref{fig:result09}, we present control plots for at $t=5$ (years)
when $R(t) \in \{0.03, 0.1\} \neq r_c$, and $R(t) = -0.0125<0$ obtained using the GBM-V and JD-V models.
Comparing Figure~\ref{fig:result09}(a), (c) and (e) with Figure~\ref{fig:result05}(c),
as well as comparing Figure~\ref{fig:result09}(b), (d) and (f) with Figure~\ref{fig:result05}(d)),
suggests that the optimal withdrawal behaviours depend considerably on spot rates,
and they are significantly different from those obtained using a comparable rate $r_c$,
with a more conservative withdraw behaviours, especially in withdrawing a finite amount,
when the spot interst rate is low.

\begin{figure}
\begin{center}
\vspace*{-0.8cm}
\includegraphics[width = 0.45\textwidth]{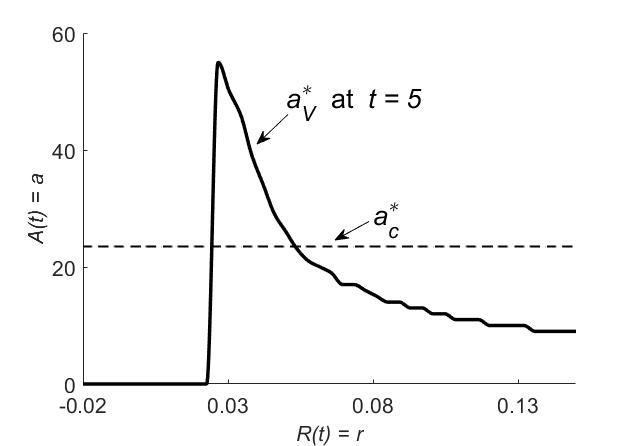}
\captionof{figure}{A plot of $a^{\ast}_V$ (for JD-V)
and $a^{\ast}_c$ (for GBM-C) against spot rate $R(t)$ at (calendar) time $t = 5$ (years);
parameters are similar to those used for Figure~\ref{fig:result05};
\label{fig:result06}}
\vspace*{-0.8cm}
\end{center}
\end{figure}


We now turn our attention to the lower-left region of the control plots in Figure~\ref{fig:result05}
and Figure~\ref{fig:result09}, where $A(t)$ dominates  $Z(t)$.
In particular, with $Z(t)$ being zero, we study the value of $a$ across which the optimal withdrawal behaviours change from withdrawing continuously at rate $C_r$ to withdrawing a finite amount.
For brevity, we only discuss the GBM-C and JD-V model.
For the GBM-C model, we denote by $a^\ast_c$ this special $a$-value, and it is given by  $a^\ast_c = -\frac{C_r}{r_c} \ln(1-\mu)$, as shown in \cite{Dai}. For the JD-V model, we denote by $a^\ast_V$ the aforementioned special value of $a$ (this is also the same $a$-value for the GBM-V model).
A closed-form expression for $a^\ast_V$ is not known to exist, and therefore,
we estimate it using numerical results.

In Figure~\ref{fig:result06}, we plot $a^\ast_c$  and $a^\ast_V$ against different spot rate $R(t)$ at $t = 5$.
We note that, when $r < 0$ and $z = e^{w} \to  0$, Figure~\ref{fig:result06} suggests that never optimal
to withdraw a finite amount (also see Figure~\ref{fig:result05}(c)).
It is observed from Figure~\ref{fig:result06} that when $R(t) \ll r_c$, $a^\ast_V$ is significantly larger than $a^\ast_c$; however,  when $R(t) \gg r_c$, $a^\ast_V$ is considerably smaller than $a^\ast_c$.
These suggest that, when the balance of sub-account balance is zero,
the holder should be much more cautious with finite amount withdrawals from the guarantee account
in a low interest rate environment than s/he is in a constant interest rate;
however, the holder should be much more aggressive in a high interest rate environment.

To summarize, our numerical results suggest a simultaneous application of
jumps and stochastic interest rate result in considerably cheaper fair fees than those obtained
under a comparable pure-diffusion model. In addition, under this realistic modeling setting,
the holder's optimal withdrawal behaviour appears to be much more conservative (resp.\ aggressive) in withdrawing a finite amount when the balance of the sub-account is negligible (resp.\ considerable) than in the optimal behaviour under a pure-diffusion model would dictate. This is possibly because of combined risk due to (i) possible downsize jumps, and (ii) stochastic interest rate, which drives lower fair insurance fees for GMWBs.
We plan to investigate these observations further in a future work.

\section{Conclusion}
\label{section:cc}

In a continuous withdrawal scenario, using an impulse control framework, the GMWB pricing problem 
under a jump-diffusion dynamics with stochastic short rate is formulated as HJB-QVI of three spatial dimensions. The viscosity solution to this HJB-QVI is shown to satisfy a strong comparison result. Utilizing a semi-Lagrangian discretization, we develop an $\epsilon$-monotone Fourier method to solve the HJB-QVI. 
We rigorously prove the convergence of the numerical solutions to the viscosity solution of the associated HJB-QVI. 
Numerical experiments demonstrate an excellent agreement with reference values obtained by the Monte Carlo simulation. 
Extensive analysis of numerical results indicate
a significant (combined) impact of jumps and stochastic interest rate dynamics on the fair insurance fees and 
on the optimal withdrawal behaviors of policy holders. For future work, we plan to investigate further 
the impact of realistic modeling with various withdrawal settings and complex contract features.


\section*{Acknowledgement}
The authors would like to thank George Labahn and Peter Forsyth of the University of Waterloo for very useful comments on earlier drafts of this paper.

{\fontsize{9.65}{10.3}\rm

\begin{appendices}
\section{Truncation error of Fourier series}
\label{app:truncation}
As $\alpha \to \infty$, there is no loss of information in
the discrete convolution  \eqref{eq:gtilde_truncated}.
However, for any finite $\alpha$, there is an error due to the use of
a truncated Fourier series. Using similar arguments in \cite{ForsythLabahn2017},
we have
\EQA
\label{eq:err_series}
\left| \tilde{g}_{n-l, k-d}(\alpha)  - \tilde{g}_{n-l, k-d}(\infty) \right|
&\leq&
\frac{2}{P^{\dagger}} \frac{1}{Q^{\dagger}} \mysum^{z \in \Z}_{s \in [\alpha N^{\dagger}/2, \infty)} \left( \frac{\sin^2 \pi \eta_s \Delta w}{(\pi \eta_s \Delta w)^2} \right) \left(\frac{\sin^2 \pi \xi_z \Delta r}{(\pi \xi_z \Delta r)^2} \right) \left| G(\eta_s, \xi_z, \Delta \tau) \right|
\nonumber
\\
&+&
\frac{1}{P^{\dagger}} \frac{2}{Q^{\dagger}} \mysum^{z \in [\alpha K^{\dagger}/2, \infty)}_{s \in \Z} \left( \frac{\sin^2 \pi \eta_s \Delta w}{(\pi \eta_s \Delta w)^2} \right) \left(\frac{\sin^2 \pi \xi_z \Delta r}{(\pi \xi_z \Delta r)^2} \right) \left| G(\eta_s, \xi_z, \Delta \tau) \right|.
\ENA
Using the closed-form expression \eqref{eq:small_g}, and noting that $\text{Re}\left(\overline{B} (\eta)\right) \leq 1$, $\left| \rho \right| < 1$, we then have
\EQA
\label{eq:repsi}
\text{Re}(\Psi(\eta, \xi)) &=& - \frac{ \sigz^2}{2} (2 \pi \eta)^2  - \rho \sigz \sigr (2 \pi \eta) (2 \pi \xi) - \frac{\sigr^2}{2} (2 \pi \xi)^2  - \lambda + \lambda \text{Re}\left(\overline{B} (\eta)\right)
\nonumber
\\
&\leq &  - \left( 1- \left| \rho \right| \right)  \frac{\sigz^2}{2} (2 \pi \eta)^2 -  \left( 1- \left| \rho \right| \right) \frac{\sigr^2}{2} (2 \pi \xi)^2.
\ENA
Thus, from \eqref{eq:repsi}, we have
\EQA
\label{eq:absbigG}
\left| G(\eta, \xi, \Delta \tau) \right| = \left| \exp \left( \Psi(\eta, \xi)\Delta \tau \right) \right|
\leq  \exp \left( - \left( 1- \left| \rho \right| \right)  \frac{\sigz^2}{2} (2 \pi \eta)^2 \Delta \tau \right) \exp \left(  -  \left( 1- \left| \rho \right| \right) \frac{\sigr^2}{2} (2 \pi \xi)^2 \Delta \tau \right).
\ENA
Let $C_6 = 2 \left( 1- \left| \rho \right| \right) \sigz^2 \pi^2 \Delta \tau / (P^{\dagger})^2$ and
$C'_6 = 2 \left( 1- \left| \rho \right| \right) \sigr^2 \pi^2 \Delta \tau / (Q^{\dagger})^2$.
Taking \eqref{eq:absbigG} into \eqref{eq:err_series}, we can bound these infinite sums as follows
\EQA
&&\left| \tilde{g}_{n-l, k-d}(\alpha)  - \tilde{g}_{n-l, k-d}(\infty) \right|
\nonumber
\\
&\leq &
\bigg( \frac{2}{P^{\dagger}} \frac{4}{\pi^2 \alpha^2} \sum_{s = \alpha N^{\dagger}/2}^{\infty} e^{-C_6 s^2} \bigg) \bigg( \frac{1}{Q^{\dagger}} \sum_{z \in \Z} \left(\frac{\sin^2 \pi \xi_z \Delta r}{(\pi \xi_z \Delta r)^2} \right)  e^{ -  C'_6 z^2 } \bigg)
\nonumber
\\
&& \qquad
+ \bigg( \frac{2}{Q^{\dagger}} \frac{4}{\pi^2 \alpha^2} \sum_{z = \alpha N^{\dagger}/2}^{\infty} e^{-C'_6 z^2} \bigg) \bigg( \frac{1}{P^{\dagger}} \sum_{s \in \Z} \left(\frac{\sin^2 \pi \eta_s \Delta w}{(\pi \eta_s \Delta w)^2} \right)  e^{ -  C_6 s^2 } \bigg)
\nonumber
\\
&\leq &
\frac{8 (K^{\dagger})^2 (1+e^{-C'_6})}{P^{\dagger} Q^{\dagger} \pi^4 \alpha^2 (1 - e^{-C'_6})} \frac{\exp \left(- C_6 N^{\dagger} \alpha^2/4 \right)}{1-e^{-C_6 N^{\dagger} \alpha}}
+
\frac{8 (N^{\dagger})^2 (1+e^{-C_6})}{P^{\dagger} Q^{\dagger} \pi^4 \alpha^2 (1 - e^{-C_6})} \frac{\exp \left(- C'_6 K^{\dagger} \alpha^2/4 \right)}{1-e^{-C'_6 K^{\dagger} \alpha}},
\nonumber
\ENA
which yields (considering fixed $P^{\dagger}$ and $Q^{\dagger}$ here)
\EQAS
\left| \tilde{g}_{n-l, k-d}(\alpha)  - \tilde{g}_{n-l, k-d}(\infty) \right| ~\simeq~
\mathcal{O}\left( e^{-1/h} / h^2 \right).
\ENAS

\section{A proof of Proposition~\ref{proposition:sum_g}}
\label{app:sum_g}

\begin{proof}[Proof of Proposition~\ref{proposition:sum_g}]
Letting $p = n-l$ and $q = k -d$, we have
\EQA
 && \Delta w \Delta r  \mysum^{k \in \KD}_{l \in \ND} \tilde{g}_{n-l, k-d}
 ~\overset{\text{(i)}}{=}~ \frac{P^{\dagger}}{N^{\dagger}} \frac{Q^{\dagger}}{K^{\dagger}} \mysum^{q \in \KD}_{p \in \ND} \tilde{g}_{p, q}
\nonumber
\\
&&
\qquad\qquad
\overset{\text{(ii)}}{=}~ \frac{P^{\dagger}}{N^{\dagger}} \frac{Q^{\dagger}}{K^{\dagger}}
\mysum^{q \in \KD}_{p \in \ND}
\frac{1}{P^{\dagger}} \frac{1}{Q^{\dagger}}
\mysum^{z \in \Ka}_{s \in \Na}
e^{2\pi i \eta_s p \Delta w}
e^{2\pi i \xi_z q \Delta r}
~
\text{tg}(s,z)
G(\eta_s, \xi_z, \Delta \tau)
\nonumber
\\
&&
\qquad\qquad
=~ \frac{1}{N^{\dagger}} \frac{1}{K^{\dagger}} \mysum^{z \in \Ka}_{s \in \Na} \text{tg}(s,z) G(\eta_s, \xi_z, \Delta \tau)  \sum_{p \in \ND} \exp{\left(\frac{2 \pi i s p}{N^{\dagger}} \right)} \sum_{q \in \KD} \exp{\left(\frac{2 \pi i z q}{K^{\dagger}} \right)}
\nonumber
\\
&&
\qquad\qquad
\overset{\text{(iii)}}{=}~ G(0, 0, \Delta \tau) ~\overset{\text{(iv)}}{=} ~ 1.
\ENA
Here, in (i), we use the periodicity of $\tilde{g}_{n-l, k-d}$, i.e.\ the sequence $\{\tilde{g}_{-N^{\dagger}/2, k}(\alpha),\ldots, \tilde{g}_{N^{\dagger}/2-1, k}(\alpha)\}$ for a fixed $k \in \KD$ is $N^{\dagger}$-periodic, and similarly, the sequence $\{\tilde{g}_{n, -K^{\dagger}/2}(\alpha),\ldots, \tilde{g}_{n, K^{\dagger}/2-1}(\alpha)\}$ for a fixed $n \in \ND$ is $K^{\dagger}$-periodic; in (ii), we use the definition of \eqref{eq:gtilde_truncated}, noting the term $\text{tg}(s,z)$ is given in \eqref{eq:trigofun}; in (iii),we apply properties of roots of unity; in (iv), we use the closed-form expression \eqref{eq:small_g}.
\end{proof}

\section{${\boldsymbol{\ell}}$-stability in ${\boldsymbol{\Omega_{\myin} \cup \Omega_{a_{\min}}}}$}
\label{app:induction}
We now show the bounds  \eqref{eq:vminusmax}-\eqref{eq:vminusmin} for $\Omega_{\myin} \cup \Omega_{a_{\min}}$.
We note that numerical solutions at  nodes in $\Omega \setminus (\Omega_{\myin} \cup \Omega_{a_{\min}})$ satisfy the bounds \eqref{eq:vminusmax}-\eqref{eq:vminusmin} at the same $j \in \J$ and $m = 0, \ldots, M$, that is
\EQ
\label{eq:boundmin}
  \max_{n \in \Nc ~\text{or}~ k \in \Kc} \left\{v_{n,k,j}^{m}\right\}~\text{satisfies  \eqref{eq:vminusmax}},
  \quad \text{and}\quad
  \min_{n \in \Nc ~\text{or}~ k \in \Kc} \left\{v_{n,k,j}^{m}\right\}~\text{satisfies  \eqref{eq:vminusmin}}.
\EN
\noindent Base case: when $m = 0$, \eqref{eq:vminusmax}-\eqref{eq:vminusmin} hold
for all $j \in \J$,
which follows from the initial condition \eqref{eq:terminal} for $n \in \N$

\noindent Induction hypothesis: we assume that \eqref{eq:vminusmax}-\eqref{eq:vminusmin} hold
for $m = \hat{m}$, where $\hat{m}\le M-1$, and $j \in \J$.

\noindent Induction: we show that \eqref{eq:vminusmax}-\eqref{eq:vminusmin} also hold
for $m = \hat{m} + 1$ and  $j \in \J$. This is done in two  steps. In Step 1, we show,
for $j \in \J$,
\EQA
\left[v_j^{\hat{m}+}\right]_{\max} &\le&
  e^{2  \hat{m} \epsilon \frac{\Delta\tau}{T}} e^{\R \hat{m} \Delta \tau}
\left(
\left\| v^{0} \right\|_{\infty} + a_{j}
\right)
\label{eq:vplusmax}
  \\
- 2\hat{m}\epsilon\frac{\Delta\tau}{T}e^{2\hat{m}\epsilon  \frac{\Delta\tau}{T}} e^{\R \hat{m} \Delta \tau}
\left(
\left\| v^{0} \right\|_{\infty} +a_{j}\right)
  &\le&
\left[v_j^{\hat{m}+}\right]_{\min},
\label{eq:vplusmin}
\ENA
where $\big[v_j^{\hat{m}+}\big]_{\max} = \max_{n,k}\big\{v_{n,k,j}^{\hat{m}+}\big\}$ and $\big[v_j^{\hat{m}+}\big]_{\min} = \min_{n,k}\big\{v_{n,k,j}^{\hat{m}+}\big\}$.
In Step~2, we bound the timestepping result \eqref{eq:scheme} at $m = \hat{m} + 1$ using
\eqref{eq:vplusmax}-\eqref{eq:vplusmin}.

\noindent Step 1 - Bound for $v_{n,k,j}^{\hat{m}+}$:
Since  $v_{n,k,j}^{\hat{m}+}  = \max\left(\vl_{n,k,j}^{\hat{m}+}, \vn_{n,k,j}^{\hat{m}+}\right)$,
using \eqref{eq:scheme*}, we have
\EQA
\label{eq:local_m_gamma}
v_{n,k,j}^{\hat{m}+}
&=&  \sup_{\gamma_{n,k,j}^{\hat{m}} \in[0, a_j]}
\left[
   \mathcal{I}\left\{v^{\hat{m}}\right\}
   \left(\max\left(e^{w_n} - \gamma_{n, j}^{\hat{m}}, e^{w^{\dagger}_{\min}}\right), r_k,
    a_j - \gamma_{n,k,j}^{\hat{m}}   \right)
    + f(\gamma_{n,k,j}^{\hat{m}})
    \right].
\ENA
As noted in Remark~\ref{rm:sup_exist}, for the case $c>0$ as considered here,
the supremum of \eqref{eq:local_m_gamma} is achieved by
an optimal control $\gamma^* \in [0, a_j]$. That is, \eqref{eq:local_m_gamma} becomes
\EQA
\label{eq:local_m_opt}
v_{n,k,j}^{\hat{m}+}
= \mathcal{I}\left\{v^{\hat{m}}\right\}
   \left(\max\left(e^{w_n} - \gamma^*, e^{w^{\dagger}_{\min}}\right), r_k,
    a_j - \gamma^* \right)
    + f(\gamma^*), \quad \gamma^* \in [0, a_j].
\ENA
We assume that
$\max\left(e^{w_n} - \gamma^*, e^{w^{\dagger}_{\min}}\right)
\in \left[e^{w_{n^\prime}}, e^{w_{n^\prime+1}}\right]
$
and
$(a_j - \gamma^*) \in
[a_{j^\prime}, a_{j^\prime+1}]$,
and nodes that are used for linear interpolation are
$({\bf{x}}_{n^\prime,k,j^\prime}^{\hat{m}}, \ldots, {\bf{x}}_{n^\prime+1,k, j^\prime+1}^{\hat{m}})$.
We note that these node could be outside $\Omega_{\myin} \cup \Omega_{a_{\min}}$,
in $\Omega_{w_{\min}} \cup \Omega_{wa_{\min}}$.
However, by \eqref{eq:boundmin}, the numerical solutions at these nodes satisfy
the same bounds \eqref{eq:vminusmax}-\eqref{eq:vminusmin}.
Computing $v_{n,k,j}^{\hat{m}+}$ using linear interpolation results in
\EQA
\label{eq:v_interp}
v_{n,k,j}^{\hat{m}+}
= x_a \left(x_w~v_{n', k, j'}^{\hat{m}} + (1-x_w)~v_{n'+1, k, j'}^{\hat{m}}\right)
+ (1-x_a) \left(x_w~v_{n', k, j'+1}^{\hat{m}} + (1-x_w)~v_{n'+1, k, j'+1}^{\hat{m}}\right),
\ENA
where $0\le x_a\le 1$ and $0\le x_w\le 1$ are interpolation weights. In particular,
\EQA
\label{eq:xa_interp}
x_a = \frac{a_{j'+1} - (a_j - \gamma^*)}{a_{j'+1} - a_{j'}}.
\ENA
Using \eqref{eq:boundmin} and the induction hypothesis for \eqref{eq:vminusmax}
gives abound for nodal values used in \eqref{eq:v_interp}
\EQA
\label{eq:bound}
\left\{v_{n^\prime, k, j^\prime}^{\hat{m}}, v_{n^\prime+1, k, j^\prime}^{\hat{m}}\right\}
&\le &
e^{2   \hat{m} \epsilon\frac{\Delta\tau}{T}} e^{\R \hat{m} \Delta \tau} (\|v^{0}\|_{\infty} + a_{j^\prime}),
\nonumber
\\
\left\{v_{n^\prime, k, j^\prime+1}^{\hat{m}},
v_{n^\prime+1, k, j^\prime+1}^{\hat{m}}\right\}
&\le &
e^{2   \hat{m} \epsilon\frac{\Delta\tau}{T}} e^{\R \hat{m} \Delta \tau} (\|v^{0}\|_{\infty} + a_{j^\prime+1}).
\ENA
Taking into account the non-negative weights in linear interpolation, particularly \eqref{eq:xa_interp},
and upper bounds in \eqref{eq:bound}, the interpolated result $\mathcal{I}\left\{v^{\hat{m}}\right\}\left(\cdot \right)$ in \eqref{eq:local_m_opt} is bounded by
\EQA
\label{eq:Ibound}
\mathcal{I}\left\{v^{\hat{m}}\right\}
   \left(\max\left(e^{w_n} - \gamma^*, e^{w^{\dagger}_{\min}}\right), r_k,
    a_j - \gamma^* \right)
\le
e^{2 \hat{m} \epsilon \frac{\Delta\tau}{T}} e^{\R \hat{m} \Delta \tau} (\|v^{0}\|_{\infty} + (a_j - \gamma^*)).
\ENA
Using \eqref{eq:Ibound} and $f(\gamma^*)~\le~\gamma^*$ (by definition in \eqref{eq:f_gamma_k_dis}),
 \eqref{eq:local_m_opt} becomes
\EQAS
v_{n,k,j}^{\hat{m}+} ~\le~ e^{2 \hat{m} \epsilon \frac{\Delta\tau}{T}} e^{\R \hat{m} \Delta \tau} \left(\|v^{0}\|_{\infty} + a_j - \gamma^*
\right) + \gamma^*
~\le ~
e^{2 \hat{m} \epsilon \frac{\Delta\tau}{T}} e^{\R \hat{m} \Delta \tau} \left(\|v^{0}\|_{\infty} + a_j\right),
\ENAS
which proves \eqref{eq:vplusmax} at $m = \hat{m}$.

For subsequent use,  we note, since
$v_{n,k,j}^{\hat{m}+}  = \max\left(\vl_{n,k,j}^{\hat{m}+}, \vn_{n,k,j}^{\hat{m}+}\right)$,
\eqref{eq:vplusmax} results in
\EQA
\label{eq:upper_bound_01}
\left\{\vl_{n,k,j}^{\hat{m}+}, \vn_{n,k,j}^{\hat{m}+}\right\}
~\leq ~
v_{n,k,j}^{\hat{m}+}
&\le &
e^{2 \hat{m} \epsilon \frac{\Delta\tau}{T}} e^{\R \hat{m} \Delta \tau} \left(\|v^{0}\|_{\infty} +  a_j\right).
\ENA
Next, we derive a lower bound for $\vl_{n,k,j}^{\hat{m}+}$ and $\vn_{n,k,j}^{\hat{m}+}$.
By the induction hypothesis for \eqref{eq:vminusmin}, we have
$v_{n,k,j}^{\hat{m}}\ge -2m\epsilon\frac{\Delta\tau}{T}e^{2 \hat{m} \epsilon\frac{\Delta\tau}{T}}
e^{\R \hat{m} \Delta \tau}
\left(
\left\| v^{0} \right\|_{\infty} + a_j\right)$.
Comparing $\vl_{n,k,j}^{\hat{m}+}$ given by  the supremum in \eqref{eq:scheme*} with
$v_{n,k,j}^{\hat{m}}$, which is the candidate for the supremum evaluated at  $\gamma_{n,k,j}^{\hat{m}} = 0$,
yields
\EQA
\label{eq:lower_bound_0l}
v_{n,k,j}^{\hat{m}} ~ \ge~ \vl_{n,k,j}^{\hat{m}+} ~\ge~  -2\hat{m}\epsilon\frac{\Delta\tau}{T}e^{2 \hat{m} \epsilon \frac{\Delta\tau}{T}} e^{\R \hat{m} \Delta \tau}
\left(
\left\| v^{0} \right\|_{\infty} + a_j\right),
\ENA
which proves \eqref{eq:vplusmin} at $m = \hat{m}$.

For  $\vn_{n,k,j}^{\hat{m}+}$ in \eqref{eq:scheme*}, we
consider optimal $\gamma = \gamma^*$, where $\gamma^* \in (C_r\Delta \tau, a_j]$.
Using the induction hypothesis and non-negative weights of linear interpolation,
noting $\gamma^* \ge 0$ and assuming $f(\gamma^*)  \ge 0$, gives
\EQA
\label{eq:lower_bound_02}
\vn_{n,k,j}^{\hat{m}+}
&\ge &
-2\hat{m}\epsilon\frac{\Delta\tau}{T}e^{2\hat{m} \epsilon\frac{\Delta\tau}{T}}  e^{\R \hat{m} \Delta \tau}
\left(
\left\| v^{0} \right\|_{\infty} + (a_j - \gamma^*)\right)  + f(\gamma^*)
\nonumber
\\
&\ge&
-2\hat{m}\epsilon\frac{\Delta\tau}{T}e^{2\hat{m} \epsilon\frac{\Delta\tau}{T}} e^{\R \hat{m} \Delta \tau} \left(
\left\| v^{0} \right\|_{\infty} + a_j\right).
\ENA
From \eqref{eq:upper_bound_01}-\eqref{eq:lower_bound_0l} and \eqref{eq:lower_bound_02}, noting $\epsilon\le 1/2$,
we have
\EQA
\label{eq:abs_bound}
\left\{|\vl_{n,k,j}^{\hat{m}+}|, |\vn_{n,k,j}^{\hat{m}+}|\right\}
\leq
e^{2 \hat{m} \epsilon \frac{\Delta\tau}{T}} e^{\R \hat{m} \Delta \tau} \left(\|v^{0}\|_{\infty} +  a_j\right).
\ENA
\noindent Step 2 - Bound for $v_{n,k,j}^{\hat{m}+1}$: We will show that \eqref{eq:vminusmax}-\eqref{eq:vminusmin}
hold at $m = \hat{m}+1$.
For all $n \in \N$, $k \in K$, $j \in J$,
using \eqref{eq:semi_lag_results} and \eqref{eq:V_nG},
we have
\EQA
\label{eq:vkk}
\bvlsl_{n,k,j}^{{\hat{m}+1}}
&=&
    \Delta w \Delta r
\mysum^{d \in \KD}_{l \in \ND}
    \tilde{g}_{n-l, k-d}
    \bvlsl_{l,d,j}^{\hat{m}+}
\nonumber
\\
&=&
\Delta w \Delta r
\mysum^{d \in \KD}_{l \in \ND}
    \big(\max\left(\tilde{g}_{n-l, k-d}, 0\right)
    +
    \min\left(\tilde{g}_{n-l, k-d}, 0\right)
    \big)
    \bvlsl_{l,d,j}^{{\hat{m}+}}.
\ENA
Note that $\bvlsl_{l,d,j}^{{\hat{m}+}}$ is computed by \eqref{eq:semi_lag_results}, where $\bw_l$ and $\br_d$ have no dependence on $a_j$.  From \eqref{eq:vkk}, using the property of linear interpolation
and the upper bound \eqref{eq:abs_bound}, we have
\EQA
|\bvlsl_{n,k,j}^{{\hat{m}+1}} |
&\leq &
\frac{\Delta w \Delta r}{|1 + \Delta \tau r_d|}
\mysum^{d \in \KD}_{l \in \ND}
    \big(\max\left(\tilde{g}_{n-l, k-d}, 0\right)
    +
    \left| \min\left(\tilde{g}_{n-l, k-d}, 0\right) \right|
    \big)
\big| \mathcal{I}\big\{\vl^{\hat{m}+}\big\}(\bw_l, \br_d, a_j) \big|
\nonumber
\\
&\overset{\text{(i)}}{\le}&
(1+2\epsilon \frac{\Delta \tau}{T})
e^{2\epsilon \hat{m} \frac{\Delta\tau}{T}}
\left(
1
+
\Delta \tau \R
\right) e^{\R \hat{m} \Delta \tau }
\left(\|v^{0}\|_{\infty} + a_j\right)
\nonumber
\\
&\le &
e^{2\epsilon (\hat{m}+1) \frac{\Delta\tau}{T}}
e^{\R (\hat{m}+1)  \Delta \tau }
\left(\|v^0\|_{\infty} + (1+\mu)a_j+c \right),
\label{eq:vl}
\ENA
where in (i), we use \eqref{eq:sum_g} and \eqref{eq:r_k_con}.
Similarly, for $n \in \N$, $k \in \K$, $j \in \J$,
we also have
\EQ
\label{eq:vn}
|\bvnsl_{n,k,j}^{\hat{m}+1}| \le e^{2 (\hat{m} +1) \epsilon \frac{\Delta\tau}{T}} e^{\R (\hat{m}+1)  \Delta \tau } (\|v^{0}\|_{\infty} +  a_j).
\EN
Therefore, from \eqref{eq:vl}-\eqref{eq:vn}, we conclude,
for $n \in \N$, $k \in \K$, $j \in \J$,
\begin{linenomath}
\EQS
|v_{n,k,j}^{\hat{m}+1}| \le  e^{2 (\hat{m} +1) \epsilon \frac{\Delta\tau}{T}}  e^{\R (\hat{m}+1)  \Delta \tau } (\|v^{0}\|_{\infty} +  a_j).
\ENS
\end{linenomath}
This proves \eqref{eq:vminusmax} at time $m = \hat{m}+1$.


To prove \eqref{eq:vminusmin}, similarly with \eqref{eq:vkk},
for $n \in \N$, $k \in \K$, $j \in \J$, we have
\EQA
\bvlsl_{n,k,j}^{{\hat{m}+1}}
&=&
\Delta w \Delta r
   \mysum^{d \in \KD}_{l \in \ND}
    \tilde{g}_{n-l, k-d}~
    \bvlsl_{l,d,j}^{{\hat{m}+}}
\nonumber
    \\
&\geq&
    \Delta w \Delta r
  \bigg[ \mysum^{d \in \KD}_{l \in \ND}
    \max \left(\tilde{g}_{n-l, k-d}, 0 \right)
     \bvlsl_{l,d,j}^{{\hat{m}+}}
  -
  \mysum^{d \in \KD}_{l \in \ND}
    \big| \min \left(\tilde{g}_{n-l, k-d}, 0 \right) \big|
    \big|  \bvlsl_{l,d,j}^{{\hat{m}+}} \big| \bigg]
\nonumber
  \\
 &\overset{\text{(i)}}{\geq} &  \frac{\Delta w \Delta r}{1 + \Delta \tau r_d} \mysum^{d \in \KD}_{l \in \ND}
    \tilde{g}_{n-l, k-d}~ \left[ - 2\epsilon \hat{m} \frac{\Delta\tau}{T} e^{2\epsilon \hat{m} \frac{\Delta\tau}{T}}  e^{\R \hat{m} \Delta \tau}
\left(
\left\| v^{0} \right\|_{\infty} + a_{j}
\right) \right]
  \\
  && -~  \frac{\Delta w \Delta r}{1 + \Delta \tau r_d} \mysum^{d \in \KD}_{l \in \ND}
   \big| \min\left( \tilde{g}_{n-l, k-d}, 0 \right) \big| \left[
e^{2\epsilon \hat{m} \frac{\Delta\tau}{T}} e^{\R \hat{m} \Delta \tau}
\left(
\left\| v^{0} \right\|_{\infty} +  a_{j}
\right) \right]
\nonumber
   \\
   & \overset{\text{(ii)}}{\geq} & - 2\epsilon (\hat{m}+1) \frac{\Delta\tau}{T}  e^{2\epsilon (\hat{m}+1) \frac{\Delta\tau}{T}}  e^{\R (\hat{m}+1)  \Delta \tau}
\left(
\left\| v^{0} \right\|_{\infty} +  a_{j}
\right),
\label{eq:lower_bound_03}
\ENA
where, in (i), we used \eqref{eq:lower_bound_0l}, \eqref{eq:abs_bound}, and the property of linear interpolation; in (ii), we used \eqref{eq:test1}, \eqref{eq:sum_g} and \eqref{eq:r_k_con}. Thus, by \eqref{eq:lower_bound_03}, we have
\EQAS
v_{n,k,j}^{\hat{m}+1}  ~\geq~  \bvlsl_{n,k,j}^{{\hat{m}+1}} ~\geq~ - 2\epsilon (\hat{m}+1) \frac{\Delta\tau}{T} e^{2\epsilon (\hat{m}+1) \frac{\Delta\tau}{T}}
e^{\R (\hat{m}+1) \Delta \tau} \left(
\left\| v^{0} \right\|_{\infty} + a_j
\right),
\ENAS
which proves \eqref{eq:vminusmin} at $m = \hat{m}+1$.

\end{appendices} 
\end{document}